%% file: Driver_P2_double_porosity.tex
\documentclass[11pt,reqno]{amsproc}
\allowdisplaybreaks
\usepackage{cite}
\usepackage{tcolorbox}
\tcbuselibrary{breakable}
\usepackage{enumerate}
\usepackage{psfrag}
\usepackage{caption}
\usepackage{color}
\usepackage{fullpage}
\usepackage{graphicx}
\usepackage{mathrsfs}
\usepackage{subfigure}
\DeclareGraphicsExtensions{.eps}
\usepackage{ucs}
\usepackage[utf8x]{inputenc}
\usepackage{epstopdf}
\usepackage[semicolon,square,authoryear]{natbib}
\usepackage[debug=false, colorlinks=true, pdfstartview=FitV, 
linkcolor=blue, citecolor=blue, urlcolor=blue]{hyperref}
\numberwithin{equation}{section}
\makeatletter
\def\maketag@@@#1{\hbox{\m@th\normalfont\normalsize#1}}
\makeatother
\newtheorem{lemma}{Lemma}[section]
\newtheorem{theorem}{Theorem}[section]
\newtheorem{remark}{Remark}[section]
\newtheorem{proposition}{Proposition}[section]

\usepackage{morefloats}
\newlength{\drop}
\definecolor{amethyst}{rgb}{0.6, 0.4, 0.8}
\definecolor{burgundy}{rgb}{0.5, 0.0, 0.13}

\title{\textbf{Modeling flow in porous media 
with double porosity/permeability:~A stabilized mixed formulation, error analysis, and numerical 
solutions}}

\author{\textbf{S.~H.~S.~Joodat}, \textbf{K.~B.~Nakshatrala},
  and \textbf{R.~Ballarini} \\
  {\small Department of Civil and Environmental
    Engineering, University of Houston. \\
    \textbf{Correspondence to:}~\textsf{knakshatrala@uh.edu}}}

\keywords{stabilized mixed formulations; error estimates; 
patch tests; double porosity/permeability; flow through 
porous media}

\begin{document}

\date{\today}

\begin{titlepage}
  \drop=0.1\textheight
  \centering
  \vspace*{\baselineskip}
  \rule{\textwidth}{1.6pt}\vspace*{-\baselineskip}\vspace*{2pt}
  \rule{\textwidth}{0.4pt}\\[\baselineskip]
       {\Large \textbf{\color{burgundy}
           Modeling flow in porous media with double porosity/permeability\\
           {\small A stabilized mixed formulation, error analysis, and
             numerical solutions}}}\\[0.3\baselineskip]
       \rule{\textwidth}{0.4pt}\vspace*{-\baselineskip}\vspace{3.2pt}
       \rule{\textwidth}{1.6pt}\\[\baselineskip]
       \scshape
       An e-print of the paper is available on arXiv:~1705/08883.  \par 
       \vspace*{1\baselineskip}
       Authored by \\[\baselineskip]

  {\Large S.~H.~S.~Joodat\par}
  {\itshape Graduate Student, University of Houston}\\[0.75\baselineskip]
           
  {\Large K.~B.~Nakshatrala\par}
  {\itshape Department of Civil \& Environmental Engineering \\
  University of Houston, Houston, Texas 77204--4003 \\ 
  \textbf{phone:} +1-713-743-4418, \textbf{e-mail:} knakshatrala@uh.edu \\
  \textbf{website:} http://www.cive.uh.edu/faculty/nakshatrala}\\[0.75\baselineskip]
    
  {\Large R.~Ballarini\par}
  {\itshape Thomas and Laura Hsu Professor and Chair \\
    Department of Civil and Environmental Engineering,
     University of Houston.}
  
  \vspace{0.1in} 
  
\begin{figure}[h]
  \includegraphics[clip,scale=0.35]{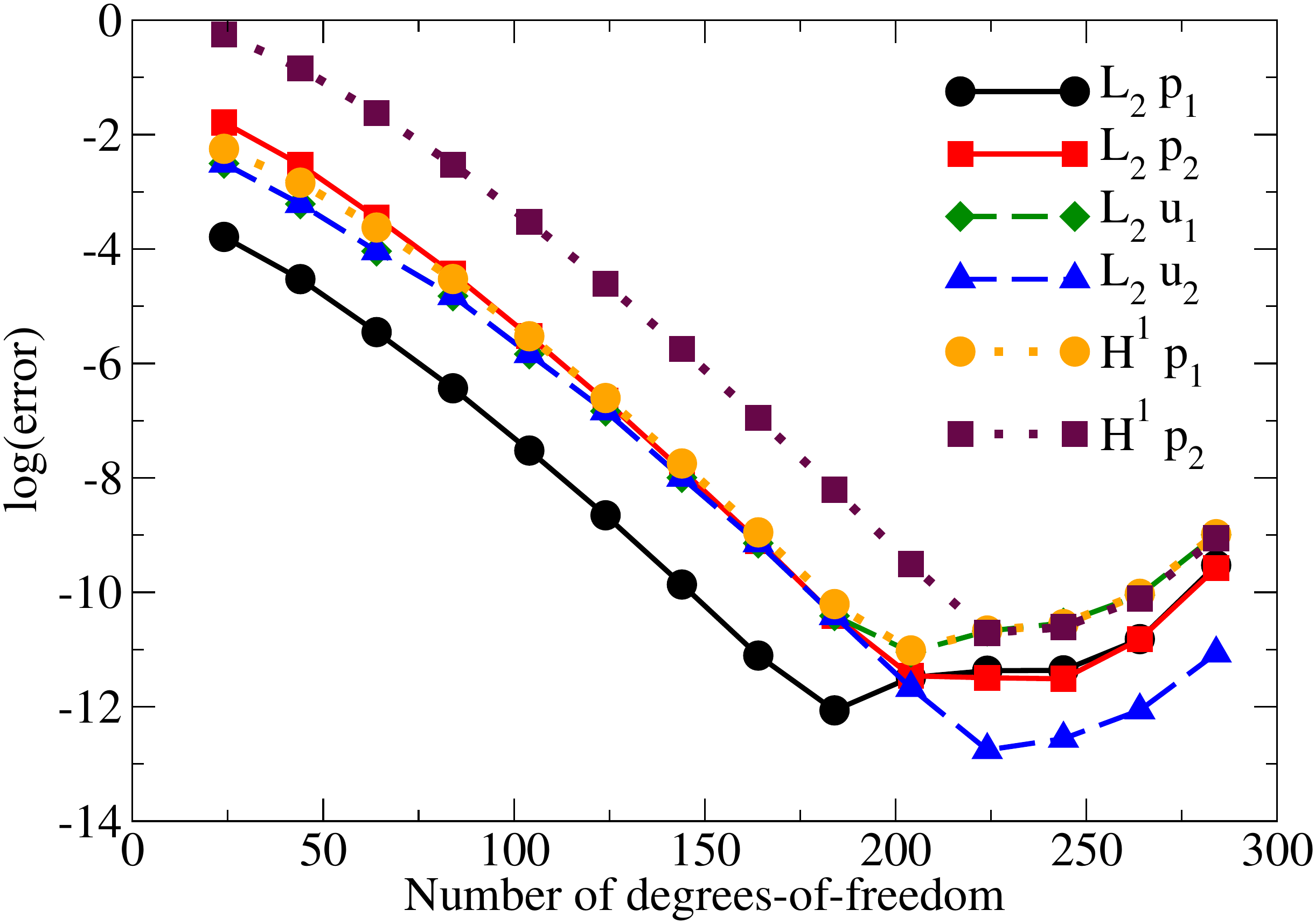}

    \emph{This figure shows that the rates of convergence 
    under the proposed stabilized mixed formulation is 
    exponential with respect to $p$-refinement, which 
    is in accordance with the theory.}
\end{figure}
  \vfill
  {\scshape 2017} \\
  {\small Computational \& Applied Mechanics Laboratory} \par
\end{titlepage}

\begin{abstract}
  The flow of incompressible fluids through 
  porous media plays a crucial role in many 
  technological applications such as enhanced 
  oil recovery and geological carbon-dioxide 
  sequestration. 
  The flow within numerous natural and synthetic 
  porous materials that contain multiple scales of 
  pores cannot be adequately described by the 
  classical Darcy equations. 
  It is for this reason that mathematical models for fluid 
  flow in media with multiple scales of pores have been 
  proposed in the literature. However, these models are 
  analytically intractable for realistic problems. 
  In this paper, a stabilized mixed four-field finite 
  element formulation is presented to study the flow 
  of an incompressible fluid in porous media exhibiting
  double porosity/permeability. The stabilization terms 
  and the stabilization parameters are derived in a 
  mathematically consistent 
  manner, and the computationally convenient equal-order 
  interpolation of all the field variables is shown to 
  be stable. A systematic error analysis is performed on 
  the resulting stabilized weak formulation. Representative 
  problems, patch tests and numerical convergence analyses 
  are performed to illustrate the performance and convergence 
  behavior of the proposed mixed formulation in the discrete 
  setting.
  The accuracy of numerical solutions is assessed using the 
  mathematical properties satisfied by the solutions of this 
  double porosity/permeability model. Moreover, it is shown 
  that the proposed framework can perform well under transient 
  conditions and that it can capture well-known instabilities 
  such as viscous fingering.
  \end{abstract}

\maketitle

\vspace{-0.4in}

\input{Sections/S1_VMS_Introduction.tex}

\input{Sections/S2_VMS_GE.tex}

\input{Sections/S3_VMS_Mixed_Weak}

\input{Sections/S4_VMS_Theoretical_convergence}

\input{Sections/S5_VMS_Numerical_convergence}

\input{Sections/S6_VMS_NR}

\input{Sections/S7_VMS_Verification}

\input{Sections/S8_VMS_Transient}

\input{Sections/S9_VMS_Coupled}

\input{Sections/S10_VMS_CR}

\appendix
\input{Sections/VMS_Appendix}

\bibliographystyle{plainnat}
\bibliography{Master_References/Books,Master_References/Master_References}
\clearpage
\newpage
\input{Sections/VMS_Figures}
\end{document}

%% file: Sections/S1_VMS_Introduction.tex
\section{INTRODUCTION}
Fluid flow in porous media has been extensively
studied, both theoretically and computationally, because of 
its broad applications in different branches of science
and engineering. The most popular model of
flow of an incompressible fluid in rigid porous media
is the Darcy model, which is based on the assumption
that the domain contains only one pore-network.
Due to the restricting assumptions in the classical
Darcy model \citep{rajagopal2007hierarchy,Nakshatrala_Rajagopal_2011, Chang_Nakshatrala_2017}, its application has been
limited and several modifications and
alternative models have been proposed that predict more
realistic flow behaviors.
In particular, due to the complexity of the pore-structure
in many geo-materials such as shale, many studies have
focused on developing mathematical models and computational
frameworks that consider the presence of two (or more)
dominant pore-networks exhibiting different hydro-mechanical
properties. Some of the recent studies on multiple pore-networks
include \citep{Borja_Koliji_2009,Choo_White_Borja_2015_IJG}.

The mathematical models pertaining to the flow in porous
media with multiple pore-networks are complex and involve
numerous field variables. It is not always possible to
derive analytical solutions to these mathematical models,
and one has to resort to numerical solutions for realistic 
problems. 
Different approaches are available for developing
formulations for multi-field mathematical models.
Mixed finite element formulations, which offer
the flexibility of using different approximations
for different field variables, are particularly
attractive for multi-field problems. Accurate
numerical solutions have been obtained using
mixed finite element for various porous media
models; for example, see 
\citep{Masud_Hughes_2002,badia2010stabilized,
  Nakshatrala_Turner_Hjelmstad_Masud_CMAME_2006_v195_p4036,
  Nakshatrala_Rajagopal_2011,Choo_Borja_2015stabilized}.
Moreover, many of the mathematical models pertaining
to the multiple pore-networks, and in particular,
the mathematical model considered in this paper,
cannot be written in terms of a single-field variable.
Although mixed methods are considered a powerful tool,
especially for modeling flow problems in porous media,
they suffer from some restrictions. To obtain stable 
and convergent solutions, a mixed formulation should 
satisfy the Ladyzhenskaya-Babu\v ska-Brezzi
(LBB) stability condition \citep{Babuvska1973finite,
Brezzi_Fortin}. Numerical instability of the solution 
and probable spurious oscillations in the profile of 
unknown variables are the main consequences of the
violation of this condition. Such drawbacks are observed 
in many of the existing formulations and highlight the 
need for developing more robust computational frameworks.
%
In order to resolve numerical instabilities resulting from violation of
the LBB condition, computational approaches are divided broadly into two classes \citep{Franca_Hughes_CMAME_1988_v69_p89}: those that satisfy the LBB condition and those that circumvent it. 

In the former approach, elements are developed by placing restrictions on the interpolation spaces so as to satisfy the LBB condition under the classical mixed (Galerkin) formulation. Such elements are collectively referred to as the H(div) elements \citep{Brezzi_Fortin,brezzi2008mixed}. Two popular works of this type are Raviart-Thomas (RT) spaces \citep{Raviart_Thomas_MAFEM_1977_p292}, and Brezzi-Douglas-Marini (BDM) spaces \citep{Brezzi_Douglas_Marini_NumerMath_1985_v47_p217,Brezzi_Douglas_Durran_Marini_NumerMath_1987_v51_p237}. 
The class of stabilized methods, which falls under the latter approach,
is an attractive way of circumventing the LBB condition. In a stabilized
formulation, stabilization terms are augmented to the classical mixed
formulation to avoid a saddle-point problem as well as mathematical
instabilities \citep{hughes2004multiscale}.
Various stabilized formulations have
been published for flow problems (e.g., see
\citep{badia2009unified,brooks1982streamline,hughes2000large,
  Turner_Nakshatrala_Hjelmstad_IJNMF_2009_v60_p1291}) and
for flow problems in porous media, in particular, (e.g.,
see \citep{Masud_Hughes_2002,badia2010stabilized,
  Nakshatrala_Turner_Hjelmstad_Masud_CMAME_2006_v195_p4036,Choo_Borja_2015stabilized}). 

Herein, we develop a stabilized mixed formulation of the double porosity/permeability model
proposed by \citep{Nakshatrala_Joodat_Ballarini}.
The stabilization terms and the
stabilization parameter are derived in a mathematically consistent
manner by appealing to the variational multiscale formalism
\citep{Hughes_1995}.
It is noteworthy that the nodal-based equal-order
interpolation for all the field variables is stable
under the proposed stabilized mixed formulation.
Such a feature for interpolations is particularly
desirable for studies in porous media for two
reasons. The obvious reason is that the equal-order
interpolation is computationally the most convenient.
The second reason is that, in many porous media
applications, the flow and transport equations
are coupled (Section \ref{Sec:S8_2_VMS_Coupled}
of this paper deals with such a coupled problem).
But many existing formulations (including the stabilized
formulations) produce non-physical negative solutions
for the transport equations (i.e., a negative value
for concentration fields), especially when the
diffusion/dispersion is anisotropic
\citep{nagarajan2011enforcing}. The known robust
non-negative finite element based formulations
for the transport equations are nodal-based (e.g.,
refer to \citep{nagarajan2011enforcing,mudunuru2016enforcing}).
By choosing nodal-based unknowns even for the flow
problem, one can avoid projections from nodal to
non-nodal interpolation spaces and vice-versa.

To determine whether a computational framework is robust, systematic convergence and error
analyses are required.
To this end, we first perform a mathematically
rigorous stability analysis of the proposed
stabilized mixed formulation. Since the proposed
formulation is residual-based, consistency is 
shown quite easily. We also present patch tests
and representative numerical results to
show that the obtained numerical results are stable.
After establishing the stability of the proposed
formulation, we perform a thorough accuracy assessment
of the approximations by estimating the error
associated with the numerical solutions.
Specifically, we perform both \emph{a priori}
and \textit{a posteriori} error estimations,
which individually serve different purposes
\citep{babuska2010finite}.
\textit{A posteriori} error estimations monitor
different forms of the error in the numerical
solution \citep{Becker_Rannacher_2001,babuvska2001finite}
and using the computed approximate solution, they provide
an estimate of the form $ \|u − u_h\| \leq \epsilon$,
where $u$ is the solution, $u_h$ is the finite element
solution for a mesh with mesh size $h$, $\|\cdot\|$
denotes an appropriate norm, and $\epsilon$ is a
constant (real) number.
On the other hand, \textit{a priori} error
estimations provide us with the order of
convergence of a given finite element
method \citep{Ainsworth_Oden_1997}.

\citep{2016_Shabouei_Nakshatrala_CiCP} have shown that
porous media models such as those defined by the Darcy 
and Darcy-Brinkman equations satisfy certain mechanics-based 
properties, and they have utilized these properties to construct
solution verification procedures.
Recently, \citep{Nakshatrala_Joodat_Ballarini}
have shown that the double porosity/permeability model 
also enjoys properties with strong mechanics underpinning. 
These include
the minimum dissipation theorem and a reciprocal
relation. Herein, we utilize these mechanics-based
properties to construct \emph{a posteriori} solution
verification procedures to assess the accuracy of
numerical solutions obtained under the proposed
formulation for the double porosity/permeability
model.

Another type of numerical instability, known as Gibbs phenomenon, can also be observed in the numerical solutions of problems associated with flow through porous media with disparate properties. In layered porous domains, conventional continuous finite element methods are not capable of capturing abrupt changes in material properties and result in overshoots and undershoots in the profiles of numerical solutions along the interface of layers where there are jump discontinuities.
In order to eliminate such erroneous oscillations, one possible approach is discontinuous Galerkin (DG) methods. DG methods have been successfully employed by \citep{Hughes_Masud_Wan_2006} for the case of Darcy equations. An extension of the proposed framework using discontinuous Galerkin method for double porosity/permeability model can be obtained using a method similar to the one proposed by \citep{Hughes_Masud_Wan_2006}. However, obtaining such an extension and comparison between the performance of continuous and discontinuous formulations for capturing abrupt changes in material properties are beyond the scope of this paper and will be addressed in a subsequent one. 

A common assumption in models of flow
in porous media is that of steady-state
conditions. However, many flows occurring in
porous media such as aquifers and oil-bearing
strata are transient or unsteady in nature. In this
paper, we extend the proposed stabilized mixed formulation
for the double porosity/permeability mathematical model to
the transient case, and we illustrate this extension can
accurately capture the transient flow characteristics.

Recently, it has been shown that some stabilized
methods (which are primarily designed to suppress
numerical instabilities) when applied to solve
problems with physical instabilities, suppress
both types of instabilities
\citep{Shabouei_Nakshatrala_VF,Shabouei_PhDThesis_UH_2016}.
Therefore, a good test of the
proposed stabilized mixed formulation for
a coupled flow and transport problem involves a problem that exhibits a physical instability similar to the
classical Saffman-Taylor instability
\citep{Saffman_1958}. Using numerical
simulations we show that the proposed
formulation suppresses only the spurious
numerical instabilities while capturing 
the underlying physical instability. 

The rest of this paper is organized as follows.
After an outline of the governing equations of
the double porosity/permeability model in Section
\ref{Sec:S2_VMS_GE}, the corresponding stabilized
mixed formulation is presented in Section
\ref{Sec:S3_VMS_Mixed} with a derivation
provided in Appendix \ref{Sec:VMS_Appendix_A}.
The theoretical convergence analysis for the proposed stabilized mixed formulation is presented in Section \ref{Sec:S4_VMS_Theoretical}, followed by the numerical convergence behavior of the elements presented in Section \ref{Sec:S5_VMS_Canonical} where patch tests in one- and three-dimensional spaces are described. The representative numerical results are used to showcase the performance of the proposed mixed formulation in Section \ref{Sec:S6_VMS_NR}. Section \ref{Sec:S7_VMS_Verification} provides the mechanics-based assessment of the numerical accuracy. The transient analysis and the capability of the computational framework for modeling coupled problems and capturing well-known physical instabilities in fluid mechanics are discussed in Sections \ref{Sec:S8_VMS_Transient} and \ref{Sec:S8_2_VMS_Coupled}.
Finally, conclusions are drawn in Section \ref{Sec:S10_VMS_CR}.

Throughout this paper, repeated indices do not imply
summation. The terms \emph{classical mixed formulation}
and \emph{Galerkin formulation} are used interchangeably.

%% file: Sections/S2_VMS_GE.tex
\section{GOVERNING EQUATIONS FOR DOUBLE POROSITY/PERMEABILITY}
\label{Sec:S2_VMS_GE}
For convenience to the reader and for future
referencing, we document the equations that
govern the double porosity/permeability
mathematical model considered in
\citep{Nakshatrala_Joodat_Ballarini}.
Let $\Omega \subset \mathbb{R}^{nd}$ be a bounded domain,
where ``$nd$'' denotes the number of spatial dimensions.
The boundary of the domain $\partial \Omega$ is assumed
to be piecewise smooth. Mathematically, $\partial \Omega
\equiv \overline{\Omega} - \Omega$, where the superposed
bar denotes the set closure \citep{Evans_PDE}.
A spatial point is denoted by $\mathbf{x} \in \overline{\Omega}$. 
The gradient and divergence operators with respect
to $\mathbf{x}$ are denoted by $\mathrm{grad}[\cdot]$ 
and $\mathrm{div}[\cdot]$, respectively. The unit
outward normal to the boundary is denoted by
$\widehat{\mathbf{n}}(\mathbf{x})$.

The porous domain is assumed to consist of two
dominant pore-networks, which will be referred
to as the macro-pore and micro-pore networks
and are, respectively, denoted by subscripts
$1$ and $2$. These pore-networks are connected
with the possibility of mass exchange between
them.
The pressure field and the discharge (or Darcy) velocity
in the macro-pore network are, respectively, denoted by
$p_{1}(\mathbf{x})$ and $\mathbf{u}_{1}(\mathbf{x})$,
and the corresponding ones in the micro-pore network
are denoted by $p_{2}(\mathbf{x})$ and $\mathbf{u}_{2}
(\mathbf{x})$.
The governing equations under the double porosity/permeability
model take the following form: 
\begin{subequations}
  \begin{alignat}{2}
    \label{Eqn:Dual_GE_Darcy_BLM_1}
    &\mu \mathbf{K}_{1}^{-1} \mathbf{u}_1(\mathbf{x})
    + \mathrm{grad}[p_1] = \gamma \mathbf{b}(\mathbf{x})
    &&\quad \mathrm{in} \; \Omega \\
    \label{Eqn:Dual_GE_Darcy_BLM_2}
    &\mu \mathbf{K}_{2}^{-1} \mathbf{u}_2(\mathbf{x})
    + \mathrm{grad}[p_2] = \gamma \mathbf{b}(\mathbf{x})
    &&\quad \mathrm{in} \; \Omega \\
    \label{Eqn:Dual_GE_Darcy_mass_balance_1}
    &\mathrm{div}[\mathbf{u}_1] = +\chi(\mathbf{x})
    &&\quad \mathrm{in} \; \Omega \\
    \label{Eqn:Dual_GE_Darcy_mass_balance_2}
    &\mathrm{div}[\mathbf{u}_2] = -\chi(\mathbf{x})
    &&\quad \mathrm{in} \; \Omega \\
    \label{Eqn:Dual_GE_vBC_1}
    &\mathbf{u}_1(\mathbf{x}) \cdot
    \widehat{\mathbf{n}}(\mathbf{x})
    = u_{n1}(\mathbf{x})
    &&\quad \mathrm{on} \; \Gamma^{u}_{1} \\
    \label{Eqn:Dual_GE_vBC_2}
    &\mathbf{u}_2(\mathbf{x}) \cdot
    \widehat{\mathbf{n}}(\mathbf{x})
    = u_{n2}(\mathbf{x}) 
    &&\quad \mathrm{on} \; \Gamma^{u}_{2} \\
    \label{Eqn:Dual_GE_Darcy_pBC_1}
    &p_1(\mathbf{x}) = p_{01} (\mathbf{x})
    &&\quad \mathrm{on} \; \Gamma^{p}_{1} \\
    \label{Eqn:Dual_GE_Darcy_pBC_2}
    &p_2(\mathbf{x}) = p_{02} (\mathbf{x})
    &&\quad \mathrm{on} \; \Gamma^{p}_{2} 
  \end{alignat}
\end{subequations}
where $\mathbf{b}(\mathbf{x})$ is the specific
body force. The true density and coefficient
of viscosity of the fluid are, respectively,
denoted by $\gamma$ and $\mu$. $\mathbf{K}_{i} 
(\mathbf{x})$ denotes the permeability tensor for macro-pore ($i=1$) and micro-pore ($i=2$) networks. $\Gamma_{i}^{u}$ denotes the part
of the boundary on which the normal component of the velocity is 
prescribed in the macro-pore ($i=1$) and micro-pore ($i=2$) networks. Similarly, $\Gamma_{i}^{p}$ is that part of 
the boundary on which the pressure is prescribed in the macro-pore ($i=1$) and micro-pore ($i=2$) networks. $p_{01}(\mathbf{x})$ and $p_{02}(\mathbf{x})$ denote the prescribed pressures on $\Gamma_{1}^{p}$ and $\Gamma_{2}^{p}$, respectively. $u_{n1}(\mathbf{x})$ and $u_{n2}(\mathbf{x})$ denote the prescribed normal components of the velocities on $\Gamma_{1}^{u}$ and $\Gamma_{2}^{u}$, respectively. 
$\chi(\mathbf{x})$ is the rate of volume exchange
of the fluid between the two pore-networks per
unit volume of the porous medium, and we model
it as follows
\citep{Barenblatt_Zheltov_Kochina_v24_p1286_1960_ZAMM}: 
\begin{align}
\label{Eqn:Dual_GE_mass_transfer}
  \chi(\mathbf{x}) = - \frac{\beta}{\mu}
  (p_1(\mathbf{x}) - p_2(\mathbf{x}))
\end{align}
where $\beta$ is a dimensionless characteristic
of the porous medium. In the rest of the paper, 
as is commonly done in the literature, $\chi(
\mathbf{x})$ will be simply referred to as the
mass transfer. 
For mathematical well-posedness, we assume that
\begin{align}
  \Gamma_{1}^{u} \cup \Gamma_{1}^{p} = \partial \Omega, \quad 
  \Gamma_{1}^{u} \cap \Gamma_{1}^{p} = \emptyset, \quad 
  \Gamma_{2}^{u} \cup \Gamma_{2}^{p} = \partial \Omega
  \quad \mathrm{and} \quad 
  \Gamma_{2}^{u} \cap \Gamma_{2}^{p} = \emptyset
\end{align}

%% file: Sections/S3_VMS_Mixed_Weak.tex
\section{A STABILIZED MIXED WEAK FORMULATION}
\label{Sec:S3_VMS_Mixed}
In this section, we present the proposed
stabilized mixed formulation for the double
porosity/permeability model. A derivation
of the proposed formulation is provided in
Appendix \ref{Sec:VMS_Appendix_A}.
The proposed formulation
  is built upon the stabilization ideas
  put forth in a pioneering paper by
  \citep{Masud_Hughes_2002}. The proposed
  formulation for double porosity/permeability
  model can be obtained by adding a stabilization
  term, similar to the 
  one proposed by \citep{Masud_Hughes_2002}
  for the case of single-pore network Darcy equations,
  to each pore-network.
  The stabilization terms are based on the
  residual of the balance of linear momentum
  in each pore-network. The stability can be achieved without
  adding residual-based stabilization terms related
  to the mass balance equations for any of the
  pore-networks.
We also present an extension of the proposed formulation
for enforcing the velocity boundary conditions weakly,
which will be convenient for problems involving curved
boundaries. This extension is achieved by employing
a procedure similar to the one proposed by 
\citep{Nitsche_1971}.

We define the relevant function spaces, which
will be used in the rest of this paper. We
denote the set of all square-integrable
functions on $\Omega$ by $L_{2}(\Omega)$.
For mathematical well-posedness, we assume
that
\begin{align}
  \gamma \mathbf{b}(\mathbf{x}) \in (L_2(\Omega))^{nd}, \;
  p_{01}(\mathbf{x}) \in H^{1/2}(\Gamma_{1}^{p})
  \quad \mathrm{and} \quad  
  p_{02}(\mathbf{x}) \in H^{1/2}(\Gamma_{2}^{p}) 
\end{align}
where $H^{1/2}(\cdot)$ is a non-integer
Sobolev space \citep{Adams_Sobolev}.
The function spaces for the velocity
and pressures fields are defined as
follows:
\begin{subequations}
  \begin{align}
    \label{Eqn:VMS_Function_space_U1}
    \mathcal{U}_{1} &:= 
    \left\{\mathbf{u}_{1}(\mathbf{x}) \in 
    \left(L_{2}(\Omega)\right)^{nd} 
    \; \Big\vert \;
    \mathrm{div}[\mathbf{u}_{1}] \in L_{2}(\Omega), 
    \mathbf{u}_{1}(\mathbf{x}) \cdot \widehat{\mathbf{n}}(\mathbf{x}) 
    = u_{n1}(\mathbf{x}) \in H^{-1/2}(\Gamma_{1}^{u})\right\} \\
    \mathcal{U}_{2} &:= 
    \left\{\mathbf{u}_{2}(\mathbf{x}) \in 
    \left(L_{2}(\Omega)\right)^{nd} 
    \; \Big\vert \;
    \mathrm{div}[\mathbf{u}_{2}] \in L_{2}(\Omega), 
    \mathbf{u}_{2}(\mathbf{x}) \cdot \widehat{\mathbf{n}}(\mathbf{x}) 
    = u_{n2}(\mathbf{x}) \in H^{-1/2}(\Gamma_{2}^{u})\right\} \\
    \mathcal{W}_{1} &:= 
    \left\{\mathbf{w}_{1}(\mathbf{x}) \in 
    \left(L_{2}(\Omega)\right)^{nd} 
    \; \Big\vert \;
    \mathrm{div}[\mathbf{w}_{1}] \in L_{2}(\Omega), 
    \mathbf{w}_{1}(\mathbf{x}) \cdot \widehat{\mathbf{n}}(\mathbf{x}) 
    = 0 \; \mathrm{on} \; \Gamma_{1}^{u} \right\} \\
    \mathcal{W}_{2} &:= 
    \left\{\mathbf{w}_{2}(\mathbf{x}) \in 
    \left(L_{2}(\Omega)\right)^{nd} 
    \; \Big\vert \;
    \mathrm{div}[\mathbf{w}_{2}] \in L_{2}(\Omega), 
    \mathbf{w}_{2}(\mathbf{x}) \cdot \widehat{\mathbf{n}}(\mathbf{x}) 
    = 0 \; \mathrm{on} \; \Gamma_{2}^{u}\right\} \\
    \mathcal{P} &:= 
    \left\{(p_1(\mathbf{x}),p_{2}(\mathbf{x})) 
    \in L_{2}(\Omega) \times L_{2}(\Omega)
    \; \Big\vert \;
    \left(\int_{\Omega} p_{1}(\mathbf{x}) \mathrm{d} \Omega \right) 
    \left(\int_{\Omega} p_{2}(\mathbf{x}) \mathrm{d} \Omega \right) 
    = 0 \right\} \\
    \label{Eqn:VMS_Function_space_Q}
    \mathcal{Q} &:= 
    \left\{(p_1(\mathbf{x}),p_{2}(\mathbf{x})) 
    \in H^{1}(\Omega) \times H^{1}(\Omega)
    \; \Big\vert \;
    \left(\int_{\Omega} p_{1}(\mathbf{x}) \mathrm{d} \Omega \right) 
    \left(\int_{\Omega} p_{2}(\mathbf{x}) \mathrm{d} \Omega \right) 
    = 0 \right\}
  \end{align}
\end{subequations}
where $H^{1}(\Omega)$ is a standard Sobolev space,
and $H^{-1/2}(\cdot)$ is the dual space corresponding
to $H^{1/2}(\cdot)$ \citep{Adams_Sobolev}. The standard
$L_2$ inner-product over a set $K$ is denoted as
\begin{align}
  (\mathbf{a};\mathbf{b})_{K} \equiv \int_{K}
  \mathbf{a}(\mathbf{x})\cdot \mathbf{b}
  (\mathbf{x}) \; \mathrm{d} K 
\end{align}
For convenience, the subscript $K$ will be
dropped if $K = \Omega$. Moreover, the action of a
linear functional on a vector from its
associated vector space is denoted by
$\langle \cdot; \cdot\rangle$. 

A few remarks are needed regarding the
following condition on the pressures in
the function spaces $\mathcal{P}$ and
$\mathcal{Q}$:
\begin{align*}
  \left(\int_{\Omega} p_1(\mathbf{x})
  \mathrm{d} \Omega \right)
  \left(\int_{\Omega} p_2(\mathbf{x})
  \mathrm{d} \Omega \right) = 0 
\end{align*}
This condition of vanishing mean pressure
in one of pore-networks is a mathematically
elegant way of fixing the datum for the
pressure. Without fixing the datum for
the pressure (which will be the case
when only the velocity boundary conditions
are prescribed on the entire boundary), one
can find the pressures only up to an arbitrary
constant, which will be the case even under Darcy
equations
\citep{Nakshatrala_Turner_Hjelmstad_Masud_CMAME_2006_v195_p4036}.
Herein, we introduced the vanishing mean
pressure condition into the function spaces
to ensure uniqueness of the solutions, which
will be established later in this paper.
However, it should be emphasized that
vanishing mean pressure in one of the
pore-networks is \emph{not} necessary
for all the problems under the double
porosity/permeability model. 
One can fix the datum for the pressure
under the double porosity/permeability
model by prescribing the pressure in
at least one of the pore-networks on
a portion of the boundary, which is a
set of non-zero measure. To put it
differently, for problems with
pressure boundary conditions, the
datum for the pressure is automatically
fixed through the prescribed boundary
condition, and hence, for those problems,
one does not include the zero mean
pressure condition in the function
spaces $\mathcal{P}$ and $\mathcal{Q}$. 
For example, see the problem in
subsection \ref{Sub:One-dimensional_patch_test},
which deals with prescribed pressure boundary
conditions.

The \emph{classical mixed formulation}, which is
based on the Galerkin formalism, reads as follows:~Find
$\left(\mathbf{u}_1(\mathbf{x}),\mathbf{u}_2(\mathbf{x})
\right)\in \mathcal{U}_1 \times \mathcal{U}_2$
and $\left(p_1(\mathbf{x}), p_2(\mathbf{x})\right)
\in \mathcal{P}$ such that we have
\begin{align}
\label{Eqn:VMS_classical_mixed_formulation}
  \mathcal{B}_{\mathrm{Gal}}(\mathbf{w}_1,\mathbf{w}_2,q_1,q_2;
  \mathbf{u}_1,\mathbf{u}_2,p_1,p_2) = \mathcal{L}_{\mathrm{Gal}}
  (\mathbf{w}_1,\mathbf{w}_2,q_1,q_2) \nonumber \\
  \quad \forall \left(\mathbf{w}_1(\mathbf{x}),
  \mathbf{w}_{2}(\mathbf{x})\right) \in
  \mathcal{W}_1 \times \mathcal{W}_2,~
  \left(q_1(\mathbf{x}),~q_2(\mathbf{x})\right) \in \mathcal{P} 
\end{align}
where the bilinear form and the linear functional 
are, respectively, defined as follows:
\begin{alignat}{2}
  \label{Eqn:VMS_bilinear_form}
  \mathcal{B}_{\mathrm{Gal}}(\mathbf{w}_1,\mathbf{w}_2,q_1,q_2;
  \mathbf{u}_1,\mathbf{u}_2,p_1,p_2) &:= (\mathbf{w}_1;\mu \mathbf{K}_{1}^{-1}\mathbf{u}_1)
  - (\mathrm{div}[\mathbf{w}_1];p_1)
  + (q_1;\mathrm{div}[\mathbf{u}_1]) \nonumber \\
  &~+ (\mathbf{w}_2;\mu \mathbf{K}_{2}^{-1}\mathbf{u}_2)
  - (\mathrm{div}[\mathbf{w}_2];p_2)
  + (q_2;\mathrm{div}[\mathbf{u}_2]) \nonumber \\
  &~ + (q_1 - q_2;\beta/\mu(p_1 - p_2)) \\
  \label{Eqn:VMS_linear_form}
  \mathcal{L}_{\mathrm{Gal}}(\mathbf{w}_1,\mathbf{w}_2,q_1,q_2) := (\mathbf{w}_1;\gamma \mathbf{b})
  &~+ (\mathbf{w}_2;\gamma \mathbf{b})
  - \langle\mathbf{w}_1 \cdot \widehat{\mathbf{n}};p_{01}
    \rangle_{\Gamma^{\mathrm{p}}_{1}} 
    - \langle\mathbf{w}_2 \cdot \widehat{\mathbf{n}};p_{02}
    \rangle_{\Gamma^{\mathrm{p}}_{2}} 
\end{alignat}

In a subsequent section, we will show that the
equal-order interpolation for all the variables,
which is computationally the most convenient, is
not stable under the classical mixed formulation.
Of course, one could use divergence-free elements
(e.g., Raviart-Thomas spaces
\citep{Raviart_Thomas_MAFEM_1977_p292}) but they need
special data structures and computer implementations.
We, therefore, present a stabilized mixed formulation,
which is stable under the equal-order interpolation
for all the field variables. 

\begin{tcolorbox}[breakable]
The \emph{proposed stabilized mixed formulation}
reads as follows:~Find $\left(\mathbf{u}_1(\mathbf{x}),
\mathbf{u}_2(\mathbf{x}) \right) \in \mathcal{U}_1
\times \mathcal{U}_2$ and $\left(p_1(\mathbf{x}),
p_2(\mathbf{x})\right) \in \mathcal{Q}$ such that we have
\begin{align}
  \mathcal{B}_{\mathrm{stab}}(\mathbf{w}_1,\mathbf{w}_2,q_1,q_2;
  \mathbf{u}_1,\mathbf{u}_2,p_1,p_2)
  = \mathcal{L}_{\mathrm{stab}}(\mathbf{w}_1,\mathbf{w}_2,q_1,q_2) \nonumber \\
  \quad \forall
  \left(\mathbf{w}_1(\mathbf{x}), \mathbf{w}_2(\mathbf{x})\right) ~
  \in \mathcal{W}_1 \times \mathcal{W}_2,
  \left(q_1(\mathbf{x}),q_2(\mathbf{x})\right)
  \in \mathcal{Q} 
\label{Eqn:VMS_Galerkin_Weak_Form} 
\end{align}
where the bilinear form and the linear functional
are, respectively, defined as follows:
\begin{align}
  \label{Eqn:Dual_B_VMS}
  \mathcal{B}_{\mathrm{stab}}(\mathbf{w}_1,\mathbf{w}_2,q_1,q_2;
  \mathbf{u}_1,\mathbf{u}_2,p_1,p_2)
  := \mathcal{B}_{\mathrm{Gal}}(\mathbf{w}_1,\mathbf{w}_2,q_1,q_2;
  \mathbf{u}_1,\mathbf{u}_2,p_1,p_2) \nonumber \\
  -\frac{1}{2} \left(\mu \mathbf{K}_1^{-1}
  \mathbf{w}_1 - \mathrm{grad}[q_1];\frac{1}{\mu}
  \mathbf{K}_1 (\mu \mathbf{K}_1^{-1}
  \mathbf{u}_1 + \mathrm{grad}[p_1])\right)
  \nonumber \\
  -\frac{1}{2} \left(\mu \mathbf{K}_2^{-1} \mathbf{w}_2
  - \mathrm{grad}[q_2]; \frac{1}{\mu} \mathbf{K}_2
  (\mu \mathbf{K}_2^{-1}
  \mathbf{u}_2 + \mathrm{grad}[p_2])\right)
\end{align}
\begin{align}
  \label{Eqn:Dual_L_VMS}
  \mathcal{L}_{\mathrm{stab}}(\mathbf{w}_1,\mathbf{w}_2,q_1,q_2)
  &:= \mathcal{L}_{\mathrm{Gal}}(\mathbf{w}_1,\mathbf{w}_2,q_1,q_2)
  -\frac{1}{2} \left(\mu \mathbf{K}_1^{-1}
  \mathbf{w}_1 - \mathrm{grad}[q_1];
  \frac{1}{\mu} \mathbf{K}_1 \gamma \mathbf{b}\right) \nonumber \\
  &-\frac{1}{2} \left(\mu \mathbf{K}_2^{-1} \mathbf{w}_2
  - \mathrm{grad}[q_2]; \frac{1}{\mu} \mathbf{K}_2
  \gamma \mathbf{b}\right)
\end{align}
\end{tcolorbox}

In subsequent sections, we show that the proposed
stabilized mixed formulation is consistent, stable
and accurate.

\subsection{Weak enforcement of velocity boundary conditions}
In the previous derivations made earlier in this section, the pressure boundary conditions (i.e., equations
\eqref{Eqn:Dual_GE_Darcy_pBC_1} and
\eqref{Eqn:Dual_GE_Darcy_pBC_2}) are
enforced weakly under the proposed stabilized
mixed formulation and the classical mixed
formulation. However, the velocity boundary
conditions in which the normal components
of the velocities are prescribed (i.e.,
equations \eqref{Eqn:Dual_GE_vBC_1} and
\eqref{Eqn:Dual_GE_vBC_2}) are enforced
strongly. For domains with curved
boundaries, which are commonly encountered
in subsurface modeling, it is desirable to even
prescribe the velocity boundary
conditions weakly. We, therefore, provide a possible
extension of the proposed stabilized mixed
formulation for weak enforcement of the velocity
boundary conditions. To this end, we follow the
approach proposed by \citep{Nitsche_1971}. The
Nitsche's method is a powerful tool for weakly
enforcing Dirichlet boundary conditions without
the use of Lagrange multipliers, and has been
utilized by several works such as
\citep{bazilevs2007weak,Embar_Folbow_Harari_2010,
  annavarapu2014nitsche,Schillinger_Harari_2016}.
The Nitsche's method is sometimes referred to
as a variationally consistent penalty method
to enforce Dirichlet boundary conditions
\citep{Hansbo_2005}.
We extend the Nitsche's method to the proposed
four-field stabilized formulation to enforce
the prescribed normal components of the
velocities in the macro- and micro-pore
networks. 

The stabilized mixed formulation that enforces the
velocity boundary conditions weakly can be obtained
as follows:~Find $\left(\mathbf{u}_1(\mathbf{x}),
\mathbf{u}_2(\mathbf{x}) \right) \in H(\mathrm{div},
\Omega) \times H(\mathrm{div},\Omega)$ and
$\left(p_1(\mathbf{x}),p_2(\mathbf{x})\right)
\in \mathcal{Q}$ such that we have
\begin{align}
  \mathcal{B}_{\mathrm{stab}}^{\mathrm{weak~B.C.}}(\mathbf{w}_1,\mathbf{w}_2,q_1,q_2;
  \mathbf{u}_1,\mathbf{u}_2,p_1,p_2)
  = \mathcal{L}_{\mathrm{stab}}^{\mathrm{weak~B.C.}}(\mathbf{w}_1,\mathbf{w}_2,q_1,q_2) \nonumber \\
  \quad \forall
  \left(\mathbf{w}_1(\mathbf{x}), \mathbf{w}_2(\mathbf{x})\right) ~
  \in H(\mathrm{div},\Omega) \times H(\mathrm{div},\Omega),
  \; \left(q_1(\mathbf{x}),q_2(\mathbf{x})\right)
  \in \mathcal{Q} 
\label{Eqn:Nitsche_Weak_Form} 
\end{align}
where the bilinear form and the linear functional are, respectively, defined as follows:
\begin{subequations}
\begin{alignat}{2}
  \label{Eqn:Dual_B_Nitsche}
  \mathcal{B}_{\mathrm{stab}}^{\mathrm{weak~B.C.}}(\mathbf{w}_1,\mathbf{w}_2,q_1,q_2; \mathbf{u}_1,\mathbf{u}_2,p_1,p_2) &:= \mathcal{B}_{\mathrm{stab}}(\mathbf{w}_1,\mathbf{w}_2,q_1,q_2; \mathbf{u}_1,\mathbf{u}_2,p_1,p_2) \nonumber \\
  &~+ (\mathbf{w}_{1}\cdot\widehat{\mathbf{n}};p_{1})_{\Gamma^{u}_{1}}
    + (\mathbf{w}_{2}\cdot\widehat{\mathbf{n}};p_{2})_{\Gamma^{u}_{2}}\nonumber \\
    &~+ (q_{1};\mathbf{u}_{1}\cdot\widehat{\mathbf{n}})_{\Gamma^{u}_{1}}  
    + (q_{2};\mathbf{u}_{2}\cdot\widehat{\mathbf{n}})_{\Gamma^{u}_{2}} \nonumber \\
   &~+ \frac{\eta}{h}(\mathbf{w}_{1}\cdot\widehat{\mathbf{n}};\mathbf{u}_{1}\cdot\widehat{\mathbf{n}})_{\Gamma^{u}_{1}} 
   + \frac{\eta}{h}(\mathbf{w}_{2}\cdot\widehat{\mathbf{n}};\mathbf{u}_{2}\cdot\widehat{\mathbf{n}})_{\Gamma^{u}_{2}}  \\
  \label{Eqn:Dual_L_Nitsche}
  \mathcal{L}_{\mathrm{stab}}^{\mathrm{weak~B.C.}}(\mathbf{w}_1,\mathbf{w}_2,q_1,q_2) := \mathcal{L}_{\mathrm{stab}}(\mathbf{w}_1,&\mathbf{w}_2,q_1,q_2) + (q_{1};u_{n1})_{\Gamma^{u}_{1}}  + (q_{2};u_{n2})_{\Gamma^{u}_{2}} \nonumber \\
  &+ \frac{\eta}{h}(\mathbf{w}_{1}\cdot\widehat{\mathbf{n}};u_{n1})_{\Gamma^{u}_{1}} 
  + \frac{\eta}{h}(\mathbf{w}_{2}\cdot\widehat{\mathbf{n}};u_{n2})_{\Gamma^{u}_{2}} 
\end{alignat} 
\end{subequations}
where $h$ is the mesh size and 
$\eta$ is the penalty parameter. In this
paper, we have taken $h$ to be the maximum
edge length in the mesh, and have taken the
penalty parameter to be 10.
In the above statement of the weak formulation,
since the velocity boundary conditions are
enforced weakly, the appropriate function
space for the velocities and the associated
weighting functions will be $H(\mathrm{div},\Omega)$,
which can be mathematically defined as follows: 
\begin{align}
  H(\mathrm{div},\Omega) := 
\left\{\mathbf{u}(\mathbf{x}) \in 
\left(L_{2}(\Omega)\right)^{nd} 
\; \Big\vert \;
\mathrm{div}[\mathbf{u}] \in L_{2}(\Omega)\right\} 
\end{align}
The function spaces for the pressures and their
weighting functions, however, remain same as
before (i.e., the $\mathcal{Q}$ space).

%% file: Sections/S4_VMS_Theoretical_convergence.tex
\section{A THEORETICAL ANALYSIS OF THE PROPOSED MIXED FORMULATION}
\label{Sec:S4_VMS_Theoretical}
In this section, we present a systematic mathematical
analysis (i.e., existence, uniqueness and well-posedness)
and error analysis (i.e., consistency, stability, order 
of convergence) of the proposed stabilized mixed formulation. 
For convenience, we define the following product spaces:
\begin{align}
    \mathbb{U} = \mathcal{U}_{1} \times \mathcal{U}_{2} 
    \times \mathcal{Q}, \quad \mathrm{and} \quad
    \mathbb{W} = \mathcal{W}_{1} \times \mathcal{W}_{2} 
    \times \mathcal{Q}
\end{align}
We group the field variables as follows:
\begin{subequations}
\begin{align}
\label{Eqn:VMS_Uexact}
\mathbf{U} &= (\mathbf{u}_1(\mathbf{x}),
\mathbf{u}_2(\mathbf{x}),p_1(\mathbf{x}),p_2(\mathbf{x})) \in \mathbb{U} \\
\label{Eqn:VMS_Wexact}
\mathbf{W} &= (\mathbf{w}_1(\mathbf{x}),
\mathbf{w}_2(\mathbf{x}),q_1(\mathbf{x}),q_2(\mathbf{x})) \in \mathbb{W}
\end{align}
\end{subequations}
Then, the proposed mixed formulation in equation
\eqref{Eqn:VMS_Galerkin_Weak_Form} can be compactly
written as:~Find $\mathbf{U} \in \mathbb{U}$ such
that we have
\begin{align}
\label{Eqn:VMS_Galerkin_Weak_Form2} 
\mathcal{B}_{\mathrm{stab}}(\mathbf{W},\mathbf{U}) 
= \mathcal{L}_{\mathrm{stab}}(\mathbf{W}) \quad 
\forall \mathbf{W} \in \mathbb{W}
\end{align} 
We shall establish the stability of the
formulation under the following norm:
\begin{align}
  \|\mathbf{W}\|_{\mathrm{stab}}^{2}
  := \mathcal{B}_{\mathrm{stab}}(\mathbf{W},\mathbf{W}) 
  &= \frac{1}{2}\left\|\sqrt{\mu} \mathbf{K}_1^{-1/2} \mathbf{w}_{1}\right\|^2
  + \frac{1}{2}\left\|\frac{1}{\sqrt{\mu}}\mathbf{K}_1^{1/2} \mathrm{grad}[q_1]\right\|^2
  + \frac{1}{2}\left\|\sqrt{\mu} \mathbf{K}_2^{-1/2} \mathbf{w}_2 \right\|^2 \nonumber \\
  & + \frac{1}{2}\left\|\frac{1}{\sqrt{\mu}}\mathbf{K}_2^{1/2} \mathrm{grad}[q_2]\right\|^2
  + \left\|\sqrt{\left(\frac{\beta}{\mu}\right)} (q_1 - q_2)\right\|^2 \quad 
  \forall \mathbf{W} \in \mathbb{W}
\end{align}
where $\|\cdot\|$ denotes the norm corresponding 
to the standard $L_{2}$ inner-product. We need to first
show that $\|\cdot\|_{\mathrm{stab}}$ is in fact a norm on 
$\mathbb{W}$ and $\mathbb{U}$. To this end, the following
lemma will be used. 

\begin{lemma}{(A property of semi-norms)}
\label{Lemma:DD2_semi-norm}
If $\|\cdot\|_{1}$ and $\|\cdot\|_2$ 
are semi-norms, then $\|\cdot\|_{3} := 
\sqrt{\|\cdot\|_1^{2} + \|\cdot\|^2_{2}}$ 
is also a semi-norm. 
\end{lemma}
\begin{proof}
The homogeneity of $\|\cdot\|_{3}$ directly stems from 
the homogeneity of the semi-norms $\|\cdot\|_{1}$ 
and $\|\cdot\|_{2}$. To wit, 
\begin{align}
\|\alpha \mathbf{x}\|_{3} = \sqrt{\|\alpha \mathbf{x}\|^2_{1} 
+ \|\alpha \mathbf{x}\|^2_{2}} = 
\sqrt{|\alpha|^2 \|\mathbf{x}\|^2_{1} 
+ |\alpha|^2 \|\mathbf{x}\|^2_{2}} 
= |\alpha| \sqrt{\|\mathbf{x}\|^2_{1} + \|\mathbf{x}\|^2_{2}} 
= |\alpha| \|\mathbf{x}\|_{3}
\end{align}
The non-negativity of $\|\cdot\|_{3}$ is
straightforward; that is, $\|\mathbf{x}\|_{3}
\geq 0 \; \forall \mathbf{x}$. 
The triangle inequality for the semi-norms 
$\|\cdot\|_{1}$ and $\|\cdot\|_{2}$ implies
that 
\begin{align}
    \|\mathbf{a}+\mathbf{b}\|_{1} \leq 
    \|\mathbf{a}\|_{1} +  
    \|\mathbf{b}\|_{1}, \quad \mathrm{and} \quad
    \|\mathbf{a}+\mathbf{b}\|_{2} \leq 
    \|\mathbf{a}\|_{2} +  
    \|\mathbf{b}\|_{2}   
\end{align}
These inequalities imply that 
\begin{align}
  \|\mathbf{a}+\mathbf{b}\|^{2}_{3} &= 
  \|\mathbf{a}+\mathbf{b}\|_{1}^{2} + 
  \|\mathbf{a}+\mathbf{b}\|_{2}^{2} \nonumber \\
  &\leq \|\mathbf{a}\|_{1}^{2} +  
  \|\mathbf{a}\|_{2}^{2} 
  + \|\mathbf{b}\|_{1}^{2} + 
  \|\mathbf{b}\|_{2}^{2}
  + 2 \left\{\|\mathbf{a}\|_{1} 
  \|\mathbf{b}\|_{1} 
  + \|\mathbf{a}\|_{2} \|\mathbf{b}\|_{2} 
  \right\} \nonumber \\
  &\leq \left(\sqrt{\|\mathbf{a}\|_{1}^{2}
    + \|\mathbf{a}\|_{2}^{2} }\right)^{2}
  + \left(\sqrt{\|\mathbf{b}\|_{1}^{2}     
    + \|\mathbf{b}\|_{2}^{2}}\right)^{2} 
  + 2 \; \sqrt{\|\mathbf{a}\|_{1}^{2}     
    + \|\mathbf{a}\|_{2}^{2}} \; 
  \sqrt{\|\mathbf{b}\|_{1}^{2}     
    + \|\mathbf{b}\|_{2}^{2}}
  \label{Eqn:DD2_semi-norm}
\end{align}
We have employed the AM-GM inequality in 
obtaining equation \eqref{Eqn:DD2_semi-norm}, which 
further implies that 
\begin{align}
  \|\mathbf{a}+\mathbf{b}\|_{3}
  \leq \sqrt{\|\mathbf{a}\|_{1}^{2} 
    + \|\mathbf{a}\|_{2}^{2}} 
  + \sqrt{\|\mathbf{b}\|_{1}^{2} 
    + \|\mathbf{b}\|_{2}^{2}} 
  = \|\mathbf{a}\|_{3} + \|\mathbf{b}\|_{3}
\end{align}
This establishes the triangle inequality 
for $\|\cdot\|_{3}$. The homogeneity, 
non-negativity and triangle inequality 
imply that $\|\cdot\|_{3}$ is a semi-norm.   
\end{proof}

\begin{proposition}{(Stability norm)}
\label{Prop:DD2_stab_norm}
$\|\cdot\|_{\mathrm{stab}}$ is a norm 
on $\mathbb{W}$ and $\mathbb{U}$. 
\end{proposition}
\begin{proof}
We first note that $\mathbf{K}_1$ and 
$\mathbf{K}_2$ are symmetric and positive 
definite tensors. The square root of a 
symmetric and positive definite tensor 
exists, and is itself a symmetric and 
positive definite tensor \citep{Gurtin1982}. 
This implies that the following individual 
terms form semi-norms on $\mathbb{W}$ and 
$\mathbb{U}$:

\begin{align}
\frac{1}{\sqrt{2}}\left\|\sqrt{\mu} \mathbf{K}_1^{-1/2} \mathbf{w}_{1}\right\|, \; 
\frac{1}{\sqrt{2}}\left\|\frac{1}{\sqrt{\mu}}\mathbf{K}_1^{1/2} \mathrm{grad}[q_1]\right\|, \; 
\frac{1}{\sqrt{2}}\left\|\sqrt{\mu}\mathbf{K}_2^{-1/2} \mathbf{w}_2 \right\|, \nonumber \\
\frac{1}{\sqrt{2}}\left\|\frac{1}{\sqrt{\mu}}\mathbf{K}_2^{1/2} \mathrm{grad}[q_2]\right\|,
\quad \mathrm{and} \quad 
\left\|\sqrt{\left(\frac{\beta}{\mu}\right)} (q_1 - q_2)\right\|
\end{align}
Then, Lemma \ref{Lemma:DD2_semi-norm} implies that 
$\|\cdot\|_{\mathrm{stab}}$ is a semi-norm. It is easy to 
show that $\|\mathbf{W}\|_{\mathrm{stab}} = 0$ implies 
that 
\begin{align}
\mathbf{w}_1(\mathbf{x}) = \mathbf{0}, \; 
\mathbf{w}_2(\mathbf{x}) = \mathbf{0} 
\; \mathrm{and} \; 
q_1(\mathbf{x}) = q_2(\mathbf{x}) = c
\end{align}
where $c$ is a constant. Noting that
$(q_1(\mathbf{x}),q_2(\mathbf{x})) \in
\mathcal{Q}$ and utilizing the following
condition in the definition of  $\mathcal{Q}$:
\begin{align}
  \left(\int_{\Omega} q_1(\mathbf{x}) \mathrm{d} \Omega \right) 
  \left(\int_{\Omega} q_2(\mathbf{x}) \mathrm{d} \Omega \right) = 0 
\end{align}
we conclude that $c = 0$. With this, we have established that 
$\|\mathbf{W}\|_{\mathrm{stab}}$ = 0 implies that $\mathbf{W} = 
\mathbf{0}$.  Hence, $\|\cdot\|_{\mathrm{stab}}$ is a norm. 
\end{proof}

\begin{theorem}{(Uniqueness of weak solutions)}
The weak solution under the proposed mixed formulation is unique. 
\end{theorem}
\begin{proof}
On the contrary, assume that $\mathbf{U}_1$ and 
$\mathbf{U}_{2}$ are both (weak) solutions of 
the weak formulation. This implies that 
\begin{align}
\mathcal{B}_{\mathrm{stab}}(\mathbf{W},\mathbf{U}_1) 
= \mathcal{L}_{\mathrm{stab}}(\mathbf{W})\quad 
\forall \mathbf{W} \in \mathbb{W}, \quad \mathrm{and} \quad
\mathcal{B}_{\mathrm{stab}}(\mathbf{W},\mathbf{U}_2) 
= \mathcal{L}_{\mathrm{stab}}(\mathbf{W})\quad 
\forall \mathbf{W} \in \mathbb{W} 
\end{align}
By subtracting the above two equations and 
noting the linearity in the second slot, 
we obtain 
\begin{align}
\mathcal{B}_{\mathrm{stab}}(\mathbf{W},\mathbf{U}_1 - \mathbf{U}_2) 
&= 0 \quad \forall \mathbf{W} \in \mathbb{W}  
\end{align}
Since $\mathbf{U}_1 - \mathbf{U}_2 \in 
\mathbb{W}$, we can choose $\mathbf{W} = 
\mathbf{U}_1 - \mathbf{U}_2$. This particular 
choice implies that 
\begin{align}
\mathcal{B}_{\mathrm{stab}}(\mathbf{U}_1 - \mathbf{U}_2,
\mathbf{U}_1 - \mathbf{U}_2) 
= \|\mathbf{U}_1 - \mathbf{U}_2\|_{\mathrm{stab}}^2 = 0 
\end{align}
Using Proposition \ref{Prop:DD2_stab_norm} (which establishes that $\|\cdot\|_{\mathrm{stab}}$ is 
a norm on $\mathbb{W}$) we conclude that 
$\mathbf{U}_1 = \mathbf{U}_2$.
\end{proof}

\begin{theorem}{(Boundedness)}
  The bilinear form is bounded. That is, 
  \begin{align}
    \Big|\mathcal{B}_{\mathrm{stab}}(\mathbf{W},\mathbf{U}) \Big| 
    \leq C \|\mathbf{W}\|_{\mathrm{stab}} \|\mathbf{U}\|_{\mathrm{stab}}
  \end{align}
  where $C$ is a constant.
\end{theorem}
\begin{proof}
  A direct application of the triangle inequality
  of the absolute value on real numbers implies that
  \begin{align}
    \Big|\mathcal{B}_{\mathrm{stab}}(\mathbf{W},\mathbf{U})\Big|
    &\leq
    \frac{1}{2}\Big|\left(\mathbf{w}_1;\mu \mathbf{K}_1^{-1}\mathbf{u}_1\right)\Big|
    + \frac{1}{2}\Big|\left(\mathbf{w}_2;\mu \mathbf{K}_2^{-1}\mathbf{u}_2\right)\Big|
    + \frac{1}{2}\Big|\left(\mathbf{w}_1;\mathrm{grad}[p_1]\right)\Big|
    + \frac{1}{2}\Big|\left(\mathbf{w}_2;\mathrm{grad}[p_2]\right)\Big| \nonumber \\
    &+ \frac{1}{2}\Big|\left(\mathrm{grad}[q_1];\mu^{-1} \mathbf{K}_1\mathrm{grad}[p_1]\right)\Big|
    + \frac{1}{2}\Big|\left(\mathrm{grad}[q_2];\mu^{-1} \mathbf{K}_2\mathrm{grad}[p_2]\right)\Big| \nonumber \\
    &+ \Big|\left(q_1 - q_2;\beta/\mu (p_1 - p_2)\right)\Big|
  \end{align}
  Cauchy-Schwartz inequality on $L_2$ inner-product implies that
  \begin{align}
    \Big|\mathcal{B}_{\mathrm{stab}}(\mathbf{W},\mathbf{U})\Big|
    &\leq
    \frac{1}{2}\left\|\sqrt{\mu} \mathbf{K}_{1}^{-1/2}\mathbf{w}_1\right\|
    \left\|\sqrt{\mu} \mathbf{K}_1^{-1/2}\mathbf{u}_1\right\|
    +\frac{1}{2}\left\|\sqrt{\mu} \mathbf{K}_{2}^{-1/2}\mathbf{w}_2\right\|
    \left\|\sqrt{\mu} \mathbf{K}_2^{-1/2}\mathbf{u}_2\right\| \nonumber \\
    &+ \frac{1}{2}\left\|\sqrt{\mu} \mathbf{K}_{1}^{-1/2}\mathbf{w}_1\right\|
    \left\|\frac{1}{\sqrt{\mu}}\mathbf{K}_{1}^{1/2}\mathrm{grad}[p_1]\right\|
    + \frac{1}{2}\left\|\sqrt{\mu} \mathbf{K}_{2}^{-1/2}\mathbf{w}_2\right\|
    \left\|\frac{1}{\sqrt{\mu}}\mathbf{K}_{2}^{1/2}\mathrm{grad}[p_2]\right\|
    \nonumber \\
    &+ \frac{1}{2}\left\|\frac{1}{\sqrt{\mu}}\mathbf{K}_{1}^{1/2}\mathrm{grad}[q_1]\right\|
    \left\|\frac{1}{\sqrt{\mu}} \mathbf{K}_1^{1/2}\mathrm{grad}[p_1]\right\| 
    + \frac{1}{2}\left\|\frac{1}{\sqrt{\mu}}\mathbf{K}_{2}^{1/2}\mathrm{grad}[q_2]\right\|
    \left\|\frac{1}{\sqrt{\mu}} \mathbf{K}_2^{1/2}\mathrm{grad}[p_2]\right\| \nonumber \\
    &+ \left\|\sqrt{\frac{\beta}{\mu}}\left(q_1 - q_2\right)\right\| \;
    \left\|\sqrt{\frac{\beta}{\mu}} \left(p_1 - p_2\right)\right\|
  \end{align}
  By applying Cauchy-Schwartz inequality on $n$-tuple
  real numbers (i.e., on Euclidean spaces) we obtained the following:
  {\tiny
  \begin{align}
    &\Big|\mathcal{B}_{\mathrm{stab}}(\mathbf{W},\mathbf{U})\Big| \leq \nonumber \\
    &\sqrt{\left\|\sqrt{\mu} \mathbf{K}_{1}^{-1/2}\mathbf{w}_1\right\|^{2}
    +\left\|\sqrt{\mu} \mathbf{K}_{2}^{-1/2}\mathbf{w}_2\right\|^{2}
    + \frac{1}{2}\left\|\frac{1}{\sqrt{\mu}} \mathbf{K}_1^{1/2}\mathrm{grad}[q_1]\right\|^{2} 
    + \frac{1}{2}\left\|\frac{1}{\sqrt{\mu}}\mathbf{K}_{2}^{1/2}\mathrm{grad}[q_2]\right\|^{2}
    + \left\|\sqrt{\frac{\beta}{\mu}}\left(q_1 - q_2\right)\right\|^{2}
    } \nonumber \\
    &\sqrt{
      \frac{1}{2}\left\|\sqrt{\mu} \mathbf{K}_{1}^{-1/2}\mathbf{u}_1\right\|^{2}
      + \frac{1}{2}\left\|\sqrt{\mu} \mathbf{K}_{2}^{-1/2}\mathbf{u}_2\right\|^{2}
      + \left\|\frac{1}{\sqrt{\mu}} \mathbf{K}_1^{1/2}\mathrm{grad}[p_1]\right\|^{2} 
    + \left\|\frac{1}{\sqrt{\mu}}\mathbf{K}_{2}^{1/2}\mathrm{grad}[p_2]\right\|^{2}
    + \left\|\sqrt{\frac{\beta}{\mu}}\left(p_1 - p_2\right)\right\|^{2}
    } \nonumber \\
    \leq &2 \sqrt{\frac{1}{2}\left\|\sqrt{\mu} \mathbf{K}_{1}^{-1/2}\mathbf{w}_1\right\|^{2}
    +\frac{1}{2} \left\|\sqrt{\mu} \mathbf{K}_{2}^{-1/2}\mathbf{w}_2\right\|^{2}
    + \frac{1}{2}\left\|\frac{1}{\sqrt{\mu}} \mathbf{K}_1^{1/2}\mathrm{grad}[q_1]\right\|^{2} 
    + \frac{1}{2}\left\|\frac{1}{\sqrt{\mu}}\mathbf{K}_{2}^{1/2}\mathrm{grad}[q_2]\right\|^{2}
    + \left\|\sqrt{\frac{\beta}{\mu}}\left(q_1 - q_2\right)\right\|^{2}
    } \nonumber \\
    &\sqrt{
      \frac{1}{2}\left\|\sqrt{\mu} \mathbf{K}_{1}^{-1/2}\mathbf{u}_1\right\|^{2}
      + \frac{1}{2}\left\|\sqrt{\mu} \mathbf{K}_{2}^{-1/2}\mathbf{u}_2\right\|^{2}
      + \frac{1}{2} \left\|\frac{1}{\sqrt{\mu}} \mathbf{K}_1^{1/2}\mathrm{grad}[p_1]\right\|^{2} 
      + \frac{1}{2} \left\|\frac{1}{\sqrt{\mu}}\mathbf{K}_{2}^{1/2}\mathrm{grad}[p_2]\right\|^{2}
      + \left\|\sqrt{\frac{\beta}{\mu}}\left(p_1 - p_2\right)\right\|^{2}
    } 
  \end{align}
  }
  That is, we have established that 
  \begin{align}
    \Big|\mathcal{B}_{\mathrm{stab}}(\mathbf{W},\mathbf{U}) \Big| 
    \leq 2 \|\mathbf{W}\|_{\mathrm{stab}} \|\mathbf{U}\|_{\mathrm{stab}}
  \end{align}
  which completes the proof.
\end{proof}

\begin{theorem}{(Coercivity)}
  The bilinear form is coercive. That is, the
  bilinear form is bounded below. 
\end{theorem}
\begin{proof}
  The coercivity of the bilinear form can be established
  from the definition of $\|\cdot\|_{\mathrm{stab}}$
  and Proposition \ref{Prop:DD2_stab_norm} (i.e.,
  $\|\cdot\|_{\mathrm{stab}}$ is a norm on $\mathbb{W}$)
  as
  \begin{align}
    \|\mathbf{W}\|_{\mathrm{stab}}^2 = \mathcal{B}_{\mathrm{stab}}(\mathbf{W},\mathbf{W})  \quad 
    \forall \mathbf{W} \in \mathbb{W}
  \end{align}
\end{proof}
Given the coercivity and boundedness of the bilinear
form and the continuity of the linear functional,
one can conclude that the proposed mixed weak
formulation is well-posed by invoking the
Lax-Milgram theorem \citep{Brenner_Scott}. 
\subsection{Convergence and error analysis of the finite element formulation}
We decompose the computational domain into ``$Nele$''
subdomains (which will be the elements in the context
of the finite element method) such that
\begin{align}
  \label{Eqn:Dual_FE_decomposition}
  \overline{\Omega} = \bigcup_{e = 1}^{Nele}
  \overline{\Omega}^{e}
\end{align}
where a superposed bar indicates the set
closure.
We denote the finite element solution by $\mathbf{U}^{h}$. That is, 
\begin{align}
  \mathbf{U}^{h} = (\mathbf{u}_{1}^{h}(\mathbf{x}),\mathbf{u}_{2}^{h}(\mathbf{x}),
  p_{1}^{h}(\mathbf{x}),
  p_{2}^{h}(\mathbf{x}))
\end{align}
Likewise, 
\begin{align}
\mathbf{W}^{h} = (\mathbf{w}_{1}^{h}(\mathbf{x}),\mathbf{w}_{2}^{h}(\mathbf{x}),
q_{1}^{h}(\mathbf{x}),
q_{2}^{h}(\mathbf{x}))
\end{align}

If we denote the set of all polynomials up to
and including $m$-th order over a set $K$
by $\mathscr{P}^{m}(K)$, and the set
of all continuous functions defined on
$\overline{\Omega}$ (which is the set closure
of $\Omega$) by $C^{0}(\overline{\Omega})$, then the following finite-dimensional spaces can be defined:
\begin{subequations}
  \begin{align}
    \mathcal{U}_{1}^{h} &:= \left\{\mathbf{u}_{1}^{h}(\mathbf{x})
    \in \mathcal{U}_1 \; \Big| \; \mathbf{u}_{1}^{h}(\mathbf{x})
    \in \left(C^{0}(\overline{\Omega})\right)^{nd};
    \mathbf{u}_{1}^{h}(\mathbf{x})|_{\Omega^{e}} \in
    \left(\mathscr{P}^{k}(\Omega^{e})\right)^{nd}; 
    e = 1, \cdots, Nele  \right\} \\
    \mathcal{U}_{2}^{h} &:= \left\{\mathbf{u}_{2}^{h}(\mathbf{x})
    \in \mathcal{U}_{2} \; \Big| \; \mathbf{u}_{2}^{h}(\mathbf{x})
    \in \left(C^{0}(\overline{\Omega})\right)^{nd};
    \mathbf{u}_{2}^{h}(\mathbf{x})|_{\Omega^{e}} \in
    \left(\mathscr{P}^{k}(\Omega^{e})\right)^{nd}; 
    e = 1, \cdots, Nele  \right\} \\
    \mathcal{W}_{1}^{h} &:= \left\{\mathbf{w}_{1}^{h}(\mathbf{x})
    \in \mathcal{W}_1 \; \Big| \; \mathbf{w}_{1}^{h}(\mathbf{x})
    \in \left(C^{0}(\overline{\Omega})\right)^{nd};
    \mathbf{w}_{1}^{h}(\mathbf{x})|_{\Omega^{e}} \in
    \left(\mathscr{P}^{k}(\Omega^{e})\right)^{nd}; 
    e = 1, \cdots, Nele  \right\} \\
    \mathcal{W}_{2}^{h} &:= \left\{\mathbf{w}_{2}^{h}(\mathbf{x})
    \in \mathcal{W}_{2} \; \Big| \; \mathbf{w}_{2}^{h}(\mathbf{x})
    \in \left(C^{0}(\overline{\Omega})\right)^{nd};
    \mathbf{w}_{2}^{h}(\mathbf{x})|_{\Omega^{e}} \in
    \left(\mathscr{P}^{k}(\Omega^{e})\right)^{nd}; 
    e = 1, \cdots, Nele  \right\} \\
    \mathcal{Q}^{h} &:= \left\{
    \left(p_1^{h},p_2^{h}\right) \in
    \mathcal{Q} \; \Big| \; p_{1}^{h}(\mathbf{x}), 
    p_{2}^{h}(\mathbf{x}) \in
    C^{0}(\overline{\Omega}); 
    p_{1}^{h}(\mathbf{x}), p_{2}^{h}(\mathbf{x})|_{\Omega^{e}} \in
    \mathscr{P}^{1}(\Omega^{e}); 
    e = 1, \cdots, Nele  \right\} 
  \end{align}
\end{subequations}
%
%
We define the corresponding product spaces as follows:
\begin{align}
    \mathbb{U}^{h} = \mathcal{U}_{1}^{h} \times \mathcal{U}_{2}^{h} \times \mathcal{Q}^{h}, \quad \mathrm{and} \quad
    \mathbb{W}^{h} = \mathcal{W}_{1}^{h} \times \mathcal{W}_{2}^{h} \times \mathcal{Q}^{h} 
\end{align}
It is important to note that $\mathbb{W}^{h}$
and $\mathbb{U}^{h}$ are closed linear subspaces
of $\mathbb{W}$ and $\mathbb{U}$, respectively. 
The finite element formulation corresponding
to the proposed stabilized mixed formulation
reads:~Find $\mathbf{U}^{h} \in \mathbb{U}^{h}$
such that we have
\begin{align}
  \mathcal{B}_{\mathrm{stab}}(\mathbf{W}^{h};\mathbf{U}^{h})
  = \mathcal{L}_{\mathrm{stab}}(\mathbf{W}^{h})
  \quad \forall \mathbf{W}^{h} \in \mathbb{W}^{h}
  \label{Eqn:VMS_Galerkin_Weak_form_h}
\end{align}

In a given coordinate system, we denote 
$\mathbf{x} = (x_1, \cdots, x_{nd})$. For a given 
multi-index (i.e., tuple) of non-negative integers,  
$\alpha = (\alpha_1, \cdots, \alpha_{nd})$,  
with order $|\alpha| = \alpha_1 + \cdots + 
\alpha_{nd}$, the corresponding partial 
derivative of a scalar field, $p(\mathbf{x})$,  
can be written as follows: 
\begin{align}
  \label{Eqn:VMS_partial_derivative}
  D^{\alpha} p(\mathbf{x}) = 
  \frac{\partial^{|\alpha|} p(\mathbf{x})}{\partial x_{1}^{\alpha_1}
    \partial x_{2}^{\alpha_2} 
    \cdots \partial x_{nd}^{\alpha_{nd}}}
\end{align}  
Using the above notation, the $s^{\mathrm{th}}$
Sobolev semi-norm, $|\cdot|_{s}$, for scalar
and vector fields can be compactly written
as follows: 
\begin{align}
  &|p|_{s}^2 = |p|^2_{H^{s}(\Omega;L)} 
  = \sum_{|\alpha| = s} \int_{\Omega} 
  \left(L^s D^{\alpha} p(\mathbf{x}) \right)^2 
  \mathrm{d}\Omega  \\
  &|\mathbf{u}|_{s}^2 = |\mathbf{u}|^2_{H^{s}(\Omega;L)} 
  = \sum_{|\alpha| = s} \sum_{i = 1}^{nd} \int_{\Omega} 
  \left(L^s D^{\alpha} u_{i}(\mathbf{x}) \right)^2
  \mathrm{d}\Omega 
\end{align}
where $\sum_{|\alpha| = s}$ denotes the
summation over all the possible tuples
of non-negative integers with order $s$,
and $L$ denotes the characteristic length
of the domain. 
\begin{remark}
  Although the notation introduced in equation
  \eqref{Eqn:VMS_partial_derivative} is common in the 
  theory of partial differential equations (e.g., 
  \citep{Evans_PDE}), it may not be that common in
  the engineering literature. For the benefit
  of the reader, we provide the following few examples
  to make the notation more apparent: 
  \begin{align*}
    &\mathrm{if} \; nd = 2, \; \mathbf{x} = (x_1, x_2), \alpha = (2,1) 
    \; \mathrm{then} \; |\alpha| = 3 \; \mathrm{and} \; 
    D^{\alpha} p(\mathbf{x}) 
    = \frac{\partial^3 p(\mathbf{x})}{\partial x_1^2 \partial x_2} \\
    &\mathrm{if} \; nd = 3, \; \mathbf{x} = (x_1, x_2, x_3), 
    \alpha = (3,0,6) 
    \; \mathrm{then} \; |\alpha| = 9 \; \mathrm{and} \; 
    D^{\alpha} p(\mathbf{x}) 
    = \frac{\partial^9 p(\mathbf{x})}{\partial x_1^3 
      \partial x_3^6} 
  \end{align*}
\end{remark}

We now show the consistency of the formulation,
and then establish the stability. We also
obtain the rates of convergence with the mesh
refinement and the order of interpolation. 
To this end, the error $\mathbf{E}$ 
is defined as
\begin{align}
\mathbf{E} = \mathbf{U}^{h} - \mathbf{U}
\end{align}
We employ the following standard decomposition 
of error (e.g., see \citep{Brenner_Scott}):   
\begin{align}
\label{Eqn:VMS_Error_decomposition}
\mathbf{E} = \mathbf{U}^{h} - \mathbf{U} 
= \underbrace{\mathbf{U}^{h} - \widetilde{\mathbf{U}}^{h}}_{\mbox{approximation error}} 
+ \underbrace{\widetilde{\mathbf{U}}^{h} - \mathbf{U}}_{\mbox{interpolation error}} = \mathbf{E}^h + \mathbf{H}
\end{align}
where $\widetilde{\mathbf{U}}^{h}$ denotes the
interpolate of $\mathbf{U}$ onto $\mathbb{U}^{h}$, 
$\mathbf{E}^h$ is the approximation error and 
$\mathbf{H}$ denotes the interpolation error. 
The interpolation error $\mathbf{H}$ satisfies the 
following standard inequality \citep{Brezzi_Fortin}:
%
\begin{align}
\|\mathbf{H}\|_{\mathrm{stab}} \leq C_{1} \left(\frac{h}{L}\right)^{k+1} |\mathbf{u}_{1}|_{k+1} + C_{2} \left(\frac{h}{L}\right)^{l+1} |\mathbf{u}_{2}|_{l+1} + C_{3} \left(\frac{h}{L}\right)^m |p_{1}|_{m+1} + C_{4} \left(\frac{h}{L}\right)^n |p_{2}|_{n+1}
\end{align}
In the above inequality, $h$
  is the characteristic mesh parameter,
$L$ is a characteristic dimension of the
domain $\Omega$, and $k$, $l$, $m$, and
$n$ are natural numbers.
As mentioned earlier, we have taken
$h$ to be the maximum edge length
in the mesh. However, the results
presented herein are equally valid
for other choices of $h$; for example,
the maximum element diameter.
The constants
$C_{1},~C_2,~C_3$ and $C_{4}$ are
defined as follows: 
\begin{align}
&C_{1} = C_0 \sup_{\mathbf{x}\in \Omega} \left(\mu(\mathbf{x})k^{-1}_{1}(\mathbf{x})\right)^\frac{1}{2}, \quad C_{2} = C_0 \sup_{\mathbf{x}\in \Omega} \left(\mu(\mathbf{x})k^{-1}_{2}(\mathbf{x})\right)^\frac{1}{2}, \nonumber \\
&C_{3} = \frac{C_0}{L} \sup_{\mathbf{x}\in \Omega} \left(\frac{1}{\mu(\mathbf{x})}k_{1}(\mathbf{x})\right)^\frac{1}{2} 
\quad \mathrm{and}  \quad C_{4} = \frac{C_0}{L} \sup_{\mathbf{x}\in \Omega} \left(\frac{1}{\mu(\mathbf{x})}k_{2}(\mathbf{x})\right)^\frac{1}{2}
\end{align}
where $C_0$ is a non-dimensional constant. Note that 
$C_1$, $C_2$, $C_3$ and $C_4$ are independent of 
$h$, $\mathbf{u}_{1}$, $\mathbf{u}_2$, $p_1$ and $p_2$. 

\begin{theorem}{(Consistency)}
  \label{Prop:Dual_Prop_Consistency} 
The error in the finite element solution satisfies
    \begin{align}
    \mathcal{B}_{\mathrm{stab}}(\mathbf{W}^{h};
    \mathbf{E}) = 0 \quad 
    \forall \mathbf{W}^{h} \in \mathbb{W}^{h}  
    \subset \mathbb{W}
  \end{align}
\end{theorem}
\begin{proof}
The finite element solution satisfies 
\begin{align}
\mathcal{B}_{\mathrm{stab}}(\mathbf{W}^{h},\mathbf{U}^{h}) 
= \mathcal{L}_{\mathrm{stab}}(\mathbf{W}^{h}) 
\quad \forall \mathbf{W}^{h} \in \mathbb{W}^{h} 
\end{align}
The exact solution clearly satisfies 
\begin{align}
\mathcal{B}_{\mathrm{stab}}(\mathbf{W}^{h},\mathbf{U}) 
= \mathcal{L}_{\mathrm{stab}}(\mathbf{W}^{h}) 
\quad \forall \mathbf{W}^{h} \in \mathbb{W}^{h}
\end{align}
By subtracting the above two equations and using the linearity 
of the bilinear form $\mathcal{B}_{\mathrm{stab}}(\cdot,\cdot)$ 
in the second slot, we obtain the desired result.  
\end{proof}

\begin{theorem}{(Convergence)}
\label{Thm:Dual_Thm_Stability}
For all $\tilde{\mathbf{U}}^{h} \in \mathbb{U}^h$, 
the error satisfies 
\begin{align}
\|\mathbf{E}\|_{\mathrm{stab}} \leq 
C \|\mathbf{H}\|_{\mathrm{stab}}
\end{align}
where $C$ is a non-dimensional constant.
\end{theorem}
\begin{proof}
  Noting the decomposition of error mentioned in equation \eqref{Eqn:VMS_Error_decomposition} (i.e., $\mathbf{E} =
  \mathbf{E}^{h} + \mathbf{H}$), we proceed as follows:
  \begin{alignat}{2}
    \|\mathbf{E}\|_{\mathrm{stab}}^{2} &= \mathcal{B}_{\mathrm{stab}}(\mathbf{E};\mathbf{E})
    && \quad (\mathrm{definition~of~}\|\cdot\|_{\mathrm{stab}}\mathrm{~norm}) \nonumber \\
    & = \mathcal{B}_{\mathrm{stab}}(\mathbf{E}^{h} + \mathbf{H};\mathbf{E})
    && \quad (\mathrm{standard~decomposition~of~}\mathbf{E} ) \nonumber \\
    & = \mathcal{B}_{\mathrm{stab}}(\mathbf{E}^{h};\mathbf{E}) + \mathcal{B}_{\mathrm{stab}}(\mathbf{H};\mathbf{E})
    &&\quad (\mathrm{bilinearity}) \nonumber \\
    & = \mathcal{B}_{\mathrm{stab}}(\mathbf{H};\mathbf{E})
    &&\quad (\mathrm{consistency})
    \label{Eqn:DDP2_stab_theoerem_Estab}
  \end{alignat}
  We now estimate $\mathcal{B}_{\mathrm{stab}}
  (\mathbf{H};\mathbf{E})$. To this end, we
  denote the components of $\mathbf{E}$ and
  $\mathbf{H}$ as follows:
  \begin{align*}
  \mathbf{E} = \left\{\mathbf{e}_{\mathbf{u}_{1}},\mathbf{e}_{\mathbf{u}_{2}},e_{p_{1}},e_{p_{2}} \right\},
  \quad \mathrm{and} \quad 
  \mathbf{H} = \left\{ \boldsymbol{\eta}_{\mathbf{u}_{1}},\boldsymbol{\eta}_{\mathbf{u}_{2}},\eta_{p_{1}},\eta_{p_{2}}\right\}
  \end{align*}
  By repeated use of Cauchy-Schwartz and Peter-Paul inequalities \citep{Hunter_Nachtergaele},
  we estimate $\mathcal{B}_{\mathrm{stab}}(\mathbf{H};\mathbf{E})$ as follows:  
{\small{  
  \begin{align}
    \mathcal{B}_{\mathrm{stab}}(\mathbf{H};\mathbf{E})
    &= \mathcal{B}_{\mathrm{stab}}(\boldsymbol{\eta}_{\mathbf{u}_{1}},\boldsymbol{\eta}_{\mathbf{u}_{2}},
    \eta_{p_1},\eta_{p_{2}};\mathbf{e}_{\mathbf{u}_{1}},\mathbf{e}_{\mathbf{u}_{2}},e_{p_1},e_{p_{2}}) \nonumber \\
    & = (\boldsymbol{\eta}_{\mathbf{u}_{1}};\mu \mathbf{K}_{1}^{-1}\mathbf{e}_{\mathbf{u}_{1}})
  - (\mathrm{div}[\boldsymbol{\eta}_{\mathbf{u}_{1}}];e_{p_{1}})
  + (\eta_{p_{1}};\mathrm{div}[\mathbf{e}_{\mathbf{u}_{1}}]) \nonumber \\
  & + (\boldsymbol{\eta}_{\mathbf{u}_{2}};\mu \mathbf{K}_{2}^{-1}\mathbf{e}_{\mathbf{u}_{2}})
  - (\mathrm{div}[\boldsymbol{\eta}_{\mathbf{u}_{2}}];e_{p_{2}})
  + (\eta_{p_{2}};\mathrm{div}[\mathbf{e}_{\mathbf{u}_{2}}]) \nonumber \\
  &-\frac{1}{2} \left(\boldsymbol{\eta}_{\mathbf{u}_{1}};\mu \mathbf{K}_{1}^{-1} \mathbf{e}_{\mathbf{u}_{1}}\right) -\frac{1}{2} \left(\boldsymbol{\eta}_{\mathbf{u}_{1}};\mathrm{grad}[e_{p_{1}}] \right)+ \frac{1}{2} \left(\mathrm{grad}[\eta_{p_{1}}];\mathbf{e}_{\mathbf{u}_{1}}\right) + \frac{1}{2}\left(\mathrm{grad}[\eta_{p_{1}}];\frac{1}{\mu} \mathbf{K}_{1}\mathrm{grad}[e_{p_{1}}]\right)
  \nonumber \\
   &-\frac{1}{2} \left(\boldsymbol{\eta}_{\mathbf{u}_{2}};\mu \mathbf{K}_{2}^{-1} \mathbf{e}_{\mathbf{u}_{2}}\right) -\frac{1}{2} \left(\boldsymbol{\eta}_{\mathbf{u}_{2}};\mathrm{grad}[e_{p_{2}}] \right)+ \frac{1}{2} \left(\mathrm{grad}[\eta_{p_{2}}];\mathbf{e}_{\mathbf{u}_{2}}\right) + \frac{1}{2}\left(\mathrm{grad}[\eta_{p_{2}}];\frac{1}{\mu} \mathbf{K}_{2}\mathrm{grad}[e_{p_{2}}]\right) \nonumber \\
   & + ((\eta_{p_{1}}- \eta_{p_{2}});\beta/\mu(e_{p_{1}} - e_{p_{2}})) \nonumber \\
& \leq \frac{1}{2}\left\{ \varepsilon_1 \| \sqrt{\mu} \mathbf{K}_1^{-1/2} \boldsymbol{\eta}_{\mathbf{u}_{1}} \|^{2} + \frac{1}{\varepsilon_1} \| \sqrt{\mu} \mathbf{K}_1^{-1/2} \mathbf{e}_{\mathbf{u}_{1}} \|^{2} + \varepsilon_2\| \sqrt{\mu} \mathbf{K}_1^{-1/2} \boldsymbol{\eta}_{\mathbf{u}_{1}} \|^{2}  \right. \nonumber \\
& \left. \quad + \frac{1}{\varepsilon_2} \| \frac{1}{\sqrt{\mu}}\mathbf{K}_1^{1/2}
 \mathrm{grad}[e_{p_{1}}] \|^{2} + \varepsilon_3\| \frac{1}{\sqrt{\mu}}\mathbf{K}_1^{1/2}
 \mathrm{grad}[\eta_{p_{1}}] \|^2  +\frac{1}{\varepsilon_3} \| \sqrt{\mu} \mathbf{K}_1^{-1/2} \mathbf{e}_{\mathbf{u}_{1}} \|^2 \right. \nonumber \\
&\left. \quad +\varepsilon_4 \| \sqrt{\mu} \mathbf{K}_2^{-1/2} \boldsymbol{\eta}_{\mathbf{u}_{2}} \|^{2} + \frac{1}{\varepsilon_4} \| \sqrt{\mu} \mathbf{K}_2^{-1/2} \mathbf{e}_{\mathbf{u}_{2}} \|^{2}  + \varepsilon_5 \| \sqrt{\mu} \mathbf{K}_2^{-1/2} \boldsymbol{\eta}_{\mathbf{u}_{2}} \|^{2}  \right. \nonumber \\
& \left. \quad + \frac{1}{\varepsilon_5}\| \frac{1}{\sqrt{\mu}}\mathbf{K}_2^{1/2} \mathrm{grad}[e_{p_{2}}] \|^{2} + \varepsilon_6 \| \frac{1}{\sqrt{\mu}}\mathbf{K}_2^{1/2} \mathrm{grad}[\eta_{p_{2}}] \|^2  + \frac{1}{\varepsilon_6} \| \sqrt{\mu} \mathbf{K}_2^{-1/2} \mathbf{e}_{\mathbf{u}_{2}} \|^2  \right. \nonumber \\
& \left. \quad +  \frac{\varepsilon_7}{2} \| \sqrt{\mu} \mathbf{K}_1^{-1/2} \boldsymbol{\eta}_{\mathbf{u}_{1}} \|^2 + \frac{1}{2 \varepsilon_7} \| \sqrt{\mu} \mathbf{K}_1^{-1/2} \mathbf{e}_{\mathbf{u}_{1}} \|^2 +  \frac{\varepsilon_8}{2} \| \sqrt{\mu} \mathbf{K}_1^{-1/2} \boldsymbol{\eta}_{\mathbf{u}_{1}} \|^2  \right. \nonumber \\
& \left. \quad + \frac{1}{2\varepsilon_8} \| \frac{1}{\sqrt{\mu}}\mathbf{K}_1^{1/2} \mathrm{grad}[e_{p_{1}}] \|^2 +  \frac{\varepsilon_{9}}{2} \| \frac{1}{\sqrt{\mu}}\mathbf{K}_1^{1/2} \mathrm{grad}[\eta_{p_{1}}] \|^2 + \frac{1}{2 \varepsilon_{9}} \| \sqrt{\mu} \mathbf{K}_1^{-1/2} \mathbf{e}_{\mathbf{u}_{1}} \|^2  \right. \nonumber \\
& \left. \quad + \frac{\varepsilon_{10}}{2} \| \frac{1}{\sqrt{\mu}}\mathbf{K}_1^{1/2} \mathrm{grad}[\eta_{p_{1}}] \|^2 + \frac{1}{2 \varepsilon_{10}} \| \frac{1}{\sqrt{\mu}}\mathbf{K}_1^{1/2} \mathrm{grad}[e_{p_{1}}] \|^2 + \frac{\varepsilon_{11}}{2} \| \sqrt{\mu} \mathbf{K}_2^{-1/2} \boldsymbol{\eta}_{\mathbf{u}_{2}} \|^2  \right. \nonumber \\
  & \left. \quad + \frac{1}{2\varepsilon_{11}} \| \sqrt{\mu} \mathbf{K}_2^{-1/2} \mathbf{e}_{\mathbf{u}_{2}} \|^2 +  \frac{\varepsilon_{12}}{2} \| \sqrt{\mu} \mathbf{K}_2^{-1/2} \boldsymbol{\eta}_{\mathbf{u}_{2}} \|^2 + \frac{1}{2 \varepsilon_{12}} \| \frac{1}{\sqrt{\mu}}\mathbf{K}_2^{1/2} \mathrm{grad}[e_{p_{2}}] \|^2 \right. \nonumber \\
& \left.\quad + \frac{\varepsilon_{13}}{2} \| \frac{1}{\sqrt{\mu}}\mathbf{K}_2^{1/2} \mathrm{grad}[\eta_{p_{2}}] \|^2 +\frac{1}{2 \varepsilon_{13}} \| \sqrt{\mu} \mathbf{K}_2^{-1/2} \mathbf{e}_{\mathbf{u}_{2}} \|^2 + \frac{\varepsilon_{14}}{2} \| \frac{1}{\sqrt{\mu}}\mathbf{K}_2^{1/2} \mathrm{grad}[\eta_{p_{2}}] \|^2  \right. \nonumber \\
& \left.\quad + \frac{1}{2 \varepsilon_{14}} \| \frac{1}{\sqrt{\mu}}\mathbf{K}_2^{1/2} \mathrm{grad}[e_{p_{2}}] \|^2 + \varepsilon_{15} \|\left(\beta/ \mu\right)^{1/2} (\eta_{p_{1}}-\eta_{p_{2}})\|^2 \right. \nonumber \\
& \left.\quad + \frac{1}{\varepsilon_{15}} \| \left(\beta/ \mu\right)^{1/2} (e_{p_{1}}-e_{p_{2}}) \|^2  \right\}
  \end{align}
  }}
  where $\varepsilon_{i},~i=1,\ldots, 15$ are positive
  constants. By choosing 
\begin{align}
  2\varepsilon_1 = 2 \varepsilon_3 = 2\varepsilon_4 = 2\varepsilon_6
  = \varepsilon_7 = \varepsilon_{9} = \varepsilon_{11} = \varepsilon_{13}= 10
  \quad \mathrm{and} \nonumber \\
  2\varepsilon_2 = 2\varepsilon_5 = \varepsilon_8 = \varepsilon_{10} = \varepsilon_{12} = \varepsilon_{14}= 6, \; \varepsilon_{15} = 1
\end{align}
we obtain the following inequality: 
\begin{align}
  \label{Eqn:DPP_estimate_not_sharp}
\mathcal{B}_{\mathrm{stab}}(\mathbf{H};\mathbf{E})& \leq \frac{1}{2}\left\{ \| \mathbf{E}\|_{\mathrm{stab}}^2 + 16 \| \sqrt{\mu} \mathbf{K}_1^{-1/2} \boldsymbol{\eta}_{\mathbf{u}_{1}} \|^{2} + 13 \| \frac{1}{\sqrt{\mu}}\mathbf{K}_1^{1/2}
 \mathrm{grad}[\eta_{p_{1}}] \|^2 \right. \nonumber \\
& \left. \quad + 16 \| \sqrt{\mu} \mathbf{K}_2^{-1/2} \boldsymbol{\eta}_{\mathbf{u}_{2}} \|^{2} + 13 \| \frac{1}{\sqrt{\mu}}\mathbf{K}_2^{1/2}
 \mathrm{grad}[\eta_{p_{2}}] \|^2 +  \|\left(\beta/ \mu\right)^{1/2} (\eta_{p_{1}}-\eta_{p_{2}})\|^2 \right\}\nonumber \\
& \leq \frac{1}{2} \|\mathbf{E}\|_{\mathrm{stab}}^2 + 16 \| \mathbf{H}\|_{\mathrm{stab}}^2
\end{align}
Noting equation \eqref{Eqn:DDP2_stab_theoerem_Estab} we have 
\begin{align}
  \mathcal{B}_{\mathrm{stab}}(\mathbf{H};\mathbf{E}) = 
  \| \mathbf{E}\|_{\mathrm{stab}}^2
  \leq 32 \| \mathbf{H}\|_{\mathrm{stab}}^2
\end{align}
which gives the following estimate of the total error in terms of the interpolation error: 
\begin{align}
  \|\mathbf{E}\|_{\mathrm{stab}} \leq 4\sqrt{2} \| \mathbf{H}\|_{\mathrm{stab}}
\end{align}
This completes the proof.
\end{proof}

The set of choices made for constants,
$\varepsilon_{i} \; (i = 1, \cdots,
15)$, is one of many such ones to
obtain an upper bound for
$\mathcal{B}_{\mathrm{stab}}
(\mathbf{H};\mathbf{E})$ in terms
of the total error, $\mathbf{E}$,
and the interpolation error,
$\mathbf{H}$. We do not claim
that this selection of constants
is optimal. Certainly, the estimate
\eqref{Eqn:DPP_estimate_not_sharp}
and the subsequent ones are not sharp.
Although obtaining sharp estimates is
of theoretical significance in mathematical
analysis, it is not crucial to establish the
convergence of the proposed stabilized
formulation. We, therefore, do not pursue
further with respect to obtaining the optimal
choices for the constants $\varepsilon_{i}$,
and for obtaining a sharp estimate for
$\mathcal{B}_{\mathrm{stab}}(\mathbf{H};\mathbf{E})$.

%% file: Sections/S5_VMS_Numerical_convergence.tex
\section{PATCH TESTS AND NUMERICAL CONVERGENCE ANALYSIS}
\label{Sec:S5_VMS_Canonical}
In order to assess the convergence behavior of
a numerical (finite element) formulation and to
determine whether it is programmed correctly,
patch tests are commonly used.
In this section, we first illustrate the performance
of the proposed stabilized mixed formulation under
the equal-order interpolation for all the field
variables using one-dimensional and three-dimensional
constant-flow patch tests. We also compare the results
obtained under the proposed stabilized mixed formulation
with that of the classical mixed formulation
(which is based on the Galerkin formalism).
We then perform a systematic numerical convergence
analysis of the proposed stabilized mixed formulation
under $h$- and $p$-refinements, and compare the
obtained rates of convergence with the theory. 

Under our studies on patch tests and numerical
convergence analysis, we often use the term
machine precision, which is the smallest
difference between two numbers that the
computing machine recognizes \citep{Heath_numerical}.
Mathematically, the machine precision of
a computing machine, $\epsilon_{\mathrm{mach}}$,
satisfies
\begin{alignat*}{2}
  &(1 + \epsilon) - 1 = 0 \quad
  &&\forall \epsilon < \epsilon_{\mathrm{mach}} \\
  \mathrm{and}\; &(1 + \epsilon) - 1 = \epsilon \neq 0 \quad
  &&\forall \epsilon > \epsilon_{\mathrm{mach}}
\end{alignat*}
It is important to note that the machine precision
depends on the underlying hardware of the computer,
and hence, its value can vary from one computer to
another. 
It is also important to note that the machine precision
of a computer is not the smallest number that the
computer can represent. To put the things quantitatively,
the machine precision on a 32-bit machine is approximately
$10^{-7}$ and on a 64-bit machine, it is approximately
$10^{-16}$ \citep{higham2002accuracy}.
On the other hand, the smallest positive numbers
that a 32-bit machine and a 64-bit machine can
represent are approximately $10^{-38}$ and $10^{-308}$,
respectively \citep{higham2002accuracy}. 

\subsection{One-dimensional constant flow patch test}
\label{Sub:One-dimensional_patch_test}
The purpose of solving the one-dimensional example
is to provide a simple numerical tool for testing
whether the proposed mixed formulation satisfies
the LBB condition. Figure \ref{Fig:1D_patch_test}
provides a pictorial description of the problem,
and Table \ref{Tb1:1D_patch_test_data}
  provides the data-set for this problem.
The domain is a line of unit
length along x direction. On the left end of the
domain, pressures $p_{1}^{\mathrm{L}}$ and $p_{2}^{\mathrm{L}}$
are prescribed in macro- and micro-pore networks,
respectively. Similarly, on the right end of the
domain, $p_{1}^{\mathrm{R}}$ and $p_{2}^{\mathrm{R}}$
are, respectively, prescribed in the macro- and
the micro-pore networks.
Since a pressure boundary 
condition is prescribed for at least one
of the pore-networks, the condition of
vanishing mean pressure in one of pore-networks
in the function space $\mathcal{Q}$ (which is defined
in equation \eqref{Eqn:VMS_Function_space_Q})
is not appropriate for this problem. See the
discussion in Section \ref{Sec:S3_VMS_Mixed}.

{\small
	\begin{table}[!h]
		\caption{Data-set for one-dimensional
				constant flow patch test and 1D numerical
				convergence analysis.}
		\centering
		\begin{tabular}{|c|c|} \hline
			Parameter & Value \\
			\hline
		    $\gamma b$ & $0.0$\\
		    $L$ & $1.0$ \\
		    $\mu $ & $1.0$ \\
		    $\beta $ & $1.0$ \\
		    $k_1$&  $1.0$ \\
		    $k_2$&  $0.01$\\
		    $p_1^L$&  $10.0$ \\
		    $p_1^R$&  $1.0$ \\
		    $p_2^L$&  $10.0$ \\
		    $p_2^R$&  $1.0$ \\
			\hline 
		\end{tabular}
		\label{Tb1:1D_patch_test_data}
	\end{table}
}

The governing equations can be written as follows:
\begin{subequations}
  \begin{alignat}{2}
  \label{Eqn:VMS_BoLM1_1Dproblem}
    &\mu k_{1}^{-1} \mathbf{u}_1(\mathbf{x})
    + \frac{d p_1}{dx} = 0 
    &\quad \mathrm{in} \; (0,L) \\
    \label{Eqn:VMS_BoLM2_1Dproblem}
    &\mu k_{2}^{-1} \mathbf{u}_2(\mathbf{x})
    + \frac{d p_1}{dx} = 0
    &\quad \mathrm{in} \; (0,L) \\
    \label{Eqn:VMS_BoM1_1Dproblem}
    &\frac{d u_1}{dx} = +\chi(x)  
    &\quad \mathrm{in} \; (0,L) \\
    \label{Eqn:VMS_BoM2_1Dproblem}
    &\frac{d u_2}{dx} = -\chi(x)  
    &\quad \mathrm{in} \; (0,L) \\
    \label{Eqn:VMS_pBC1_1Dproblem}
    &p_1(x = 0) = p_1^{\mathrm{L}}, \quad  
    p_1(x = L) = p_1^{\mathrm{R}} \\
    \label{Eqn:VMS_pBC2_1Dproblem}
    &p_2(x = 0) = p_2^{\mathrm{L}}, \quad  
    p_2(x = L) = p_2^{\mathrm{R}} 
  \end{alignat}
\end{subequations}
It should be noted that the quantities used in equations
\eqref{Eqn:VMS_BoLM1_1Dproblem}--\eqref{Eqn:VMS_pBC2_1Dproblem}
are non-dimensional. More details on non-dimensionalization
procedure can be found in \citep{Nakshatrala_Joodat_Ballarini}. 
In this boundary value problem, $k_1$
and $k_2$ are assumed to be independent
of $x$ and the mass transfer between the
two pore-networks takes the following form: 
\begin{align}
  \chi(x) =
  -\left(p_{1}(x)
  - p_{2}(x)\right)
\end{align}

The analytical solution for this simple 1D
problem includes constant velocities and
linearly varying pressures (from $p^L_i$
to $p^R_i$) at each pore-network along
the x direction.

Figure  \ref{Fig:Pressure_velocity_VMS_Galerkin_1D}
shows the numerical results for pressure and velocity
profiles in the two pore-networks under Galerkin and
the proposed stabilized mixed formulations. The values
of velocity vector fields in the two pore-networks match
the analytical solutions under both proposed the stabilized
mixed formulation and the Galerkin formulation. As can
be seen in Figures \ref{Fig:Macro_pressure_patch_test_VG}
and \ref{Fig:Micro_pressure_patch_test_VG}, under the stabilized mixed formulation, pressures in the two pore-networks vary linearly from the prescribed value at the left end ($p_{i}^{\mathrm{L}},~ i=1,2$) to the prescribed one at the right end ($p_{i}^{\mathrm{R}},~ i=1,2$). These results are in agreement with the corresponding analytical solutions up to the machine precision, thus showing that the proposed formulation performs well and that it satisfies the 1D patch test. However, under the Galerkin formulation, spurious oscillations are observed in the pressure fields in both macro- and micro-networks even for equal-order interpolation. 

\subsection{Three-dimensional constant flow patch test}
Previous research studies have shown that many existing numerical formulations cannot perform well when they are extended to 3D settings \citep{Nakshatrala_Turner_Hjelmstad_Masud_CMAME_2006_v195_p4036, Hughes_Masud_Wan_2006}. Herein, using the 3D constant flow patch test we will show that the proposed stabilized mixed formulation performs well even in 3D settings and it's capable of satisfying the LBB condition. To illustrate this, we consider the unit cube computational domain shown in Figure \ref{Fig:3D_patch_test}. On the left and right faces, pressures $p_{i}^{\mathrm{L}},~i=1,2$ and $p_{i}^{\mathrm{R}},~i=1,2$ are prescribed respectively where $ i=1$ denotes the macro-pore network and $i=2$ represents the micro-pore network. On the other faces, the velocity boundary condition is prescribed in the two pore-networks (i.e., $\mathbf{u}_{i} \cdot \widehat{\mathbf{n}}=0,~i=1,2$). Table \ref{Tb2:3D_patch_test_data} provides the parameter values for this test problem.
{\small
  \begin{table}[!h]
    \caption{ Data-set for
        three-dimensional constant flow
        patch test.}
    \centering
    \begin{tabular}{|c|c|} \hline
      Parameter & Value \\
      \hline
      $\gamma \mathbf{b}$ & $\{0.0,0.0,0.0\}$\\
      $L_x $ & $1.0$ \\
      $L_y $ & $1.0$ \\
      $\mu $ & $1.0$ \\
      $\beta $ & $1.0$ \\
      $k_1$&  $1.0$ \\
      $k_2$&  $0.01$\\
      $p_1^L$&  $10.0$ \\
      $p_1^R$&  $1.0$ \\
      $p_2^L$&  $10.0$ \\
      $p_2^R$&  $1.0$ \\
      \hline 
    \end{tabular}
    \label{Tb2:3D_patch_test_data}
  \end{table}
}

The analytical solution pair for this constant flow patch test includes constant velocity along x direction and pressure linearly varying along x direction at each pore-network. Figure \ref{Fig:Pressure_VMS_Galerkin_3D} shows the numerical results for pressure profiles associated with the two pore-networks under Galerkin and the stabilized mixed formulations. It is observed that the Galerkin formulation produces spurious oscillations in micro- and macro-pressures even for equal-order interpolation. This indicates that Galerkin formulation cannot accurately predict pressure variations and that the results are not stable. These oscillations are completely eliminated by the proposed stabilized mixed formulation, thus illustrating the stability of the solution. This verifies that the proposed numerical formulation performs well and satisfies the 3D constant flow patch test.
\subsection{Numerical convergence under $h$- and $p$-refinements}
In this subsection, the convergence behavior of the proposed stabilized mixed formulation is evaluated. For this purpose, the convergence analysis is performed in 1D and 2D settings. The convergence rates are obtained under two different approaches. The first method is called \emph{$h$-refinement} where the number of elements is increased and hence the size of elements (denoted by ``$h$'') in the domain is decreased. The convergence rates under $h$-refinement are obtained for various polynomial orders. In the second approach, the so-called \emph{$p$-refinement}, the convergence rate is calculated by changing the order of polynomial while the total number of elements in the domain is kept fixed (nx = 5).

\subsubsection{1D numerical convergence analysis}
For the convergence analysis in the 1D setting, we select the previously defined one-dimensional patch test (subsection \ref{Sub:One-dimensional_patch_test}). In Figures \ref{Fig:Dual_Problem_1_h_refinement} and \ref{Fig:Dual_Problem_1_p_refinement}, the convergence rates under $h$- and $p$-refinements are shown for the $L_2$-norm of the velocity fields in the macro- and micro-pore networks (denoted by ``$L_2~u_1$'' and ``$L_2~u_2$'', respectively), the $L_2$-norm of the pressure fields in the macro- and micro-pore networks (denoted by ``$L_2~p_1$'' and ``$L_1~p_2$'', respectively), and the $H^1$-norm of the pressure fields in the macro- and micro-pore networks (denoted by ``$H^1~p_1$'' and ``$H^1~p_2$'', respectively). As can be seen in these figures, the rate of convergence for $h$-refinement is polynomial and for $p$-refinement is exponential, which are in accordance with the theory. 

\subsubsection{2D numerical convergence analysis}
The convergence analysis in the 2D setting is performed on the unit square domain shown in Figure \ref{Fig:Dual_Problem_2D_domain}. The macro- and micro-pressures are prescribed on the four sides of the computational domain. Table \ref{Tb3:3D_convergence_analysis_data} provides the parameter values for the 2D convergence analysis.
%
{\small
	\begin{table}[!h]
		\caption{Data-set for 2D numerical convergence analysis.}
		\centering
		\begin{tabular}{|c|c|} \hline
			Parameter & Value \\
			\hline
			$\gamma \mathbf{b}$ & $\{0.0,0.0\}$\\
			$L_x $ & $1.0$ \\
			$L_y $ & $1.0$ \\
			$\mu $ & $1.0$ \\
			$\beta $ & $1.0$ \\
			$k_1$&  $1.0$ \\
			$k_2$&  $0.1$\\
			$\eta$ & $\sqrt{11} \simeq 3.3166$\\
			\hline
			$p_i^{\mathrm{left}},~i=1,2$&  Obtained by evaluating   \\
			$p_i^{\mathrm{right}},~i=1,2$&  the analytical solution \\
			$p_i^{\mathrm{top}},~i=1,2$&   (equations \eqref{Eqn:2D_Convergence_Analytical_p1} and \eqref{Eqn:2D_Convergence_Analytical_p2} ) \\
			$p_i^{\mathrm{bottom}},~i=1,2$&  on the respective boundaries. \\
			\hline 
		\end{tabular}
		\label{Tb3:3D_convergence_analysis_data}
	\end{table}
} 
For convenience, let us define 
\begin{align}
\eta := \sqrt{\beta \frac{k_1 + k_2}{k_1 k_2}}
\end{align}
Then the analytical solution for the velocity fields can be defined as 
\begin{align}
\mathbf{u}_1(x,y) &= -k_1\left(\begin{array}{c}
\exp(\pi x) \sin(\pi y) \\
\exp(\pi x) \cos(\pi y)
\end{array}\right) 
+ \left(\begin{array}{c}
0 \\
\frac{\eta}{\beta} \exp(\eta y) 
\end{array}\right) 
\\
\mathbf{u}_2(x,y) &= -k_2\left(\begin{array}{c}
\exp(\pi x) \sin(\pi y) \\
\exp(\pi x) \cos(\pi y)
\end{array}\right) 
- \left(\begin{array}{c}
0 \\
\frac{\eta}{\beta} \exp(\eta y)
\end{array}\right) 
\end{align}
The analytical solution for the pressure fields can then be obtained as follows:
\begin{align}
p_1(x,y) &= \frac{\mu}{\pi} \exp(\pi x) \sin(\pi y) - \frac{\mu}{\beta k_1} \exp(\eta y) \label{Eqn:2D_Convergence_Analytical_p1}\\
p_2(x,y) &= \frac{\mu}{\pi} \exp(\pi x) \sin(\pi y) + \frac{\mu}{\beta k_2} \exp(\eta y) 
\label{Eqn:2D_Convergence_Analytical_p2}
\end{align}
Figure \ref{Fig:Dual_Problem_2D_h_refinement} provides the convergence rates under $h$-refinement for the $L_2$-norm and the $H^1$-norm of the pressure fields in the macro- and micro-pore networks. The results under $p$-refinement  for the $L_2$-norm of the pressure fields are also provided in Figure \ref{Fig:Dual_Problem_2D_p_refinement}. The rates of convergence for $h$- and $p$-refinements are respectively polynomial and exponential, which are in accordance with the theory. As can be seen, the error under $p$-refinement flattened out around $10^{−16}$ for larger number of degrees-of-freedom. This is expected as the machine precision on a 64-bit machine is around $10^{−16}$.
The results obtained from the one-dimensional and two-dimensional problems verify that the proposed stabilized mixed formulation is convergent.

%% file: Sections/S6_VMS_NR.tex
\section{REPRESENTATIVE NUMERICAL RESULTS}
\label{Sec:S6_VMS_NR}

In the previous section, the convergence behavior
of the proposed mixed formulation has been assessed
using patch tests and numerical convergence analysis.
In this section, using representative problems with
relevance to technological applications, the flow
characteristics in the porous media exhibiting
double porosity/permeability are studied. The
performance of the Nitsche's method is illustrated
using two-dimensional candle filter problem and
three-dimensional hollow sphere problem. 

\subsection{Two-dimensional candle filter problem}
The aim of this problem is to show how the velocity boundary conditions can be enforced weakly in two-dimensional settings using Nitsche's method. This two-dimensional boundary value problem is a model of water flow in candle filters which are commonly used for purifying drinking water. The domain consists of a circular disc of inner radius of $r_i = a$ and outer radius of $r_o = 1$.    
For the  macro-pore network, the inner surface is subjected to a pressure ($p_{1}(r=r_i)= 1.0~\mathrm{atm}$), and the outer surface is exposed to the atmosphere ($p_{1}(r=r_{o})= 0$). For the micro-pore network, no discharge is allowed from the inner and outer surfaces (i.e. $ \mathbf{u}_{2} \cdot \widehat{\mathbf{n}} = 0 $).
Figure \ref{Fig:2D_problem_candle_filter_discription} shows the computational domain for this problem as well as the boundary conditions. Considering the underlying symmetry in the problem, the velocities and pressures in the two pore-networks are assumed to be functions of $r$ only. Parameter values for this test problem are provided in Table \ref{Tb4:Candle_filter_data}.
{\small
  \begin{table}[!h]
    \caption{Data-set for
      two-dimensional candle filter problem.}
    \centering
    \begin{tabular}{|c|c|} \hline
      Parameter & Value \\
      \hline
      $\gamma \mathbf{b}$ & $\{0.0,0.0\}$\\
      $r_o $ & $1.0$ \\
      $r_i $ & $0.3$ \\
      $\mu $ & $1.0$ \\
      $\beta $ & $1.0$ \\
      $k_1$&  $1.0$ \\
      $k_2$&  $0.01$\\
      $p_1(r=0.3)$ & $1.0$ \\
      $p_1(r=1.0)$ & $0.0$ \\            
      $u_{n2}(r=0.3)$ & $0.0$ \\
      $u_{n2}(r=1.0)$ & $0.0$ \\
      \hline 
    \end{tabular}
    \label{Tb4:Candle_filter_data}
  \end{table}
}

The relevant governing equations
in the polar coordinates can be
summarized as follows:
\begin{subequations}
  \begin{align}
    \label{Eqn:VMS_GE_BLM}
    &\frac{\mu}{k_1} u_{1} + \frac{dp_1}{dr} = 0, \quad  
    \frac{1}{r} \frac{d (r u_{1})}{dr} + (p_1 - p_2) = 0, \quad \forall r \in (a,1) \\
    \label{Eqn:VMS_GE_BofM}
    &\frac{\mu}{k_2} u_{2} + \frac{dp_2}{dr} = 0, \quad  
    \frac{1}{r} \frac{d (r u_{2})}{dr} - (p_1 - p_2) = 0,
    \quad \forall r \in (a,1)  \\
    \label{Eqn:VMS_GE_BCs}
    &p_{1}(r = a) = 1, \quad p_{1}(r = 1) = 0, \quad  
    u_{2}(r = a) = 0, \quad u_{2}(r = 1) = 0 
  \end{align}
  \label{Eqn:Candle_filter_Governing_Eqns}
\end{subequations}
Figures \ref{Fig:2D_problem_candle_filter_p} and \ref{Fig:2D_problem_candle_filter_v} show the pressure and velocity profiles under the extended framework for weak enforcement of velocity boundary conditions.
The micro-velocity profile implies that although there is no discharge
from the micro-pore network on the boundary, there is discharge
in the micro-pore network within the domain. It can be concluded that the surface pore-structure is not the only factor that characterizes the flow throughout the domain and that the internal pore-structure plays a significant role. 

\subsection{Three-dimensional hollow sphere problem}
The hollow sphere problem is used to examine the weak 
enforcement of the velocity boundary conditions in 3D 
settings using Nitsche's method. The computational 
domain consists of a sphere of radius $r_o = 1.0$, 
at the center of which is a spherical hole of radius $r_i = a$. 
At the inner surface of the hole, the macro-pore
network is subjected to a pressure $p_{1}(r=r_{i}) = 1$,
and at the outer surface of the sphere, the macro-pore
network is subjected to a pressure $p_{1}(r=r_{o}) = 0$.
For the micro-pore network, there is no discharge from
the inner and outer surfaces (i.e., $ \mathbf{u}_{2}
\cdot \widehat{\mathbf{n}} = 0 $). Table \ref{Tb5:Hollow_sphere_data} provides the parameter values for this problem.
{\small
	\begin{table}[!h]
		\caption{Data-set for three-dimensional hollow sphere problem.}
		\centering
		\begin{tabular}{|c|c|} \hline
			Parameter & Value \\
			\hline
			$\gamma \mathbf{b}$ & $\{0.0,0.0,0.0\}$\\
			$r_o $ & $1.0$ \\
			$r_i $ & $0.3$ \\
			$\mu $ & $1.0$ \\
			$\beta $ & $1.0$ \\
			$k_1$&  $1.0$ \\
			$k_2$&  $0.01$\\
			$p_1(r=0.3)$ & $1.0$ \\
			$p_1(r=1.0)$ & $0.0$ \\            
			$u_{n2}(r=0.3)$ & $0.0$ \\
			$u_{n2}(r=1.0)$ & $0.0$ \\
			\hline 
		\end{tabular}
		\label{Tb5:Hollow_sphere_data}
	\end{table}
}

Similar to the candle filter problem, all the
variables can be considered to be functions of
$r$ only due to the symmetry. Therefore, the
governing equations can be written as follows:
%
\begin{subequations}
  \begin{align}
    &\frac{\mu}{k_1} u_{1} + \frac{dp_1}{dr} = 0, \quad \frac{1}{r^2} \frac{d (r^2 u_{1})}{dr} = - (p_1 - p_2), \quad \forall r \in (a,1) \\
    &\frac{\mu}{k_2} u_{2} + \frac{dp_2}{dr} = 0,\quad 
    \frac{1}{r^2} \frac{d (r^2 u_{2})}{dr} = + (p_1 - p_2),\quad \forall r \in (a,1) \\ 
    &p_{1}(r = a) = 1, \quad p_{1}(r = 1) = 0, \quad
    u_{2}(r = a) = 0, \quad u_{2}(r = 1) = 0 
  \end{align}
\end{subequations}
The numerical results for the pressures and velocity fields are shown in Figures
\ref{Fig:3D_problem_candle_filter_p} and \ref{Fig:3D_problem_candle_filter_v}.
It is seen that under the extended framework for weak enforcement of velocity boundary conditions, the results are stable and although no discharge is considered
for the micro-pore network on the boundary, there is discharge in the micro-pore
network within the domain. The important role of the internal pore-structure in
such complex porous domains pitches a case for using advanced characterization
tools like X-ray micro-computed tomography (i.e., $\mu$-CT)
\citep{stock2008microcomputed}.

It should be noted that in order to provide a proper visualization of the velocity fields 
within the domain, it has been clipped and put in a perspective 
view. Since an unstructured mesh has been used for this 3D domain, the 
visualization software (i.e., ParaView) cuts through the elements 
and interpolates the values to draw the contours. This introduces 
some dependence on the angular coordinates (see Figure \ref{Fig:3D_problem_candle_filter_v}). 
But one does not find this angular dependence in the 
raw data for various angles for a given radius.

%% file: Sections/S7_VMS_Verification.tex
\section{MECHANICS-BASED ASSESSMENT OF NUMERICAL ACCURACY}
\label{Sec:S7_VMS_Verification}
For all the problems presented in the previous
sections (including the ones under the numerical
convergence analysis), analytical solutions are
known. For such problems, the accuracy of numerical
solutions can be easily quantified by comparing 
them with the analytical solutions (either 
point-wise or in some appropriate norm). But
for practical problems, analytical solutions are
seldom known. The question then will be how to
assess the accuracy of numerical solutions for those
problems with no available analytical solution. 
The source of possible error in the numerical solutions 
could be either due to the formulation itself or in the 
computer implementation. Even if the formulation 
is known to be a converging scheme, there could 
be errors in the computer implementation or in 
setting up the problem to obtain the numerical 
solutions (e.g., wrong input data). 

Fortunately, the solutions under the double porosity/permeability
model enjoy several important mathematical properties, which can
serve as \emph{a posteriori} error measures. More importantly,
these mathematical properties have strong mechanics underpinning
and can be applied to any problem; in particular, they are effective
for those problems without known analytical solutions.
Thus, it is appropriate to refer to such an approach as
\emph{mechanics-based solution verification method}.
Such a study has been undertaken for Darcy and
Darcy-Brinkman equations by
\citep{2016_Shabouei_Nakshatrala_CiCP}. Herein,
we extend the approach to the double porosity/permeability
model and illustrate its utility and performance
to assess the accuracy of numerical solutions
under the proposed stabilized mixed formulation.
However, it needs to be emphasized that the
mechanics-based solution verification method
can be applied to any numerical formulation
(which necessarily need not be based on
the finite element method) and to any
problem. 

Recently, \citep{Nakshatrala_Joodat_Ballarini}
have shown that the exact solutions under the
double porosity/permeability satisfy minimum
dissipation theorem, Betti-type reciprocal
relations and minimum total power theorem.
A numerical solution need not satisfy these mathematical
properties, but the associated errors can be quantified,
which can serve as measures to assess the accuracy of numerical
solutions.
We now utilize the minimum dissipation theorem
and the reciprocal relation to illustrate the
approach to assess the accuracy.

\subsection{\emph{A posteriori} criterion based
  on the minimum dissipation theorem}
Under the double porosity/permeability model,
the dissipation functional takes the following
form \citep{Nakshatrala_Joodat_Ballarini}:
\begin{equation}
  \label{Total_dissipation}
  \boldsymbol{\Phi}\left[\mathbf{u}_{1},\mathbf{u}_{2}\right] 
  := \sum_{i=1}^{2} \left(\int_{\Omega} \mu \mathbf{K}_{i}^{-1} \mathbf{u}_{i}
  \cdot \mathbf{u}_{i} \mathrm{d} \Omega 
  + \dfrac{1}{2}\int_{\Omega}\dfrac{\mu}{\beta}
  \mathrm{div}\left[\mathbf{u}_{i}\right]\mathrm{div}
  \left[\mathbf{u}_{i}\right] \mathrm{d} \Omega \right)
\end{equation}
Under the minimum dissipation theorem, it
is assumed that $\gamma \mathbf{b}
(\mathbf{x})$ is a conservative vector
field and the velocity boundary conditions
are prescribed on the entire boundary
for both pore-networks (i.e.,
$\partial \Omega = \Gamma^{u}_{1}
= \Gamma^{u}_{2}$).
Moreover, a pair of vector fields
$(\widetilde{\mathbf{u}}_1,
\widetilde{\mathbf{u}}_2)$ will
be referred to as kinematically
admissible if it satisfies the
prescribed velocity boundary
conditions and the following
condition:
\begin{align}
  \label{Eqn:DDP_kinematically_admissible}
  \mathrm{div}\left[\widetilde{\mathbf{u}}_{1}\right] 
  + \mathrm{div}\left[\widetilde{\mathbf{u}}_{2}\right] 
  = 0 \quad \mathrm{in} \; \Omega 
\end{align}
Of course, the pair of velocity fields under
the exact solution is kinematically admissible. 
The minimum dissipation theorem states that
the pair of velocity fields under the exact
solution achieves the minimum dissipation
among the set of all kinematically admissible
vector fields \citep{Nakshatrala_Joodat_Ballarini}.

Before we discuss how the minimum dissipation
theorem can be utilized as \emph{a posteriori} 
criterion, it is important to highlight the
following three points regarding the relation
between numerical solutions and this theorem: 
\begin{enumerate}[(i)]
\item A numerical solution need not be
  the minimizer of dissipation
  functional. That is, in the strict
  sense, a numerical solution does
  not satisfy the minimum dissipation
  theorem. 
\item More importantly, the pair
  of velocity fields under a
  numerical solution may not
  even be kinematically
  admissible. 
\item It is not computationally attractive to
  find a numerical solution by solving the
  constrained optimization problem that
  results from the minimum dissipation
  theorem, as such a solution procedure
  will be very expensive; especially, for
  large-scale practical problems.
\end{enumerate}

A description of the proposed \emph{a posteriori}
criterion based on the minimum dissipation theorem
is as follows: Solve the given boundary value
problem under $h$- or $p$-refinements. For each
case of refinement, evaluate the total dissipation
\eqref{Total_dissipation} using
the obtained numerical solution. Plot
the values of the dissipation with
respect to characteristic mesh size
$h$ for the case of $h$-refinement
or degrees-of-freedom in the case
of $p$-refinement. The values of the total
dissipation under the obtained numerical
solutions should decrease monotonically
and reach a plateau upon refinements.
We provide numerical results towards
the end of this section which support
this trend.

A plausible reasoning for the aforesaid
trend can be constructed as follows: Although
the pair of velocity fields under a converging
numerical formulation does not \emph{strictly} 
satisfy the condition 
\eqref{Eqn:DDP_kinematically_admissible}, 
the error in meeting this
condition will be small upon adequate
$h$- or $p$-refinement.
Assuming that the velocity boundary conditions are
accurately implemented, the minimum dissipation
theorem implies that the obtained total dissipation
under the numerical solution should be higher than
the corresponding value under the exact solution.
Moreover, for a converging formulation and
under a proper  computer implementation of
the formulation, a numerical solution should
approach the exact solution upon refinement,
and hence, the values of the total dissipation
should decrease monotonically upon refinement.
But these values are bounded below by the total
dissipation under the exact solution, which
again stems from the minimum dissipation
theorem. The mentioned lower bound will be
the plateau that the values of the total
dissipation under numerical solutions reach.

The above reasoning also reveals that if
the convergence of the total dissipation
is not monotonic with refinement, then one
of the hypotheses of the minimum dissipation
theorem should have been violated. To put
it differently, if the convergence is not
monotonic or there is no convergence at all, one
should suspect that there could be significant
errors in satisfying the local mass balance
condition \eqref{Eqn:DDP_kinematically_admissible}
or in the implementation of boundary conditions.

\subsection{\emph{A posteriori} criterion
  based on reciprocal relations}
Under the reciprocal relation of the double porosity/permeability model, if $(\mathbf{u}_1^{'},p_1^{'},\mathbf{u}_2^{'},p_2^{'})$ and $(\mathbf{u}_1^{*},p_1^{*},\mathbf{u}_2^{*},p_2^{*})$ are, respectively, the \emph{exact solutions} under prescribed data-sets $(\mathbf{b}^{'},u_{n1}^{'},p_{01}^{'},u_{n2}^{'},p_{02}^{'})$ and $(\mathbf{b}^{*},u_{n1}^{*},p_{01}^{*},u_{n2}^{*},p_{02}^{*})$, then the pair of exact solutions and the pair of prescribed data-sets satisfy the following relation \citep{Nakshatrala_Joodat_Ballarini}:
\begin{align}
\label{Eqn:Dual_Reciprocal_term}
\int_{\Omega} \gamma 
\mathbf{b}^{'}(\mathbf{x}) \cdot 
\mathbf{u}_{1}^{*}(\mathbf{x}) \; \mathrm{d} \Omega
- \int_{\Gamma^{p}_{1}} p_{01}^{'}(\mathbf{x}) 
\mathbf{u}_{1}^{*}(\mathbf{x}) \cdot \widehat{\mathbf{n}}(\mathbf{x}) 
\; \mathrm{d} \Gamma 
- \int_{\Gamma^{u}_{1}} p_{1}^{'}(\mathbf{x}) 
u_{n1}^{*}(\mathbf{x}) \; \mathrm{d} \Gamma \nonumber \\
+ \int_{\Omega} \gamma \mathbf{b}^{'}(\mathbf{x}) 
\cdot \mathbf{u}_{2}^{*}(\mathbf{x}) \; \mathrm{d} \Omega 
- \int_{\Gamma^{p}_2} p_{02}^{'}(\mathbf{x}) 
\mathbf{u}_{2}^{*}(\mathbf{x}) \cdot \widehat{\mathbf{n}}(\mathbf{x}) 
\; \mathrm{d} \Gamma 
- \int_{\Gamma^{u}_{2}} p_{2}^{'}(\mathbf{x}) 
u_{n2}^{*}(\mathbf{x}) \; \mathrm{d} \Gamma \nonumber \\
= \int_{\Omega} \gamma 
\mathbf{b}^{*}(\mathbf{x}) \cdot 
\mathbf{u}_{1}^{'}(\mathbf{x}) \; \mathrm{d} \Omega
- \int_{\Gamma^{p}_{1}} p_{01}^{*}(\mathbf{x}) 
\mathbf{u}_{1}^{'}(\mathbf{x}) \cdot \widehat{\mathbf{n}}(\mathbf{x}) 
\; \mathrm{d} \Gamma 
- \int_{\Gamma^{u}_{1}} p_{1}^{*}(\mathbf{x}) 
u_{n1}^{'}(\mathbf{x}) \; \mathrm{d} \Gamma \nonumber \\
+ \int_{\Omega} \gamma \mathbf{b}^{*}(\mathbf{x}) 
\cdot \mathbf{u}_{2}^{'}(\mathbf{x}) \; \mathrm{d} \Omega 
- \int_{\Gamma^{p}_2} p_{02}^{*}(\mathbf{x}) 
\mathbf{u}_{2}^{'}(\mathbf{x}) \cdot \widehat{\mathbf{n}}(\mathbf{x}) 
\; \mathrm{d} \Gamma 
- \int_{\Gamma^{u}_{2}} p_{2}^{*}(\mathbf{x}) 
u_{n2}^{'}(\mathbf{x}) \; \mathrm{d} \Gamma 
\end{align}  
Unlike the minimum dissipation theorem, the reciprocal 
relation does not require the velocity boundary conditions 
to be prescribed on the entire boundary of the two pore-networks. 
However, the domain, $\Omega$, and 
the boundaries, $\Gamma_1^{u}$, $\Gamma_1^{p}$, 
$\Gamma_2^{u}$, and $\Gamma_2^{p}$, are considered to be 
the same for prescribed data-sets. 
Also the reciprocal relation under the double porosity/permeability model 
does not require the set of solutions to be kinematically admissible.

It is important to note that numerical solutions
do not possess reciprocal relations. There will
always be an error under numerical solutions
with respect to the reciprocal relation
\eqref{Eqn:Dual_Reciprocal_term}.
However, this error can be quantified,
and a way to achieve this is by defining
the following scalar quantity, which is 
a form of relative error: 
\begin{equation}
\label{Eqn:VMS_reciprocal_error}
\varepsilon_{\mathrm{reciprocal}} := \frac{|\mathrm{l.h.s~of~Eqn.~\eqref{Eqn:Dual_Reciprocal_term}} - \mathrm{r.h.s~of~Eqn.~ \eqref{Eqn:Dual_Reciprocal_term}}|}{\mathrm{l.h.s~of~Eqn.~ \eqref{Eqn:Dual_Reciprocal_term}}}
\end{equation}
For exact solutions, we will have
$\varepsilon_{\mathrm{reciprocal}} = 0$.
For those problems in which left hand
side of equation
\eqref{Eqn:Dual_Reciprocal_term} vanishes,
one can use an absolute error measure
instead of this relative error measure.
Thus, the magnitude of
$\varepsilon_{\mathrm{reciprocal}}$ will
serve as a measure to assess the accuracy
of a numerical formulation. 

A description of the proposed \emph{a
  posteriori} criterion based on the
reciprocal relation is as follows: Solve
the given boundary value problem under
$h$- or $p$-refinements. For each case
of refinement, evaluate the relative
error $\varepsilon_{\mathrm{reciprocal}}$
using the obtained numerical solution.
Plot the values of $\varepsilon_{\mathrm{reciprocal}}$
with respect to characteristic mesh size
$h$ for the case of $h$-refinement
or degrees-of-freedom in the case
of $p$-refinement. The values of 
$\varepsilon_{\mathrm{reciprocal}}$ under 
the obtained numerical solutions should 
decrease monotonically and reach a plateau 
upon refinements. Similar to the case of
\emph{a posteriori} criterion based on the
minimum dissipation theorem, the
numerical results provided at
the end of this section support
this trend.

One can construct a plausible reasoning
for the mentioned trend in
$\varepsilon_{\mathrm{reciprocal}}$ 
similar to the reasoning
provided under the minimum dissipation
theorem. 
Since the reciprocal relation does not
require the velocity fields to be kinematically
admissible (specifically, the velocity fields
need not satisfy the local mass balance condition
\eqref{Eqn:DDP_kinematically_admissible}), it
is reasonable to conclude that 
if $\varepsilon_{\mathrm{reciprocal}}$ does not decrease
monotonically with refinement, then one
should suspect that there could be significant
errors in the implementation of boundary conditions.

\subsection{Representative numerical results}
To illustrate the performance and utility
of the mentioned mechanics-based \emph{a
  posteriori} criteria, we employ the
pipe bend problem, which is widely used
as a benchmark problem for flow through
porous media \citep{Challis_Guest_2009,
  Aage_Poulsen_Hansen_Sigmund_2008,
  Borrvall_Petersson_2003}.
A pictorial description of the problem is shown
in Figure \ref{Fig:pipe_bend_problem_domain}.
The computational domain is a unit square
($L=1.0$). For the velocity boundary conditions, two different cases are considered. For the macro-pore network in case 1, an inflow parabolic velocity is enforced on a portion of the left boundary (denoted as $\Gamma^{u}_{\mathrm{inflow}}$) while an outflow parabolic velocity is applied on a portion of the bottom boundary (denoted as $\Gamma^{u}_{\mathrm{outflow}}$). In case 2, an inflow constant velocity is enforced on $\Gamma^{u}_{\mathrm{inflow}}$ while an outflow constant velocity is applied on $\Gamma^{u}_{\mathrm{outflow}}$ for the macro-pore network. For both cases, the normal component of macro-velocity is prescribed to be zero on the rest of the boundary (i.e., $u_{n}(\mathbf{x}) = 0$). The normal component of micro-velocity in both data-sets ($u^{'}_{n2}$ and $u^{*}_{n2}$) is zero in the data-sets on the entire boundary. These sample data-sets are provided in Table \ref{Tbl:pipe_bend_reciprocal_error1}. 
{\small
	\begin{table}
		\caption{Data-sets for the pipe bend problem.}
		\centering
		\begin{tabular}{|c|c|} \hline
			Case 1 & Case 2 \\
			\hline
			$L^{'}_x = 1.0 $ & $L^{*}_x = 1.0$ \\
			$L^{'}_y = 1.0 $ & $L^{*}_y = 1.0$ \\
			$\mu^{'} = 1.0 $ & $\mu^{*} = 1.0$ \\
			$\beta^{'} = 1.0 $ & $\beta^{*} = 1.0$ \\
			$k^{'}_1 = 1.0$&  $k^{*}_1 = 1.0$ \\
			$k^{'}_2 = 0.01$&  $k^{*}_2 = 0.01$\\
			$\gamma \mathbf{b}^{'} = \left\{1.0,1.0\right\} $ & $\gamma \mathbf{b}^{*} = \left\{0.0,0.0\right\}$ \\
			$u^{'}_{n1} = 100 (y - 0.6) (0.8 - y)$ on $\Gamma^{u}_{\mathrm{inflow}}$&  $u^{*}_{n1} = 1.0 $ on $\Gamma^{u}_{\mathrm{inflow}}$ \\
			$u^{'}_{n1} = -100 (x - 0.6) (0.8 - x)$ on $\Gamma^{u}_{\mathrm{outflow}}$&  $u^{*}_{n1} = -1.0$ on $\Gamma^{u}_{\mathrm{outflow}}$\\
			$u^{'}_{n1} = 0$ on the other parts of $\partial \Omega$ &  $u^{*}_{n1} = 0$ on the other parts of $\partial \Omega$ \\
			$u^{'}_{n2} = 0$ on $\partial \Omega$ & $u^{*}_{n2} = 0$ on $\partial \Omega$\\
			\hline 
		\end{tabular}
		\label{Tbl:pipe_bend_reciprocal_error1}
	\end{table}
}

Figure \ref{Fig:pipe_bend_problem_Dissipation} shows how the deviation in dissipation varies with mesh refinement for the numerical solutions obtained using both data-sets. Under $h$-refinement, as the mesh size $h$ decreases (or the total number of the elements increases), the deviation in the dissipation value decreases for both cases and the convergence is monotonic. 
This deviation can be further quantified using $\varepsilon_{\mathrm{reciprocal}}$ under the double porosity/permeability model for the sample data-sets as shown in Figure \ref{Fig:Pipebend_problem_reciprocal_error}. For different orders of interpolation, the error in reciprocal relation for the two sets of numerical solutions decreases monotonically with mesh refinement for this test problem which implies that the numerical solutions converge monotonically. As can be seen, by increasing the order of interpolation for the primary variables, the value of error is decreased and the numerical solutions get closer to the exact solutions of the model.

%% file: Sections/S8_VMS_Transient.tex
\section{AN EXTENSION TO TRANSIENT ANALYSIS}
\label{Sec:S8_VMS_Transient}
The discussions and the results presented in the previous
sections neglected transient flow behavior within the
porous domain. However, unsteady flow characteristics
are indispensable in a wide variety of applications
such as the ones observed in aquifers and oil-bearing
strata \citep{Mongan_1985}, and composite manufacturing
applications based on resin transfer molding
\citep{Nakshatrala_Turner_Hjelmstad_Masud_CMAME_2006_v195_p4036,
  Pacquaut_Bruchon_Moulin_Drapier_2012} where two different
fibers are usually used, providing two different pathways
for the fluid.
In this section, the proposed mixed formulation is
extended to the transient case. We first document
the governing equations in a transient setting,
which will have an unsteady term in the balance
of momentum equation for each pore-network. A
stabilized mixed formulation is then derived for
the transient case. Finally, the performance of
the proposed formulation in the transient case
will be illustrated using a representative example.

\subsection{Unsteady governing equations}
Same as before, we consider a bounded domain, $\Omega \subset \mathbb{R}^{nd}$, with a piecewise smooth boundary denoted by $\partial \Omega$. The time is denoted by $t \in \left[0, T\right]$, where $T$ is the total time of interest. Darcy velocity (vector) fields in macro- and micro-pores at any spatial point $\mathbf{x}$ are denoted by $\mathbf{u}_{1}(\mathbf{x}, t)$ and $\mathbf{u}_{2}(\mathbf{x}, t)$ respectively, while macro- and micro-pressure (scalar) fields are denoted by $p_{1}(\mathbf{x}, t)$, and  $p_{2}(\mathbf{x}, t)$. The specific body force can also depend on time and is denoted by $\mathbf{b}(\mathbf{x}, t)$. Assuming that the porosities in the two pore-networks do not change with time, the transient
governing equations can be written as follows:
\begin{subequations}
  \begin{alignat}{2}
    \label{Eqn:Dual_GE_BLM_1_Transient}
    &\rho_{1} \frac{\partial \mathbf{u}_{1}}{\partial t} + \mu \mathbf{K}_{1}^{-1} \mathbf{u}_1 + \mathrm{grad}[p_1] = \gamma \mathbf{b}
    &&\quad \mathrm{in} \; \Omega \times \left(0,T\right) \\
    \label{Eqn:Dual_GE_BLM_2_Transient}
    &\rho_{2} \frac{\partial \mathbf{u}_{2}}{\partial t} + \mu \mathbf{K}_{2}^{-1} \mathbf{u}_2 + \mathrm{grad}[p_2] = \gamma \mathbf{b}
    &&\quad \mathrm{in} \;  \Omega \times \left(0,T\right) \\
    \label{Eqn:Dual_GE_mass_balance_1_Transient}
    &\mathrm{div}[\mathbf{u}_1] = +\chi
    &&\quad \mathrm{in} \;  \Omega \times \left(0,T\right) \\
    \label{Eqn:Dual_GE_mass_balance_2_Transient}
    &\mathrm{div}[\mathbf{u}_2] = -\chi
    &&\quad \mathrm{in} \;  \Omega \times \left(0,T\right) \\
    \label{Eqn:Dual_GE_vBC_1_Transient}
    &\mathbf{u}_1(\mathbf{x},t) \cdot
    \widehat{\mathbf{n}}(\mathbf{x})
    = u_{n1}(\mathbf{x},t)
    &&\quad \mathrm{on} \; \Gamma^{u}_{1} \times \left(0,T\right) \\
    \label{Eqn:Dual_GE_vBC_2_Transient}
    &\mathbf{u}_2(\mathbf{x},t) \cdot
    \widehat{\mathbf{n}}(\mathbf{x})
    = u_{n2}(\mathbf{x},t)
    &&\quad \mathrm{on} \; \Gamma^{u}_{2} \times \left(0,T\right) \\
    \label{Eqn:Dual_GE_Darcy_pBC_1_Transient}
    &p_1(\mathbf{x},t) = p_{01} (\mathbf{x},t)
    &&\quad \mathrm{on} \; \Gamma^{p}_{1} \times \left(0,T\right)\\
    \label{Eqn:Dual_GE_Darcy_pBC_2_Transient}
    &p_2(\mathbf{x},t) = p_{02} (\mathbf{x},t)
    &&\quad \mathrm{on} \; \Gamma^{p}_{2} \times \left(0,T\right)\\
    \label{Eqn:Dual_GE_Darcy_iBC_1}
    &\mathbf{u}_1(\mathbf{x},0) = \mathbf{u}_{01}(\mathbf{x})
    &&\quad \mathrm{in} \; \Omega \\
    \label{Eqn:Dual_GE_Darcy_iBC_2}
    &\mathbf{u}_2(\mathbf{x},0) = \mathbf{u}_{02}(\mathbf{x})
    &&\quad \mathrm{in} \; \Omega
  \end{alignat}
\end{subequations}
where $\mathbf{u}_{01}(\mathbf{x})$ and $\mathbf{u}_{02}
(\mathbf{x})$ are the prescribed initial velocities
within the domain. The definitions for the other
symbols remain the same as before. It is understood
that the quantities corresponding to these symbols
will now depend on the time, expect for the unit
outward normal, as the domain is fixed and does
not evolve with respect to the time. 
We now derive a stabilized formulation for the
mentioned transient governing equations under
the double porosity/permeability model.

\subsection{A stabilized mixed formulation for the transient case}
We employ the method of horizontal lines (also known as the
Rothe's method) \citep{Rothe_1930}, which is different from
the semi-discrete method (also known as the method of vertical
lines) \citep{Hughes_1987finite}. Under the method of horizontal
lines, a given partial differential equation (which depends
on both space and time) is discretized temporally using a
time-stepping scheme. This gives rise to another partial
differential equation which depends only on the spatial
coordinates, and can be further discretized spatially
using the finite element method, the finite difference
method or the finite volume method.
On the other hand, under the semi-discrete method, the
given spatially and temporally dependent partial
differential equation is first spatially discretized,
say, using the finite element method, giving rise to
a system of ordinary differential equations; which
can be numerically solved by employing a convenient
time-stepping scheme.

Herein, we employ the backward Euler time stepping scheme
for the temporal discretization of the transient governing
equations under the method of horizontal lines. However,
one can employ any other time-stepping scheme with a
straightforward modification. The backward Euler is
first-order accurate and unconditionally stable when
applied to a linear system of ordinary differential
equations \citep{Hairer_Norsett_Wanner}.
The time interval of interest is discretized into $N+1$ time levels denoted as $t_n$ ($n = 0, \cdots ,N$) by assuming uniform time steps ($\Delta t = t_n−t_{n−1}$); however, one can consider non-uniform time steps by applying simple modifications. For a given quantity $z(\mathbf{x}, t)$, the time discretized version at the instant of time $t_n$ can be written as follows:
\begin{align}
  z^{(n)}(\mathbf{x}) \approx z(\mathbf{x}, t_n),
  \quad n = 0, \cdots, N
\end{align}

The resulting time discretized equations at time
level $t = t_{n+1}$ under the method of horizontal
lines using the backward Euler time-stepping
scheme take the following form: 
\begin{subequations}
  \begin{alignat}{2}
    \label{Eqn:Dual_GE_BLM_1_Transient_Discretized}
    &\rho_{1} \frac{\mathbf{u}^{(n+1)}_{1}-\mathbf{u}^{(n)}_{1}}{\Delta t} + \mu \mathbf{K}_{1}^{-1} \mathbf{u}^{(n+1)}_1 + \mathrm{grad}[p^{(n+1)}_1] = \gamma \mathbf{b}^{(n+1)}
    &&\quad \mathrm{in} \; \Omega \\
    \label{Eqn:Dual_GE_BLM_2_Transient_Discretized}
    &\rho_{2} \frac{\mathbf{u}^{(n+1)}_{2}-\mathbf{u}^{(n)}_{2}}{\Delta t} + \mu \mathbf{K}_{2}^{-1} \mathbf{u}^{(n+1)}_2 + \mathrm{grad}[p^{(n+1)}_2] = \gamma \mathbf{b}^{(n+1)}
    &&\quad \mathrm{in} \;  \Omega \\
    \label{Eqn:Dual_GE_mass_balance_1_Transient_Discretized}
    &\mathrm{div}[\mathbf{u}^{(n+1)}_1] = +\chi^{(n+1)}
    &&\quad \mathrm{in} \;  \Omega \\
    \label{Eqn:Dual_GE_mass_balance_2_Transient_Discretized}
    &\mathrm{div}[\mathbf{u}^{(n+1)}_2] = -\chi^{(n+1)}
    &&\quad \mathrm{in} \;  \Omega \\
    \label{Eqn:Dual_GE_vBC_1_Transient_Discretized}
    &\mathbf{u}^{(n+1)}_1(\mathbf{x}) \cdot
    \widehat{\mathbf{n}}(\mathbf{x})
    = u_{n1}(\mathbf{x},t=t_{n+1})
    &&\quad \mathrm{on} \; \Gamma^{u}_{1} \\
    \label{Eqn:Dual_GE_vBC_2_Transient_Discretized}
    &\mathbf{u}^{(n+1)}_2(\mathbf{x}) \cdot
    \widehat{\mathbf{n}}(\mathbf{x})
    = u_{n2}(\mathbf{x},t=t_{n+1})
    &&\quad \mathrm{on} \; \Gamma^{u}_{2} \\
    \label{Eqn:Dual_GE_Darcy_pBC_1_Discretized}
    &p^{(n+1)}_1(\mathbf{x}) = p^{(n+1)}_{01} (\mathbf{x}) = p_{01} (\mathbf{x},t=t_{n+1}) 
    &&\quad \mathrm{on} \; \Gamma^{p}_{1}\\
    \label{Eqn:Dual_GE_Darcy_pBC_2_Discretized}
    &p^{(n+1)}_2(\mathbf{x}) = p^{(n+1)}_{02} (\mathbf{x}) = p_{02} (\mathbf{x},t=t_{n+1})
    &&\quad \mathrm{on} \; \Gamma^{p}_{2}\\
    \label{Eqn:Dual_GE_Darcy_iBC_1_Discretized}
    &\mathbf{u}^{(0)}_1(\mathbf{x}) = \mathbf{u}_{01}(\mathbf{x})
    &&\quad \mathrm{in} \; \Omega \\
    \label{Eqn:Dual_GE_Darcy_iBC_2_Discretized}
    &\mathbf{u}^{(0)}_2(\mathbf{x}) = \mathbf{u}_{02}(\mathbf{x})
    &&\quad \mathrm{in} \; \Omega
  \end{alignat}
\end{subequations}
Equations \eqref{Eqn:Dual_GE_BLM_1_Transient_Discretized} and \eqref{Eqn:Dual_GE_BLM_2_Transient_Discretized} can be rearranged as follows:
\begin{subequations}
  \begin{alignat}{2}
    &\left(\frac{\rho_{1}}{\Delta t}\mathbf{I} + \mu \mathbf{K}_{1}^{-1} \right) \mathbf{u}^{(n+1)}_{1} + \mathrm{grad}[p^{(n+1)}_1] = \gamma \left( \mathbf{b}^{(n+1)} + \frac{\phi_1}{\Delta t} \mathbf{u}^{(n)}_{1}\right)
    &&\quad \mathrm{in} \; \Omega \\
    &\left(\frac{\rho_{2}}{\Delta t} \mathbf{I}  + \mu \mathbf{K}_{2}^{-1} \right) \mathbf{u}^{(n+1)}_{2} + \mathrm{grad}[p^{(n+1)}_2] = \gamma \left( \mathbf{b}^{(n+1)} + \frac{\phi_2}{\Delta t} \mathbf{u}^{(n)}_{2}\right)
    &&\quad \mathrm{in} \; \Omega 
  \end{alignat}
\end{subequations}
where the (modified) drag coefficients and (modified) body forces can be written as follows:
\begin{subequations}
  \begin{alignat}{2}
&\widehat{\alpha}_{1} = \left(\frac{\rho_{1}}{\Delta t}\mathbf{I} + \mu \mathbf{K}_{1}^{-1} \right) \\
&\widehat{\alpha}_{2} = \left(\frac{\rho_{2}}{\Delta t} \mathbf{I}  + \mu \mathbf{K}_{2}^{-1} \right) \\
&\tilde{\mathbf{b}}^{(n+1)}_{1} = \left( \mathbf{b}^{(n+1)} + \frac{\phi_1}{\Delta t} \mathbf{u}^{(n)}_{1}\right) \\
&\tilde{\mathbf{b}}^{(n+1)}_{2} = \left( \mathbf{b}^{(n+1)} + \frac{\phi_2}{\Delta t} \mathbf{u}^{(n)}_{2}\right)
  \end{alignat}
\end{subequations}
where $\phi_1$ and $\phi_2$ are, respectively, the volume fractions associated with the two pore-networks, relating the bulk density and the true density as $\rho_{i} = \phi_{i} \gamma~(i=1,2)$. 

The stabilized mixed formulation for the unsteady condition at time level $t = t_{n+1}$ reads as:~Find $\left(\mathbf{u}^{(n+1)}_1(\mathbf{x}),
\mathbf{u}^{(n+1)}_2(\mathbf{x}) \right) \in \bar{\mathcal{U}}_{1,t=t_{n+1}}
\times \bar{\mathcal{U}}_{2,t=t_{n+1}}$, and $\left(p^{(n+1)}_1(\mathbf{x}),
p^{(n+1)}_2(\mathbf{x})\right) \in \bar{\mathcal{Q}}$ 
such that we have
\begin{align}
  \mathcal{B}_{\mathrm{stab}}(\mathbf{w}_1,\mathbf{w}_2,q_1,q_2;
  \mathbf{u}^{(n+1)}_1,\mathbf{u}^{(n+1)}_2,p^{(n+1)}_1,p^{(n+1)}_2)
  = \mathcal{L}_{\mathrm{stab}}^{\mathrm{tran}}(\mathbf{w}_1,\mathbf{w}_2,q_1,q_2)
  \nonumber \\
  \quad \forall
  \left(\mathbf{w}_1(\mathbf{x}), \mathbf{w}_2(\mathbf{x})\right) ~
  \in \bar{\mathcal{W}}_1 \times \bar{\mathcal{W}}_2,
  \left(q_1(\mathbf{x}),q_2(\mathbf{x})\right)
  \in \bar{\mathcal{Q}}
\label{Eqn:VMS_Weak_Form_transient} 
\end{align}
The linear functional under the transient condition $\mathcal{L}_{\mathrm{stab}}^{\mathrm{tran}}(\mathbf{w}_1,\mathbf{w}_2,q_1,q_2)$ is slightly different from the one under the steady-state condition. Under the steady-state condition, the body forces, denoted by $\mathbf{b}(\mathbf{x})$, are similar in both pore-networks. However, for the transient case, body forces in the macro- and micro-pore networks ($\tilde{\mathbf{b}}^{(n+1)}_{1}$ and $\tilde{\mathbf{b}}^{(n+1)}_{2}$) are different and one should substitute $\mathbf{b}(\mathbf{x})$ in the $\mathcal{L}_{\mathrm{stab}}(\mathbf{w}_1,\mathbf{w}_2,q_1,q_2)$ by the associated value in each pore-network in order to obtain $\mathcal{L}_{\mathrm{stab}}^{\mathrm{tran}}(\mathbf{w}_1,\mathbf{w}_2,q_1,q_2)$. It should also be noted that in the bilinear form and the linear functional of the proposed formulation provided in equations \eqref{Eqn:Dual_B_VMS} and \eqref{Eqn:Dual_L_VMS}, $\mu \mathbf{K}_i^{-1}$ should be replaced by $\widehat{\alpha}_{i},~(i=1,2)$. 

A systematic numerical implementation of the proposed formulation is outlined in Algorithm 1. It should be noted that, we need not evaluate all the terms in the variational form, especially the terms in $\mathcal{B}_{\mathrm{stab}}$, at each time step since most of them do not depend on the temporal variable. Therefore, it is enough to only evaluate the terms involving $\tilde{\mathbf{b}}^{(n+1)}_{i},~i=1,2$ in the $\mathcal{L}_{\mathrm{stab}}^{\mathrm{tran}}$ repeatedly.
{\small
\begin{table}[h!]
  \label{tab:table1}
  \begin{tabular}{l}
  \hline 
  \textbf{Algorithm 1} Implementation of the proposed formulation.\\
  \hline
  1:  Inputs: Initial conditions $\mathbf{u}_{01}$
  and $\mathbf{u}_{02}$, time period of integration
  $T$, maximum allowable time step $\Delta t_{\mathrm{max}}$\\
  2: Set $\mathbf{u}^{(n)}_{1} = \mathbf{u}_{01}$ and
  $\mathbf{u}^{(n)}_{2} = \mathbf{u}_{02}$\\
  3: Set $t=0$ \\
  4: \textbf{while} $t<T$ \textbf{do}\\
  5: $\quad \Delta t = \mathrm{min}[\Delta t_{\mathrm{max}}, T-t],~t = t+ \Delta t$\\
    6: $\quad$Using $\mathbf{u}^{(n)}_1$ and $\mathbf{u}^{(n)}_2$, solve equation \eqref{Eqn:VMS_Weak_Form_transient} to obtain $\mathbf{u}^{(n+1)}_1$, $\mathbf{u}^{(n+1)}_2$, $p^{(n+1)}_1$, and $p^{(n+1)}_2$\\
    7: $\quad$Set $\mathbf{u}^{(n)}_{1}= \mathbf{u}^{(n+1)}_1$ and $\mathbf{u}^{(n)}_{2}= \mathbf{u}^{(n+1)}_2$\\
    8: \textbf{end while}\\
    \hline
  \end{tabular}
\end{table}
}

The relevant function spaces for the velocity
and pressure fields and their corresponding
weighting functions under the transient case
are defined as follows:
\begin{subequations}
\begin{align}
\label{Eqn:VMS_Function_space_U1_Transient}
\bar{\mathcal{U}}_{1,t} &:= 
\left\{\mathbf{u}_{1}(\mathbf{x}) \in 
\left(L_{2}(\Omega)\right)^{nd} 
\; \Big\vert \;
\mathrm{div}[\mathbf{u}_{1}] \in L_{2}(\Omega), 
\mathbf{u}_{1}(\mathbf{x}) \cdot \widehat{\mathbf{n}}(\mathbf{x}) 
= u_{n1}(\mathbf{x},t) \; \mathrm{on} \; \Gamma_{1}^{u}\right\} \\
\bar{\mathcal{U}}_{2,t} &:= 
\left\{\mathbf{u}_{2}(\mathbf{x}) \in 
\left(L_{2}(\Omega)\right)^{nd} 
\; \Big\vert \;
\mathrm{div}[\mathbf{u}_{2}] \in L_{2}(\Omega), 
\mathbf{u}_{2}(\mathbf{x}) \cdot \widehat{\mathbf{n}}(\mathbf{x}) 
= u_{n2}(\mathbf{x},t) \; \mathrm{on} \; \Gamma_{2}^{u}\right\} \\
\bar{\mathcal{W}}_{1} &:= 
\left\{\mathbf{w}_{1}(\mathbf{x}) \in 
\left(L_{2}(\Omega)\right)^{nd} 
\; \Big\vert \;
\mathrm{div}[\mathbf{w}_{1}] \in L_{2}(\Omega), 
\mathbf{w}_{1}(\mathbf{x}) \cdot \widehat{\mathbf{n}}(\mathbf{x}) 
= 0 \; \mathrm{on} \; \Gamma_{1}^{u} \right\} \\
\bar{\mathcal{W}}_{2} &:= 
\left\{\mathbf{w}_{2}(\mathbf{x}) \in 
\left(L_{2}(\Omega)\right)^{nd} 
\; \Big\vert \;
\mathrm{div}[\mathbf{w}_{2}] \in L_{2}(\Omega), 
\mathbf{w}_{2}(\mathbf{x}) \cdot \widehat{\mathbf{n}}(\mathbf{x}) 
= 0 \; \mathrm{on} \; \Gamma_{2}^{u}\right\} \\
%
\label{Eqn:VMS_Function_space_Q_Transient}
\bar{\mathcal{Q}} &:= 
\left\{(p_1(\mathbf{x}),p_{2}(\mathbf{x})) 
\in H^{1}(\Omega) \times H^{1}(\Omega)
\; \Big\vert \;
\left(\int_{\Omega} p_{1}(\mathbf{x}) \mathrm{d} \Omega \right) 
\left(\int_{\Omega} p_{2}(\mathbf{x}) \mathrm{d} \Omega \right) 
= 0 \right\}
\end{align}
\end{subequations}

\subsection{A representative numerical example}
We now illustrate the performance of the proposed
stabilized mixed formulation for studying transient
flow problems using a two-dimensional problem.
Moreover, some unique features of flows
in porous media exhibiting two distinct
pore-networks are illustrated.

The computational domain $\Omega$ is chosen to be the region in-between a rectangle of length 10.0 and height 1.0 and two square holes each of length 0.4. Zero-flux boundary conditions for both macro-pore and micro-pore networks are prescribed at the holes as well as top and bottom edges of the rectangular domain. At the right end, pressure is prescribed at both pore-networks. At the left end, however, zero-flux boundary condition is prescribed for the micro-pore network and pressure is prescribed for the macro-pore network. The initial velocities for both fluid constituents are assumed to be zero. A pictorial description of the domain as well as the initial and boundary conditions are illustrated in Figure \ref{Fig:Flow_Transport}. Table \ref{Tb8:2D_transient_data} provides the parameter values for the two-dimensional transient flow problem.
{\small
	\begin{table}[!h]
		\caption{Data-set for two-dimensional transient flow problem.}
		\centering
		\begin{tabular}{|c|c|} \hline
			Parameter & Value \\
			\hline
			$\gamma \mathbf{b}$ & $\{ 0.0,0.0\}$\\
			$L_x $ & $10.0$ \\
			$L_y $ & $1.0$ \\
			$L_{\mathrm{hole}}$ & $0.4$\\
			$\mu $ & $1.0$ \\
			$\beta $ & $1.0$ \\
			$k_1$&  $10000$ \\
			$k_2$&  $1.0$\\
			$p_1^L$ & $10~\mathrm{sin}(0.4 ( y + 2.0t))$\\
			$p_1^R$ & $10.0$ \\  
			$p_2^R$ & $10.0$ \\           
			$u_{n2}^L$ & $0.0$ \\
			$u_{n2}^R$ & $0.0$ \\
			$u_{01}$ & $0.0$ \\
			$u_{02}$ & $0.0$ \\
			$\Delta t$ & $5\mathrm{e}-11$\\
			$T$ & $6\mathrm{e}-8$\\
			\hline 
		\end{tabular}
		\label{Tb8:2D_transient_data}
	\end{table}
}

Figure \ref{Fig:V1_vs_V2_Transport} shows
a comparison between macro-velocity
($\mathbf{u}_1$) and micro-velocity
($\mathbf{u}_2$) at selected time steps.
\emph{As can be seen in this figure, the
rate of decay of the solution in the
macro-pore network (which has a higher
permeability than the micro-pore network)
is slower than that of the micro-pore
network, and hence, the micro-velocity
reaches the steady-state faster than
the macro-velocity.}
  This is not counterintuitive if one realizes
  that the rate of dissipation in a pore-network
  is inversely proportional to the permeability
  of the pore-network. Specifically, the
  rates of dissipation in macro- and micro-pore
  networks under the double porosity/permeability
  model are, respectively, defined as follows
  \citep{Nakshatrala_Joodat_Ballarini}:
  \[
  \mu \mathbf{K}_{1}^{-1} \mathbf{u}_1 \cdot
  \mathbf{u}_1
  \quad \mathrm{and} \quad
  \mu \mathbf{K}_{2}^{-1} \mathbf{u}_2 \cdot
  \mathbf{u}_2
  \]
  It is also noteworthy to recall the definition
  of permeability of a porous medium, which is
  a measure of the ability of the porous medium
  to transmit fluids through it.
  To put it differently, the greater is the
  permeability the lesser will be resistance
  offered by the pore-network, and hence the
  greater will be the ease with which a fluid
  flows through the pore-network.

%% file: Sections/S9_VMS_Coupled.tex
\section{COUPLED PROBLEMS}
\label{Sec:S8_2_VMS_Coupled}

Experimental studies on Darcy flow coupled with transport problem have revealed the possibility of occurrence  of certain physical instabilities called Saffman-Taylor
instability \citep{Saffman_1958,Chuoke_VanMeurs_vanderPoel_1959}. In the miscible displacement of fluids in porous media with a single pore-network, a more viscous fluid is displaced by a less viscous fluid within the domain \citep{Stalkup_1983, Homsy_1987}. Imposing any disturbance or perturbation on the interface of the two fluids leads to appearance of finger-like patterns at the interface of the two fluids due to the penetration of the less viscous fluid into the more viscous one. 
This type of physical instability, which is commonly observed in a wide variety of industrial and environmental applications such as carbon-dioxide sequestration and secondary and tertiary oil recovery, is also referred to as viscous fingering (VF) instability \citep{Homsy_1987,Chen_Meiburg_1998a, Chen_Meiburg_1998b}. 

All the existing theoretical and numerical studies in the literature are available for the \emph{classical} 
Saffman-Taylor instability. The questions remaining are whether similar physical instabilities can be captured under the double porosity/permeability model and if so, how the flow model can affect the mechanism of the instabilities and their characteristics (i.e., number of fingers, their characteristic length, growth rate, scaling laws, etc.). 
Herein, we cannot provide an exhaustive study on such well-known instabilities in fluid mechanics and many important areas of research associated with viscous fingering are
not included in our discussion. Therefore, we only address the former question by studying the possibility of occurrence of Saffman-Taylor-type instabilities under the double porosity/permeability model. 
The proposed stabilized formulation will be employed for modeling double porosity/permeability model coupled with transport problem to illustrate the capability of the proposed computational framework for capturing Saffman-Taylor-type instabilities within a porous domain exhibiting double pore-networks. However, studying the effects of the flow model (double porosity/permeability model versus Darcy model) on the mode and patterns of the instabilities is beyond the scope of this paper and will be addressed in a separate paper.

\subsection{Governing equations: Coupled flow and transport problem}
Viscous fingering can be considered as a two-way coupled flow and transport problem and is studied in the Hele-Shaw cell. 
The governing equations can be written as follows:
\begin{subequations}
  \begin{align}
    \label{Eqn:BLM_1_Transient}
    &\mu \mathbf{K}_{1}^{-1} \mathbf{u}_1 (\mathbf{x},t) + \mathrm{grad}[p_1 (\mathbf{x},t)] = \gamma \mathbf{b}(\mathbf{x},t)
    &&\quad \mathrm{in} \; \Omega \times \left(0,T\right) \\
    \label{Eqn:BLM_2_Transient}
    &\mu \mathbf{K}_{2}^{-1} \mathbf{u}_2 (\mathbf{x},t) + \mathrm{grad}[p_2 (\mathbf{x},t)] = \gamma \mathbf{b}(\mathbf{x},t)
    &&\quad \mathrm{in} \;  \Omega \times \left(0,T\right) \\
    \label{Eqn:mass_balance_1_Transient}
    &\mathrm{div}[\mathbf{u}_1 (\mathbf{x},t)] = +\chi(\mathbf{x},t)
    &&\quad \mathrm{in} \;  \Omega \times \left(0,T\right) \\
    \label{Eqn:mass_balance_2_Transient}
    &\mathrm{div}[\mathbf{u}_2 (\mathbf{x},t)] = -\chi(\mathbf{x},t)
    &&\quad \mathrm{in} \;  \Omega \times \left(0,T\right) \\
    \label{Eqn:vBC_1_Transient}
    &\mathbf{u}_1(\mathbf{x},t) \cdot
    \widehat{\mathbf{n}}(\mathbf{x})
    = u_{n1}(\mathbf{x},t)
    &&\quad \mathrm{on} \; \Gamma^{u}_{1} \times \left(0,T\right) \\
    \label{Eqn:vBC_2_Transient}
    &\mathbf{u}_2(\mathbf{x},t) \cdot
    \widehat{\mathbf{n}}(\mathbf{x})
    = u_{n2}(\mathbf{x},t)
    &&\quad \mathrm{on} \; \Gamma^{u}_{2} \times \left(0,T\right) \\
    \label{Eqn:pBC_1_Transient}
    &p_1(\mathbf{x},t) = p_{01} (\mathbf{x},t)
    &&\quad \mathrm{on} \; \Gamma^{p}_{1} \times \left(0,T\right)\\
    \label{Eqn:pBC_2_Transient}
    &p_2(\mathbf{x},t) = p_{02} (\mathbf{x},t)
    &&\quad \mathrm{on} \; \Gamma^{p}_{2} \times \left(0,T\right)\\
    \label{Eqn:Transport_GE}
    &\frac{\partial c(\mathbf{x},t)}{\partial t} + \mathrm{div}\left[\mathbf{u}(\mathbf{x},t) c(\mathbf{x},t)- \mathbf{D}(\mathbf{x},t) \mathrm{grad}[c(\mathbf{x},t)]\right] = f(\mathbf{x},t)
    &&\quad \mathrm{in} \; \Omega \times \left(0,T\right) \\
    \label{Eqn:Transport_Dirichlet_BC}
    &c(\mathbf{x},t) = c^{p}(\mathbf{x},t)
    &&\quad \mathrm{on} \; \Gamma^{D} \times \left(0,T\right) \\
    \label{Eqn:Transport_Neumann_BC}
    &\widehat{\mathbf{n}} (\mathbf{x}) \cdot \left(\mathbf{u}(\mathbf{x},t) c(\mathbf{x},t) - \mathbf{D}(\mathbf{x},t) \mathrm{grad}[c(\mathbf{x},t)]\right) = q^{p} (\mathbf{x},t)
    &&\quad \mathrm{on} \; \Gamma^{N} \times \left(0,T\right) \\
    \label{Eqn:Transport_IC}
    &c(\mathbf{x},t=0) = c_{0}(\mathbf{x})
    &&\quad \mathrm{in} \; \Omega
  \end{align}
\end{subequations}
where equations \eqref{Eqn:BLM_1_Transient} -- \eqref{Eqn:pBC_2_Transient} represent the flow equations under the double porosity/permeability model, and equations \eqref{Eqn:Transport_GE} --
\eqref{Eqn:Transport_IC} represent the transient advection-diffusion problem. Herein, $c(\mathbf{x},t)$ denotes the concentration and $\mathbf{D}(\mathbf{x},t)$ is the diffusivity tensor. 

In order to assure the proper coupling between flow problem and the transient advection-diffusion problem, the viscosity is assumed to exponentially depend on the concentration as follows:
\begin{align}
\mu(c(\mathbf{x},t)) = \mu_{0}\mathrm{exp}\left[R_c (1-c(\mathbf{x},t))\right]
\end{align}
where $\mu_{0}$ is the base viscosity and $R_c$ denotes the log-mobility ratio in an isothermal miscible displacement.
Figure \ref{Fig:Hele_shaw_Schm} represents the computational domain as well as the assigned initial and boundary conditions for this boundary value problem. Parameter values for this coupled flow and transport problem are provided in Table \ref{Tb9:Coupled_flow_transport_data}.
The perturbation on the interface of the two fluids is imposed by considering heterogeneous material properties for the porous domain, such as heterogeneous permeabilities. Moreover, the initial condition for the transport problem is defined using a random function throughout the domain.
{\small
	\begin{table}[!h]
		\caption{Data-set for coupled flow and transport problem.}
		\centering
		\begin{tabular}{|c|c|} \hline
			Parameter & Value \\
			\hline
			$\gamma \mathbf{b}$ & $\{0.0,0.0\}$\\
			$f$ & $0.0$\\
			$L_x $ & $1.0$ \\
			$L_y $ & $0.4$ \\
			$\mu_0 $ & $0.001$ \\
			$R_c$ & $3.0$ \\
			$D$ & $2\mathrm{e}-6$\\
			$\beta $ & $1.0$ \\
			$\mathbf{K}_1$ & $\left[1.0,0.0;0.0,0.5\right]$\\
			$\mathbf{K}_2$ &  $\left[0.05,0.0;0.0,0.01\right]$\\
			$c_0$ & $0.0$ \\
			$c_{\mathrm{inj}}$ & $1.0$ \\
			$p_{\mathrm{atm}}$ & $1.0$ \\            
			$u_{\mathrm{inj}}$ & $0.004$ \\
			$\Delta t$ & $0.5$\\
			$T$ & $150$\\
			\hline 
		\end{tabular}
		\label{Tb9:Coupled_flow_transport_data}
	\end{table}
}

Figure \ref{Fig:Hele_Shaw_Concentration_profiles} shows the
concentration profile under the double porosity/permeability
model. Two main inferences can be drawn from this figure. 
\emph{First}, Saffman-Taylor-type physical instability can
\emph{also} occur under the double porosity/permeability
model. As discussed earlier, the classical Saffman-Taylor
instability has been shown to occur under the Darcy model.
However, a further systematic study needs to be conducted
to find out the similarities and differences between the
classical Saffman-Taylor instability and the one under
the double porosity/permeability model.
\emph{Second}, the proposed stabilized formulation is capable of
eliminating the spurious numerical instabilities without
suppressing the underlying physical instability. Achieving
this important attribute under the proposed stabilized
formulation is one of the main contributions of this paper,
as it has been shown recently that some stabilized formulations
(for example, the Streamline/Upwind Petrov Galerkin (SUPG), 
and Galerkin Least-Squares (GLS) formulations) which are
commonly used to suppress spurious numerical instabilities,
may also suppress physical instabilities in some cases
\citep{Shabouei_Nakshatrala_VF,Shabouei_PhDThesis_UH_2016}.

%% file: Sections/S10_VMS_CR.tex
\section{CONCLUDING REMARKS}
\label{Sec:S10_VMS_CR}
This paper has made several contributions to the
modeling of fluid flow in porous media with dual
pore-networks and possible mass transfer across
the pore-networks. 
\emph{First}, a stabilized mixed finite element
formulation has been presented for the double
porosity/permeability mathematical model. 
\emph{Second}, a systematic error analysis has
been performed on the proposed stabilized weak
formulation. Numerical convergence analysis and
patch tests have been used to illustrate the
convergence behavior and accuracy of the proposed
mixed formulation in the discrete setting.
\emph{Third}, the mathematical properties that the solutions
of the double porosity/permeability model enjoy have been
utilized to construct mechanics-based \textit{a posteriori} error
measures to assess the accuracy of the numerical solutions.
\emph{Last but not least}, the performance of the
proposed stabilized mixed formulation for modeling
the transient flow as well as coupled problems has
been illustrated using representative numerical
examples.
Some of the significant findings of this
paper can be summarized as follows:
\begin{enumerate}[(C1)]
\item[(C1)] Equal-order interpolation for all the field
  variables (pressure and velocity vector fields), which
  is computationally the most convenient, is stable under
  the proposed stabilized mixed formulation.
\item[(C2)] Patch tests revealed that the classical mixed
  formulation produces spurious node-to-node oscillations
  in the pressure fields under equal-order interpolation
  for all the field variables. The proposed stabilized
  mixed formulation was able to eliminate such unphysical
  oscillations in the pressure fields, and passed the
  patch tests up to the machine precision. 
 \item[(C3)] The numerical convergence rates
   obtained using the proposed stabilized
   formulation were in accordance with the
   theory for both $h$- and $p$-refinements.
 \item[(C4)] The accuracy of numerical solutions was assessed using the mechanics-based \textit{a posteriori} error measures for the pipe bend problem. The errors decreased \emph{monotonically} with mesh refinement for different orders of interpolation. This implies that the stabilized formulation is convergent and the computer implementation is correct. It should be noted that the mechanics-based solution verification method can be applied to any problem with any boundary condition.
 \item[(C5)] An extension of the proposed formulation
   to the transient case has performed well, as it was
   able to predict accurately that the rate of decay
   of the response (e.g., the velocity front) in the
   macro-pore network is slower than that of the
   micro-pore network. Physically, this
     phenomenon of slower decay can be attributed to
     the higher permeability (which implies lower
     dissipation, as dissipation is inversely
     proportional to the permeability) in the
     macro-pore network.
 \item[(C6)] The proposed stabilized mixed formulation
   suppressed the unphysical numerical instabilities
   but yet captured the underlying physical instability
   when applied to a coupled flow and transport problem
   in porous media with dual pore-networks. The captured  
   physical instability, is similar to the classical 
   Saffman-Taylor instability that has been shown to exist 
   for coupled Darcy and transport equations. The proposed 
   formulation will be particularly attractive for studying
   physical instabilities, as it has been shown recently
   that some well-known stabilized formulations which are
   designed to suppress numerical instabilities also
   suppressed physical instabilities.
\end{enumerate}

The research presented herein can be extended
on three fronts.
\begin{enumerate}[(R1)]
\item One can develop a hierarchy of mathematical
  models by incorporating other processes into
  the double porosity/permeability model. For example,
  the flow of multi-phase fluids in porous media
  exhibiting double porosity/permeability, and
  the incorporation of deformation of
  porous solid.
\item One can perform a theoretical study on
  the Saffman-Taylor-type instabilities under
  the double porosity/permeability model. In
  particular, one can address whether there
  are additional instability modes under the double
  porosity/permeability model when compared
  with the classical Saffman-Taylor instability 
  (which is based on the Darcy model).
  One can also obtain scaling laws. 
\item Heterogeneity of material properties
  and discontinuous distribution of permeability are
  very common in subsurface formations. Studies for
  the case of Darcy equations have shown that continuous
  formulations cannot properly handle abrupt changes
  in material properties, as the numerical solutions
  suffer from Gibbs phenomenon (which manifests as
  spurious oscillations in the numerical solution
  fields) \citep{Hughes_Masud_Wan_2006}.
  Thus one can develop a stabilized mixed discontinuous Galerkin
  formulation for the double porosity/permeability model that
  does not suffer from the Gibbs phenomenon in the solution fields
  when applied to problems with disparate medium properties.
\end{enumerate}

%% file: Sections/VMS_Appendix.tex
\section{Derivation of the proposed stabilized formulation}
\label{Sec:VMS_Appendix_A}
We provide a formal mathematical derivation of the
proposed stabilized mixed weak formulation. We
employ the variational multiscale paradigm 
\citep{Hughes_1995}, and obtain the stabilization
terms and the stabilization parameter in a consistent
manner. 
Such an approach has been successfully employed to 
develop stabilized formulations for porous media 
models with single pore-network; for example, see 
\citep{Masud_Hughes_2002,Hughes_Masud_Wan_2006,Nakshatrala_Turner_Hjelmstad_Masud_CMAME_2006_v195_p4036}. 
The basic idea is to decompose the solution
into resolved and unresolved components, estimate
the unresolved component, and substitute the estimated
component into the weak form to obtain the proposed
stabilized mixed formulation. 
By a resolved component, we refer to that part of the solution
that is captured by the underlying formulation (which, in our 
case, is the classical mixed formulation). The unresolved 
component can be interpreted as the difference between 
the exact solution and the resolved component. To improve 
the accuracy of the numerical solution, the unresolved
components need to be estimated accurately, which 
can be achieved using the variational multiscale paradigm.

We start our derivation by decomposing the macro-scale 
and micro-scale velocities into resolved and unresolved 
components. Mathematically,
\begin{align}
  \mathbf{u}_{1}(\mathbf{x}) = \underbrace{\overline{\mathbf{u}}_{1}
    (\mathbf{x})}_{\mathrm{resolved}} + \underbrace{\mathbf{u}_{1}^{'}(\mathbf{x})}_{\mathrm{unresolved}}
  \quad \mathrm{and} \quad 
  \mathbf{u}_{2}(\mathbf{x}) = \underbrace{\overline{\mathbf{u}}_{2}
      (\mathbf{x})}_{\mathrm{resolved}} + \underbrace{\mathbf{u}_{2}^{'}(\mathbf{x})}_{\mathrm{unresolved}}
    \label{Eqn:VMS_Resolved_Unresolved_u}
\end{align}
where the resolved components are denoted by
over-lines, and the primed quantities represent
the unresolved components. Similarly, the weighting
functions corresponding to these velocities are
decomposed as follows:
\begin{align}
  \label{Eqn:VMS_Resolved_Unresolved_w}
    \mathbf{w}_{1}(\mathbf{x}) = \overline{\mathbf{w}}_{1}
    (\mathbf{x}) + \mathbf{w}_{1}^{'}(\mathbf{x})
    \quad \mathrm{and} \quad
    \mathbf{w}_{2}(\mathbf{x}) = \overline{\mathbf{w}}_{2}
    (\mathbf{x}) + \mathbf{w}_{2}^{'}(\mathbf{x})
\end{align}
In principle, one could perform a similar decomposition to 
the macro- and micro-pressure fields. Herein, we assume 
that the pressure fields will be adequately resolved. Therefore,  
we do not decompose the pressure fields (i.e., $p_{1}(\mathbf{x})$ 
and $p_{2}(\mathbf{x})$) and the corresponding 
weighting functions (i.e., $q_{1}(\mathbf{x})$ 
and $q_{2}(\mathbf{x})$). 
In Sections \ref{Sec:S4_VMS_Theoretical} and \ref{Sec:S5_VMS_Canonical}, 
we have illustrated, through stability analysis and numerical simulations, that 
such an assumption is still able to provide a stable and accurate formulation. 
To localize the unresolved components, we enforce the
closure conditions that the unresolved components of
the velocities and their weighting functions vanish on 
the element boundaries. That is,
\begin{align}
  \label{Eqn:VMS_unresolved_BC}
  \mathbf{u}_{1}^{'}(\mathbf{x}) = \mathbf{0}, \quad 
  \mathbf{u}_{2}^{'}(\mathbf{x}) = \mathbf{0}, \quad
  \mathbf{w}_{1}^{'}(\mathbf{x}) = \mathbf{0}
  \quad \mathrm{and} \quad 
  \mathbf{w}_{2}^{'}(\mathbf{x}) = \mathbf{0}
  \quad \mathrm{on} \; \partial \Omega^{e};
  \; e = 1, \cdots, Nele
\end{align}

By substituting the multiscale decompositions given
by equations \eqref{Eqn:VMS_Resolved_Unresolved_u}
an \eqref{Eqn:VMS_Resolved_Unresolved_w} into the
classical mixed formulation given in equation
\eqref{Eqn:VMS_classical_mixed_formulation}, 
invoking the arbitrariness of the weighting functions ( 
$\overline{\mathbf{w}}(\mathbf{x})$ and $\mathbf{w}^{'}(\mathbf{x})$), 
and enforcing the closure conditions given by equation
\eqref{Eqn:VMS_unresolved_BC}, we obtain two subproblems for 
each pore-network. The two subproblems corresponding to the
macro-pore network can be written as follows:   
\begin{subequations}
\begin{alignat}{2}
\label{Eqn:VMS_coarse_scale_macro}
&(\overline{\mathbf{w}}_{1}
  ;\mu \mathbf{K}_{1}^{-1} \overline{\mathbf{u}}_{1}
  + \mu \mathbf{K}_{1}^{-1} \mathbf{u}_{1}^{'})
  - (\mathrm{div}[\overline{\mathbf{w}}_{1}];p_1)
  + (q_1;\mathrm{div}[\overline{\mathbf{u}}_{1}] + \mathrm{div}[\mathbf{u}_{1}^{'}])
  + (q_1;\beta/\mu(p_1 - p_2)) = \nonumber \\
  & (\overline{\mathbf{w}}_{1};\gamma \mathbf{b})
  - \langle \overline{\mathbf{w}}_{1} \cdot \widehat{\mathbf{n}} ;p_{01}
  \rangle_{\Gamma^{\mathrm{p}}_{1}}  \\
  \label{Eqn:VMS_fine_scale_macro}
&(\mathbf{w}_{1}^{'};\mu \mathbf{K}_{1}^{-1} \overline{\mathbf{u}}_{1}
  + \mu \mathbf{K}_{1}^{-1} \mathbf{u}_{1}^{'})_{\Omega^{e}}
  - (\mathrm{div}[\mathbf{w}_{1}^{'}];p_1)_{\Omega^{e}} = 
   (\mathbf{w}_{1}^{'};\gamma \mathbf{b})_{\Omega^{e}}   
   \quad \forall e = 1, \cdots, Nele
\end{alignat}
\end{subequations}
The two subproblems corresponding to the
micro-pore network can be written as follows:   
\begin{subequations}
\begin{alignat}{2}
  \label{Eqn:VMS_coarse_scale_micro}
  & (\overline{\mathbf{w}}_{2};\mu \mathbf{K}_{2}^{-1} \overline{\mathbf{u}}_{2}
  + \mu \mathbf{K}_{2}^{-1} \mathbf{u}_{2}^{'})
  - (\mathrm{div}[\overline{\mathbf{w}}_{2}] ;p_2)
  + (q_2;\mathrm{div}[\overline{\mathbf{u}}_{2}]+ \mathrm{div}[\mathbf{u}_{2}^{'}])
  - (q_2;\beta/\mu(p_1 - p_2)) = \nonumber \\
  & (\overline{\mathbf{w}}_{2};\gamma \mathbf{b})
  - \langle\overline{\mathbf{w}}_{2} \cdot \widehat{\mathbf{n}} ;p_{02}
  \rangle_{\Gamma^{\mathrm{p}}_{2}} \\
  \label{Eqn:VMS_fine_scale_micro}
  & (\mathbf{w}_{2}^{'};\mu \mathbf{K}_{2}^{-1}
  \overline{\mathbf{u}}_{2} + \mu \mathbf{K}_{2}^{-1}
  \mathbf{u}_{2}^{'})_{\Omega^{e}}
  - (\mathrm{div}[\mathbf{w}_{2}^{'}];p_2)_{\Omega^{e}}
  =  (\mathbf{w}_{2}^{'};\gamma \mathbf{b})_{\Omega^{e}} 
  \quad \forall e = 1, \cdots, Nele
\end{alignat}
\end{subequations}

We enforce the closure conditions using bubble
functions, which vanish on the boundary of the
domain on which they are defined
\citep{Baiocchi_Brezzi_France_CMAME_1993_v105_p125}.
We, therefore, mathematically write the unresolved
quantities as follows: 
\begin{align}
  \label{Eqn:VMS_bubble_functions_macro}
  \mathbf{u}_{1}^{'} (\mathbf{x}) = b^{e}(\mathbf{x})
  \boldsymbol{\xi}_{1}, \; 
  \mathbf{w}_{1}^{'} (\mathbf{x}) = b^{e}(\mathbf{x})
  \boldsymbol{\zeta}_{1}, \; 
  \mathbf{u}_{2}^{'} (\mathbf{x}) = b^{e}(\mathbf{x})
  \boldsymbol{\xi}_{2}
  \quad \mathrm{and} \quad
  \mathbf{w}_{2}^{'} (\mathbf{x}) = b^{e}(\mathbf{x})
  \boldsymbol{\zeta}_{2}
  \quad \forall \mathbf{x} \in \Omega^{e}
\end{align}
where $\boldsymbol{\xi}_{1}$, $\boldsymbol{\xi}_{2}$,
$\boldsymbol{\zeta}_{1}$ and $\boldsymbol{\zeta}_{2}$
are constant vectors independent of $\mathbf{x}$, and
$b^{e}(\mathbf{x})$ is a bubble function defined on
the element $\Omega^{e}$.
By substituting equation \eqref{Eqn:VMS_bubble_functions_macro}
into the subproblems given by equations
\eqref{Eqn:VMS_fine_scale_macro} and
\eqref{Eqn:VMS_fine_scale_micro}, and
noting that $\boldsymbol{\zeta}_{1}$ and
$\boldsymbol{\zeta}_{2}$ are arbitrary
vectors; we estimate the unresolved
velocities as follows: 
\begin{subequations}
  \begin{alignat}{2}
    \label{Eqn:VMS_bubble_integral_form_macro}
    &\mathbf{u}_{1}^{'} (\mathbf{x}) = - b^{e} (\mathbf{x}) \left( \int_{\Omega^{e}} \left(b^{e} (\mathbf{x})\right)^{2} d\Omega\right)^{-1}  \int_{\Omega^{e}} b^{e} (\mathbf{y}) \overline{\mathbf{r}}_{1}(\mathbf{y})  d\Omega_{\mathbf{y}}   \\
    \label{Eqn:VMS_bubble_integral_form_micro}
    &\mathbf{u}_{2}^{'} (\mathbf{x}) = - b^{e} (\mathbf{x}) \left( \int_{\Omega^{e}} \left(b^{e} (\mathbf{x})\right)^{2} d\Omega\right)^{-1}  \int_{\Omega^{e}} b^{e} (\mathbf{y}) \overline{\mathbf{r}}_{2}(\mathbf{y})  d \Omega_{\mathbf{y}}    
  \end{alignat}
\end{subequations}
where the residuals of the resolved quantities for
the macro and micro pore-networks are, respectively,
defined as follows:
\begin{align}
  \label{Eqn:VMS_r1}
  \overline{\mathbf{r}}_{1}(\mathbf{x}) &=
  \overline{\mathbf{u}}_{1}(\mathbf{x}) +
  \frac{1}{\mu}\mathbf{K}_{1}
  \left(\mathrm{grad}[p_{1}]
  - \gamma \mathbf{b}(\mathbf{x}) \right) \\
  \label{Eqn:VMS_r2}
  \overline{\mathbf{r}}_{2} (\mathbf{x}) &=
  \overline{\mathbf{u}}_{2}(\mathbf{x}) +
  \frac{1}{\mu}\mathbf{K}_{2} \left(\mathrm{grad}[p_{2}]
  - \gamma \mathbf{b}(\mathbf{x}) \right)
\end{align}
Since in a finite element setting, the residuals
($\overline{\mathbf{r}}_{1}(\mathbf{x})$ and
$\overline{\mathbf{r}}_{2}(\mathbf{x})$) are
essentially constant over an element in the
limit of an adequately refined mesh, the
velocities in equations
\eqref{Eqn:VMS_bubble_integral_form_macro}
and \eqref{Eqn:VMS_bubble_integral_form_micro}
can be written as follows:
\begin{align}
  \label{Eqn:VMS_bubble_integral_form_macro2}
  \mathbf{u}_{1}^{'} (\mathbf{x}) = -\tau(\mathbf{x})
  \overline{\mathbf{r}}_{1}(\mathbf{x}) 
  \quad \mathrm{and} \quad 
  \mathbf{u}_{2}^{'} (\mathbf{x}) = -\tau(\mathbf{x})
  \overline{\mathbf{r}}_{2}(\mathbf{x})    
\end{align}
where the stabilization parameter $\tau(\mathbf{x})$
takes the following form: 
\begin{align}
  \tau (\mathbf{x}) = b^{e} (\mathbf{x}) \left( \int_{\Omega^{e}} \left(b^{e} (\mathbf{x})\right)^{2} d\Omega\right)^{-1} \left( \int_{\Omega^{e}} b^{e} (\mathbf{x}) d\Omega \right)
\end{align}
One can employ the above stabilization parameter for
obtaining a stabilized formulation. However, for the
double porosity/permeability model it is adequate
to employ a representative value for the stabilization
parameter, which is justified by the convergence analysis
we presented in this paper. To obtain a representative
value for the stabilization parameter, we consider the
average of $\tau(\mathbf{x})$, which can be written
as follows:
\begin{align}
  \tau_{\mathrm{avg}}
  = \frac{1}{\mathrm{meas}(\Omega^{e})} \int_{\Omega^{e}}
  \tau(\mathbf{x}) d \Omega
  = \left( \int_{\Omega^{e}} \left(b^{e} (\mathbf{x})\right)^{2}
  d\Omega\right)^{-1} \left( \int_{\Omega^{e}} b^{e}(\mathbf{x})
  d\Omega \right)^{2}
\end{align}
where $\mathrm{meas}(\Omega^{e})$ denotes
the measure of $\Omega^{e}$ (By measure
we mean length in 1D, area in 2D and
volume in 3D.) 
It has been shown in
\citep{Nakshatrala_Turner_Hjelmstad_Masud_CMAME_2006_v195_p4036}
that it is possible to construct a bubble
function that gives a value of one-half for
$\tau_{\mathrm{avg}}$ for a given $\Omega^{e}$.
We thus take one-half to be the representative
value for the stabilization parameter. We then
approximate the unresolved components of the
velocities as follows:
\begin{align}
  \mathbf{u}_{1}^{'}(\mathbf{x}) \approx -\frac{1}{2}
  \overline{\mathbf{r}}_{1}(\mathbf{x})
  \quad \mathrm{and} \quad 
  \mathbf{u}_{2}^{'}(\mathbf{x}) \approx -\frac{1}{2}
  \overline{\mathbf{r}}_{2}(\mathbf{x})
\end{align}
By substituting the above expressions into
the subproblems given by equations
\eqref{Eqn:VMS_coarse_scale_macro} and
\eqref{Eqn:VMS_coarse_scale_micro}, and
noting the definitions for 
$\overline{\mathbf{r}}_{1}(\mathbf{x})$ and
$\overline{\mathbf{r}}_{2}(\mathbf{x})$, we
obtain a stabilized formulation of the
following form:
\begin{align}
  \label{Eqn:VMS_Gal_relation2}
  \mathcal{B}_{\mathrm{Gal}}(\overline{\mathbf{w}}_1,
  \overline{\mathbf{w}}_2,q_1,q_2;
  \overline{\mathbf{u}}_1,\overline{\mathbf{u}}_2,p_1,p_2)
  - \frac{1}{2} \left(\mu \mathbf{K}_{1}^{-1} \overline{\mathbf{w}}_1
  - \mathrm{grad}[q_1];\overline{\mathbf{u}}_{1}+ \frac{1}{\mu} \mathbf{K}_{1} \mathrm{grad}[p_{1}]\right) \nonumber \\
  - \frac{1}{2} \left(\mu \mathbf{K}_{2}^{-1}
  \overline{\mathbf{w}}_2 - \mathrm{grad}[q_2];
  \overline{\mathbf{u}}_{2}+ \frac{1}{\mu} \mathbf{K}_{2}
  \mathrm{grad}[p_{2}] \right) 
  = \mathcal{L}_{\mathrm{Gal}}(\overline{\mathbf{w}}_1,
  \overline{\mathbf{w}}_2,q_1,q_2)\nonumber \\
  -\frac{1}{2} \left(\mu \mathbf{K}_{1}^{-1} \overline{\mathbf{w}}_1
  - \mathrm{grad}[q_1];\frac{1}{\mu} \mathbf{K}_{1} \gamma \mathbf{b} \right) 
  - \frac{1}{2}\left(\mu \mathbf{K}_{2}^{-1} \overline{\mathbf{w}}_2 - \mathrm{grad}[q_2];
  \frac{1}{\mu} \mathbf{K}_{2} \gamma \mathbf{b} \right)
\end{align}
where $\mathcal{B}_{\mathrm{Gal}}$ and $\mathcal{L}_{\mathrm{Gal}}$ are defined in equations \eqref{Eqn:VMS_bilinear_form} and \eqref{Eqn:VMS_linear_form}, respectively.
It should be noted that all the quantities in
the above equation are resolved components.
We therefore drop the over-lines for convenience,
and write the above stabilized mixed formulation
in the following compact form: 
\begin{align}
  \mathcal{B}_{\mathrm{stab}}(\mathbf{w}_1,\mathbf{w}_2,q_1,q_2;
  \mathbf{u}_1,\mathbf{u}_2,p_1,p_2)
  = \mathcal{L}_{\mathrm{stab}}(\mathbf{w}_1,\mathbf{w}_2,q_1,q_2) \nonumber \\
   \quad \forall
  \left(\mathbf{w}_1(\mathbf{x}), \mathbf{w}_2(\mathbf{x})\right) ~
  \in \mathcal{W}_1 \times \mathcal{W}_2,
  \left(q_1(\mathbf{x}),q_2(\mathbf{x})\right)
  \in \mathcal{Q} 
\end{align}
where the bilinear form and the linear functional
are, respectively, defined as follows:
\begin{align}
  \mathcal{B}_{\mathrm{stab}}(\mathbf{w}_1,\mathbf{w}_2,q_1,q_2;
  \mathbf{u}_1,\mathbf{u}_2,p_1,p_2) := \mathcal{B}_{\mathrm{Gal}}(\mathbf{w}_1,\mathbf{w}_2,q_1,q_2;
  \mathbf{u}_1,\mathbf{u}_2,p_1,p_2) \nonumber \\
  -\frac{1}{2} \left(\mu \mathbf{K}_{1}^{-1} \mathbf{w}_1 - \mathrm{grad}[q_1];
  \frac{1}{\mu} \mathbf{K}_{1} (\mu \mathbf{K}_{1}^{-1} \mathbf{u}_1 + \mathrm{grad}[p_1])\right)
  \nonumber \\
  -\frac{1}{2} \left(\mu \mathbf{K}_{2}^{-1} \mathbf{w}_2 - \mathrm{grad}[q_2];
  \frac{1}{\mu} \mathbf{K}_{2} (\mu \mathbf{K}_{2}^{-1} \mathbf{u}_2 + \mathrm{grad}[p_2])\right) 
\end{align}  
\begin{align}
  \mathcal{L}_{\mathrm{stab}}(\mathbf{w}_1,\mathbf{w}_2,q_1,q_2) &:= \mathcal{L}_{\mathrm{Gal}}(\mathbf{w}_1,\mathbf{w}_2,q_1,q_2)
  -\frac{1}{2} \left(\mu \mathbf{K}_{1}^{-1} \mathbf{w}_1 - \mathrm{grad}[q_1];
  \frac{1}{\mu} \mathbf{K}_{1} \gamma \mathbf{b}\right)\nonumber \\
  &-\frac{1}{2} \left(\mu \mathbf{K}_{2}^{-1} \mathbf{w}_2 - \mathrm{grad}[q_2];
  \frac{1}{\mu} \mathbf{K}_{2} \gamma \mathbf{b}\right)
\end{align}

It is important to note that the stabilization terms
are residual-based. Moreover, the stabilization terms
are of adjoint-type and are not of least-squares-type.

%% file: Sections/VMS_Figures.tex
%
%
\begin{figure}
\subfigure[\label{Fig:1D_patch_test}]{
\includegraphics[clip,scale=0.6]{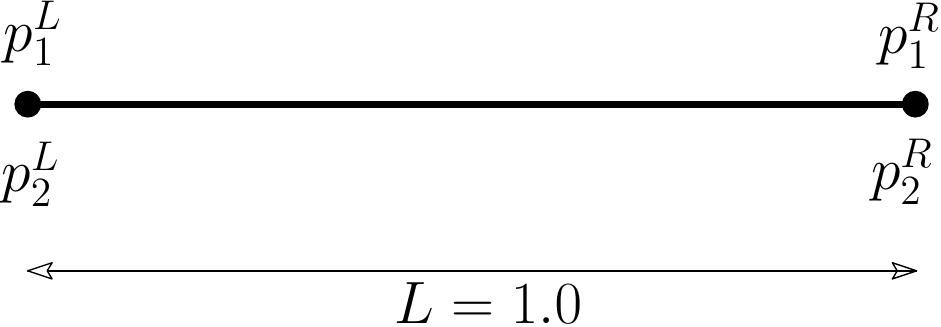}}
\subfigure[\label{Fig:3D_patch_test}]{
\includegraphics[clip,scale=0.6]{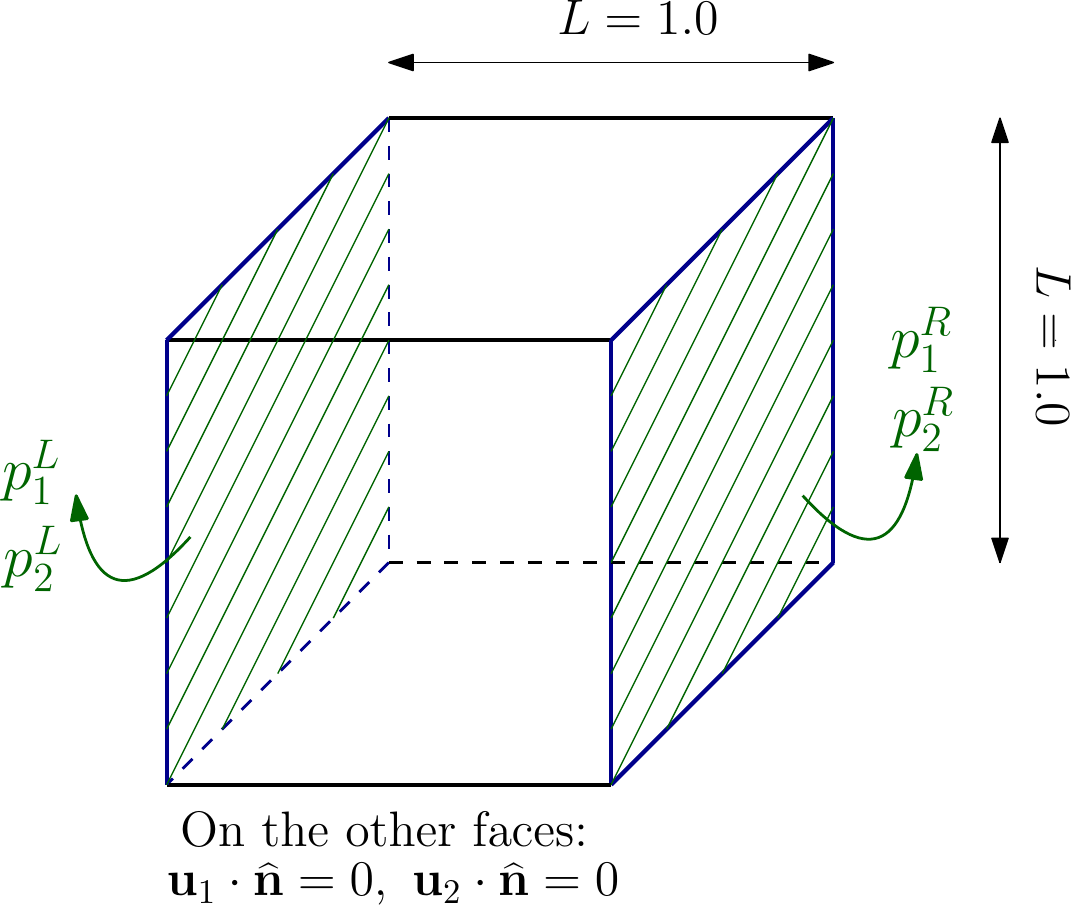}}
\caption{\textsf{Patch tests:}~The left figure provides a pictorial description of the 1D patch test and the right one shows the domain and the boundary conditions in 3D patch test. In 3D test, pressures are prescribed on the left and right faces and on the other faces, the normal component of velocity is zero in both pore-networks.}
\label{Fig:VMS_Patch_tests}
\end{figure}
%
\begin{figure}[!h]
\centering
  \subfigure[]{
	\includegraphics[clip,scale=0.24]{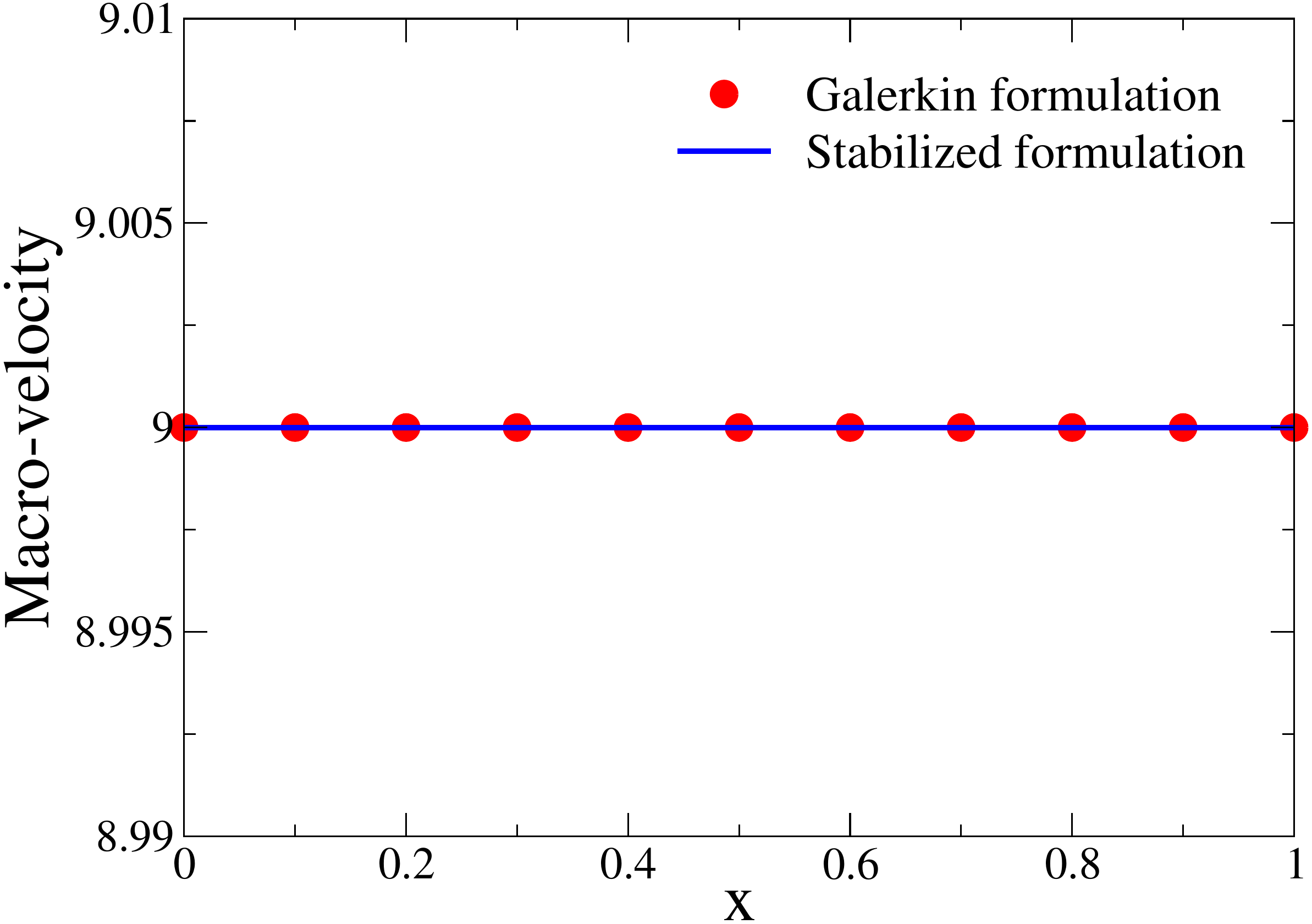}}
  \hspace{15mm}
  \subfigure[]{
	\includegraphics[clip,scale=0.24]{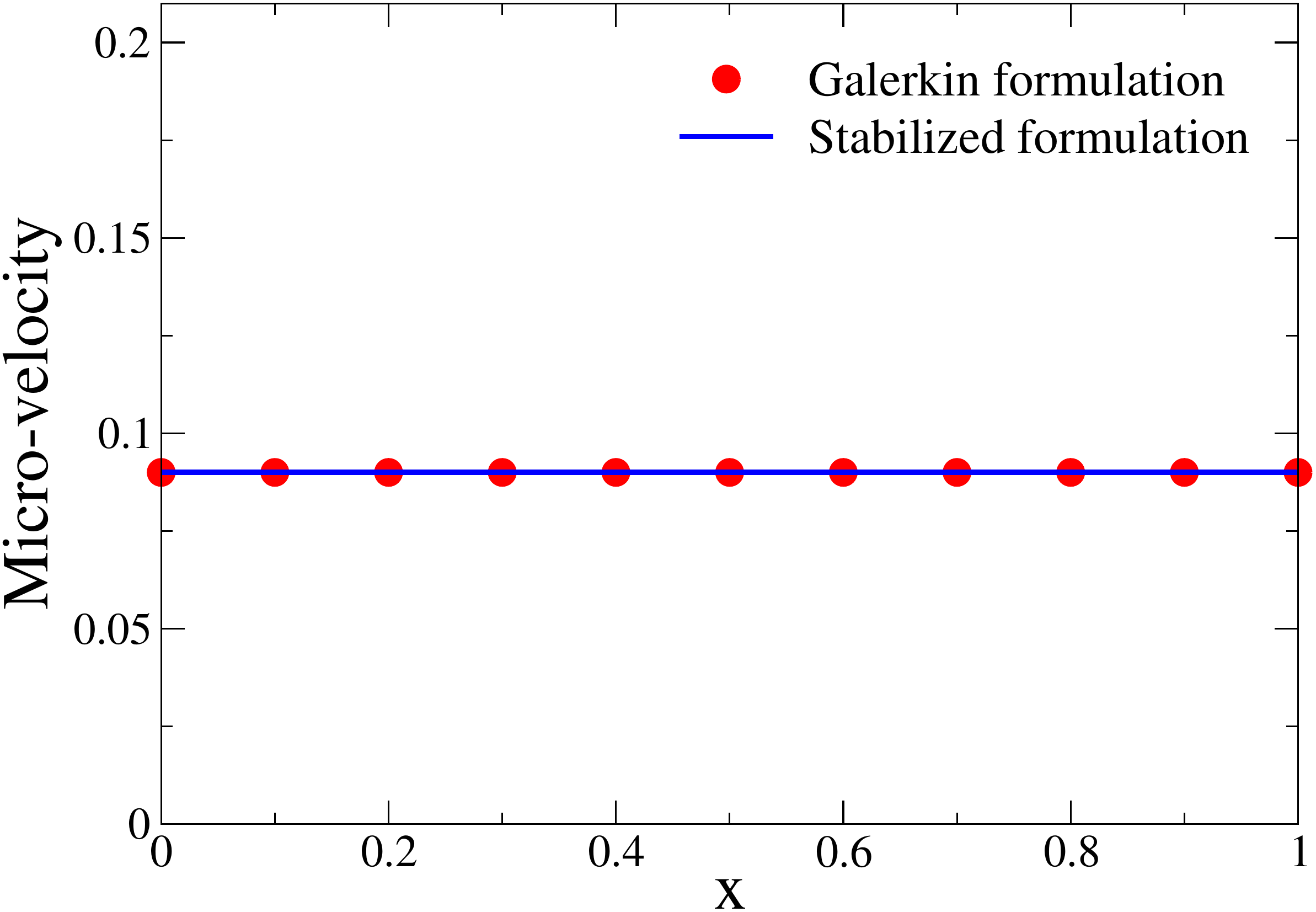}}
  \subfigure[\label{Fig:Macro_pressure_patch_test_VG}]{
  	\includegraphics[clip,scale=0.24]{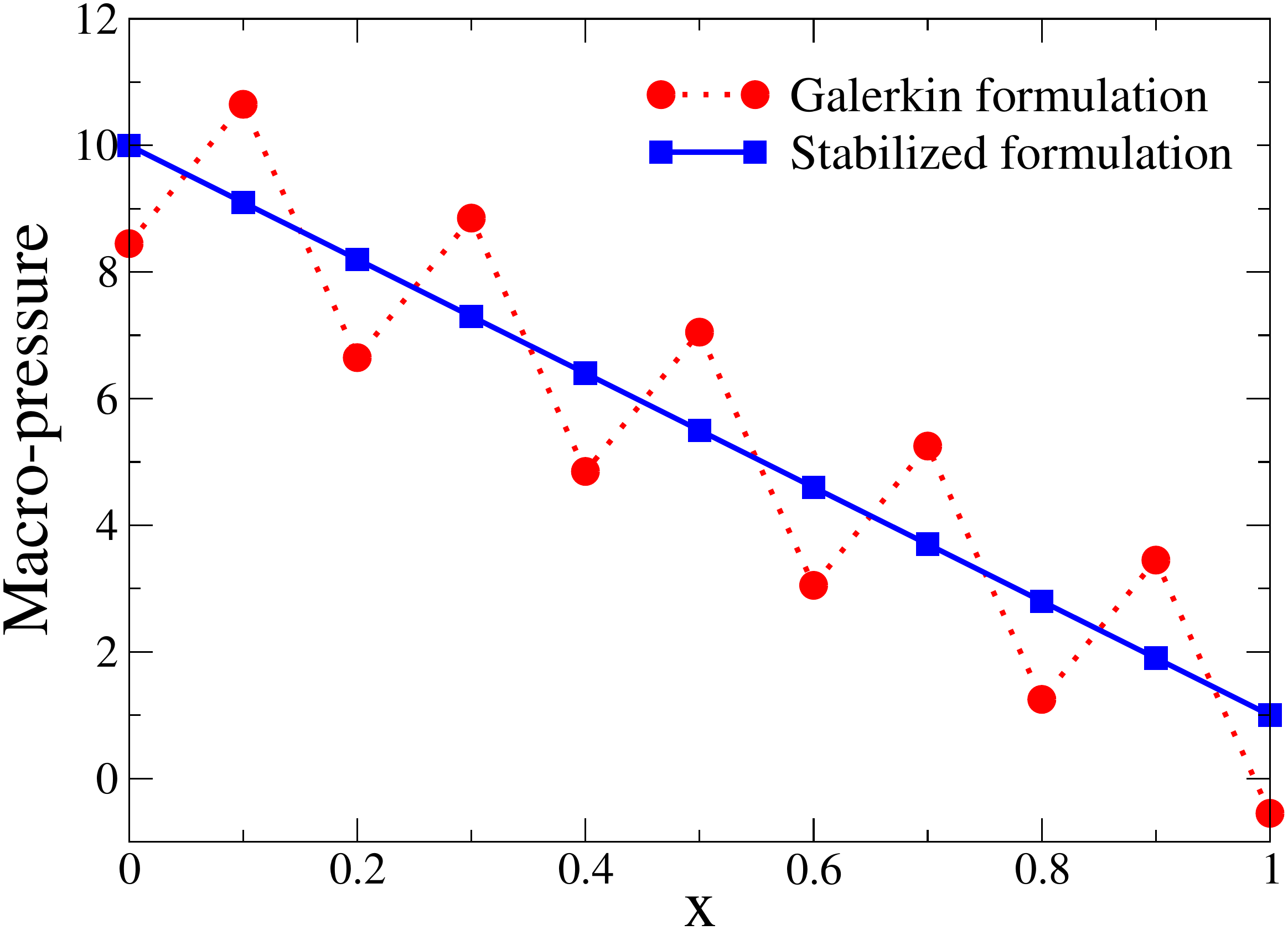}}
  \hspace{15mm}
  \subfigure[\label{Fig:Micro_pressure_patch_test_VG}]{
  	\includegraphics[clip,scale=0.24]{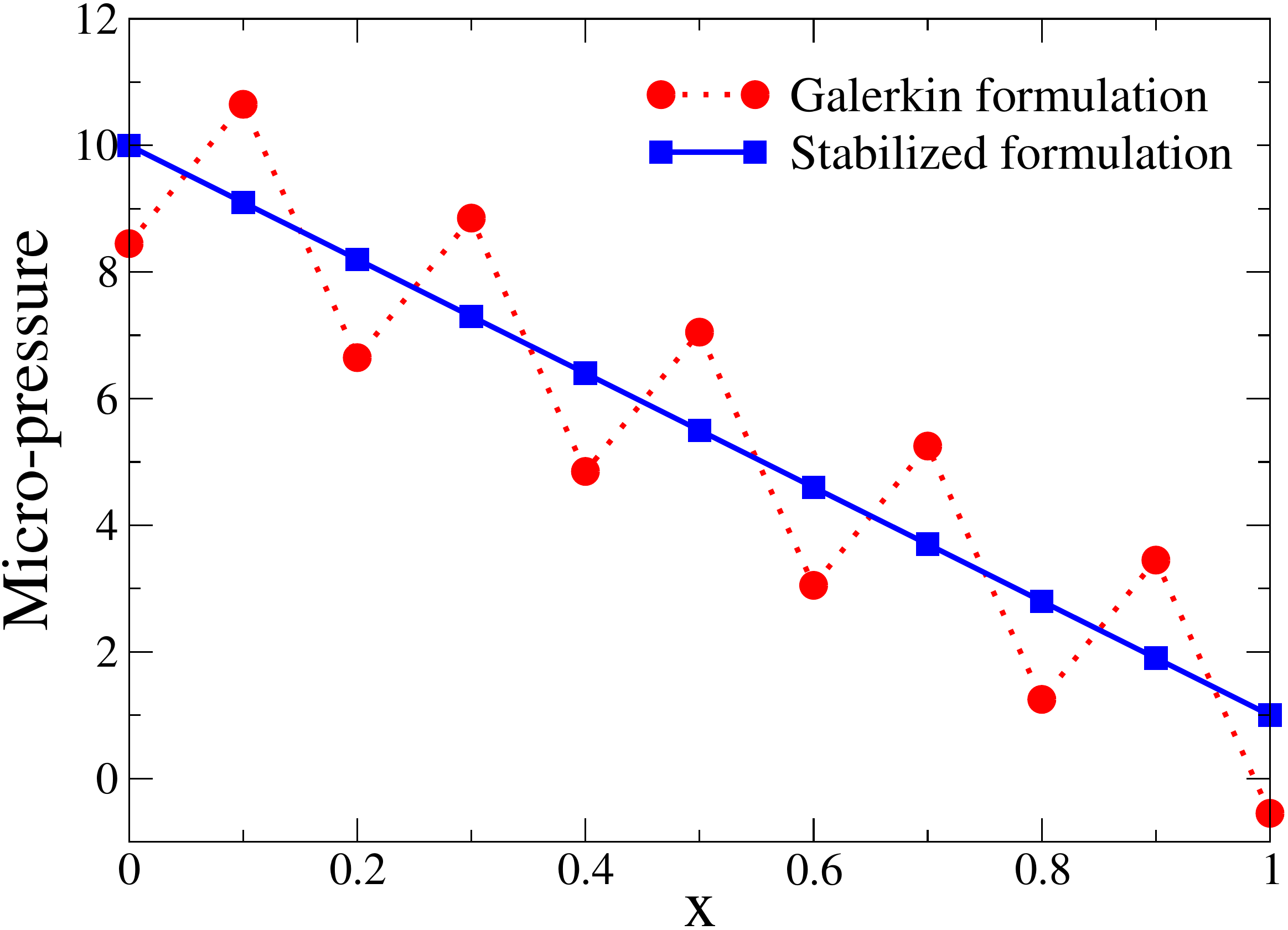}}
	\caption{\textsf{1D patch test:}~Pressure and velocity in pore-networks under Galerkin and proposed  formulations. For velocity fields, the values match with the analytical solution up-to machine precision under both formulations. For pressure fields, spurious oscillations are observed under Galerkin formulation, even for equal-order interpolation. Under stabilized mixed formulation, such oscillations are eliminated.}
    \label{Fig:Pressure_velocity_VMS_Galerkin_1D}
\end{figure}

\begin{figure}
\includegraphics[clip,scale=1.0]{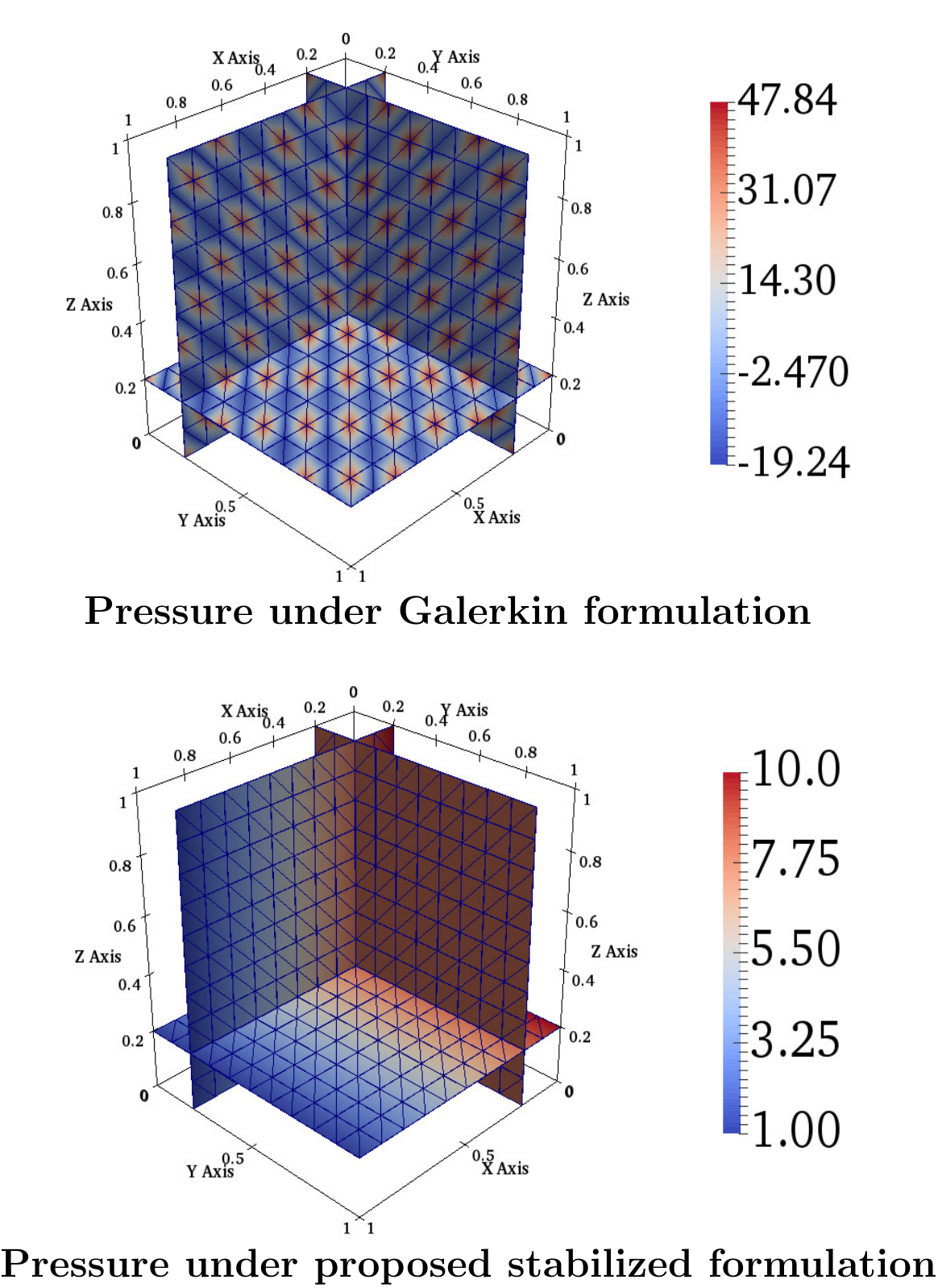}
\caption{\textsf{3D patch test:}~Pressure profiles in the micro-pores and macro-pores in 3D constant flow patch test under Galerkin and the proposed stabilized mixed formulations. Under Galerkin formulation, spurious oscillations are observed in pressure fields, even for equal-order interpolation, which implies that the results are unstable. Such oscillations are eliminated from the pressure profiles under the proposed mixed formulation.}
\label{Fig:Pressure_VMS_Galerkin_3D}
\end{figure}
%
\begin{figure}
  \psfrag{L2p1}{$L_2 \; p_1$}
  \psfrag{L2p2}{$L_2 \; p_2$}
  \psfrag{L2v1}{$L_2 \; v_1$}
  \psfrag{L2v2}{$L_2 \; v_2$}
  \psfrag{H1p1}{$H^1 \; p_1$}
  \psfrag{H1p2}{$H^1 \; p_2$}
  \subfigure{
    \includegraphics[clip,scale=0.38]{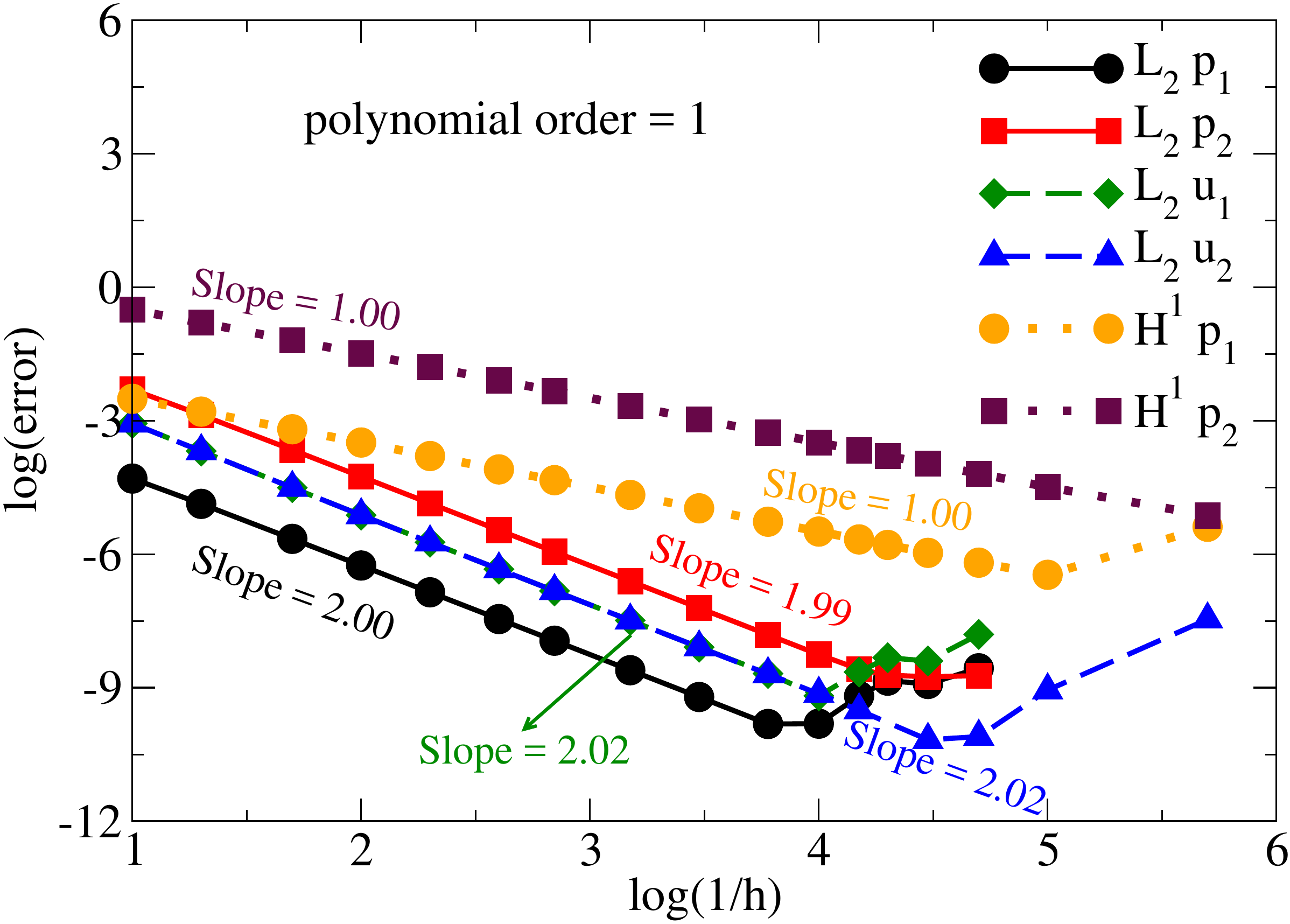}}
    \subfigure{
    \includegraphics[clip,scale=0.38]{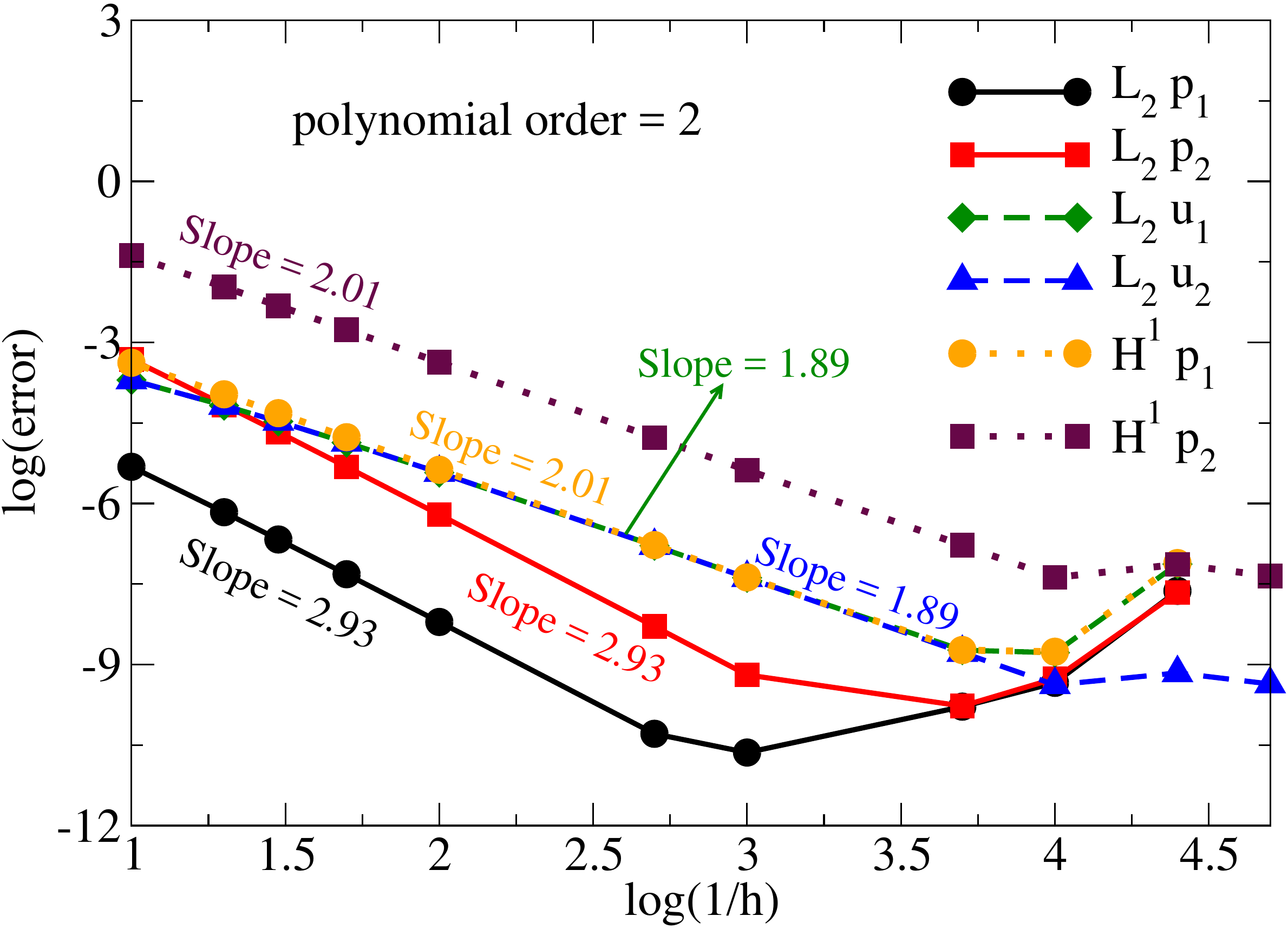}}
  \subfigure{
    \includegraphics[clip,scale=0.38]{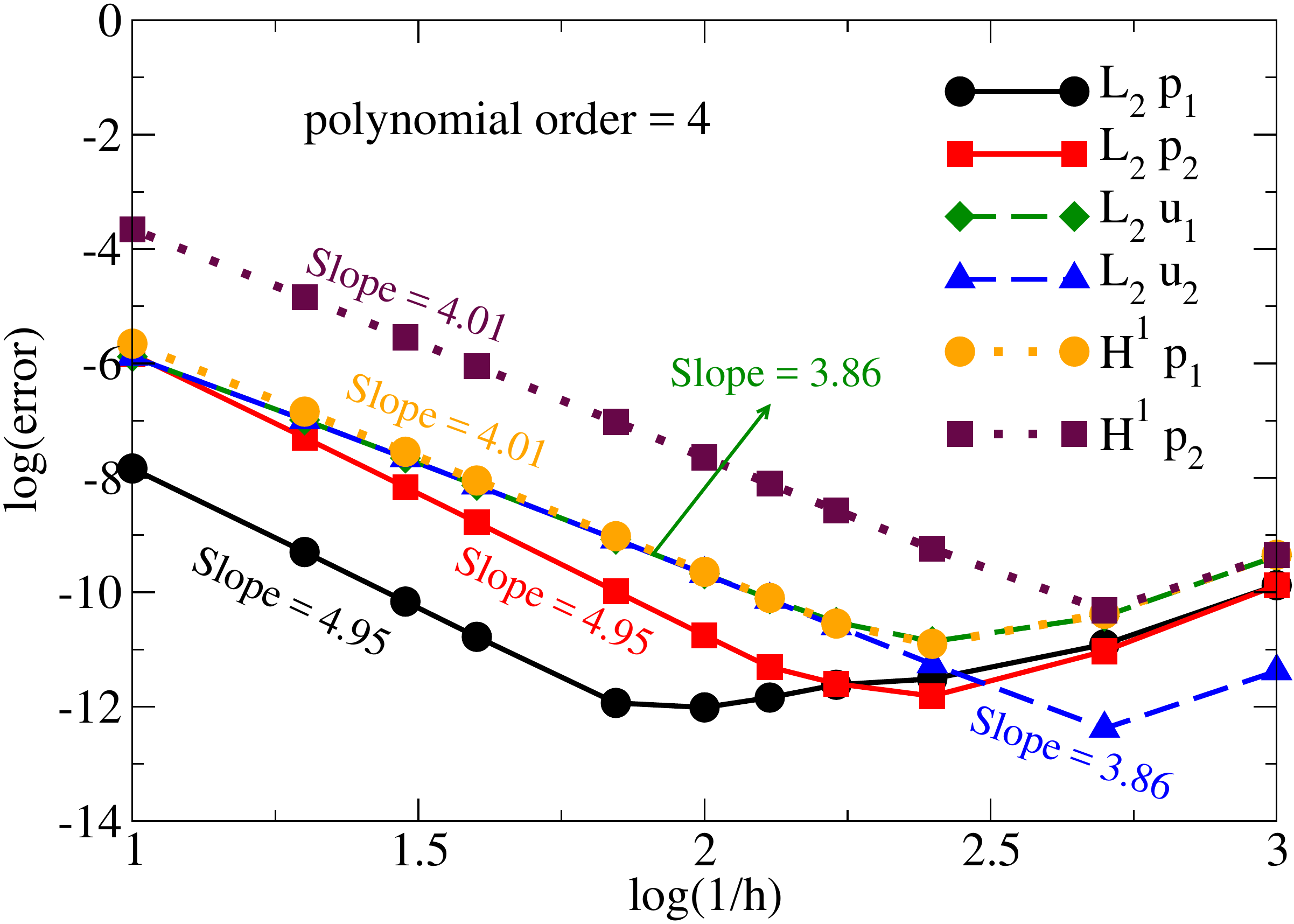}}
  \caption{\textsf{1D numerical convergence
      analysis:}~This figure illustrates the
    numerical convergence of the proposed
    stabilized mixed formulation under
    $h$-refinement for various polynomial
    orders. The rate of convergence is
    polynomial, which is in accordance with the theory.
    \label{Fig:Dual_Problem_1_h_refinement}}
\end{figure}
%
\begin{figure}
  \includegraphics[clip,scale=0.4]{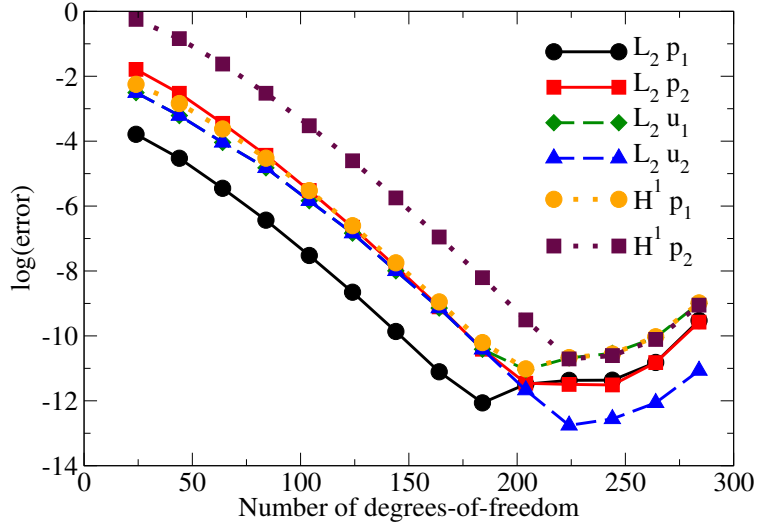}
  \caption{\textsf{1D numerical convergence
      analysis:}~This figure illustrates the
    numerical convergence of the proposed
    stabilized mixed formulation under
    $p$-refinement for a fixed mesh size
    ($h = 0.2$). The number of degrees-of-freedom
    corresponds to $p = 1$ to $14$. The rate of
    convergence is exponential, which is in
    accordance with the theory.
    \label{Fig:Dual_Problem_1_p_refinement}}
\end{figure}

\begin{figure}
  \includegraphics[clip,scale=0.8]{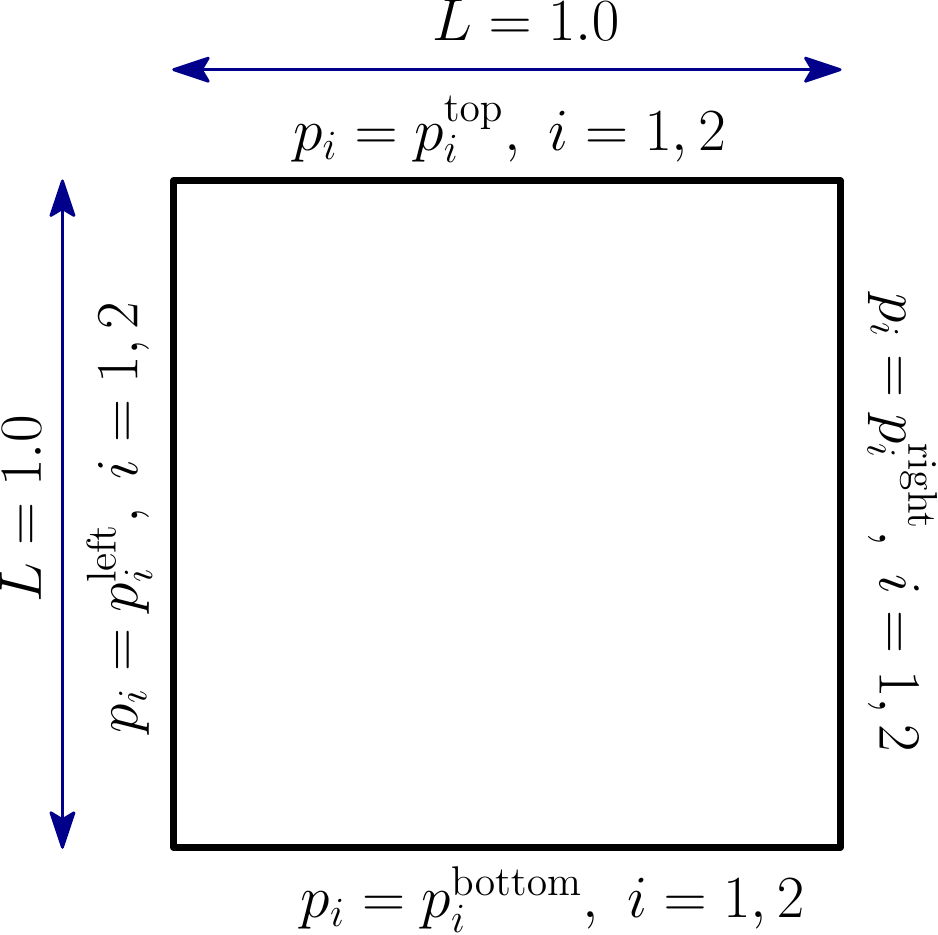}
  \caption{\textsf{2D numerical convergence analysis:}~This
    figure provides a pictorial description of the boundary
    value problem employed in the 2D numerical convergence
    analysis. \label{Fig:Dual_Problem_2D_domain}}
\end{figure}
%
\begin{figure}
  \subfigure{
    \includegraphics[clip,scale=0.38]{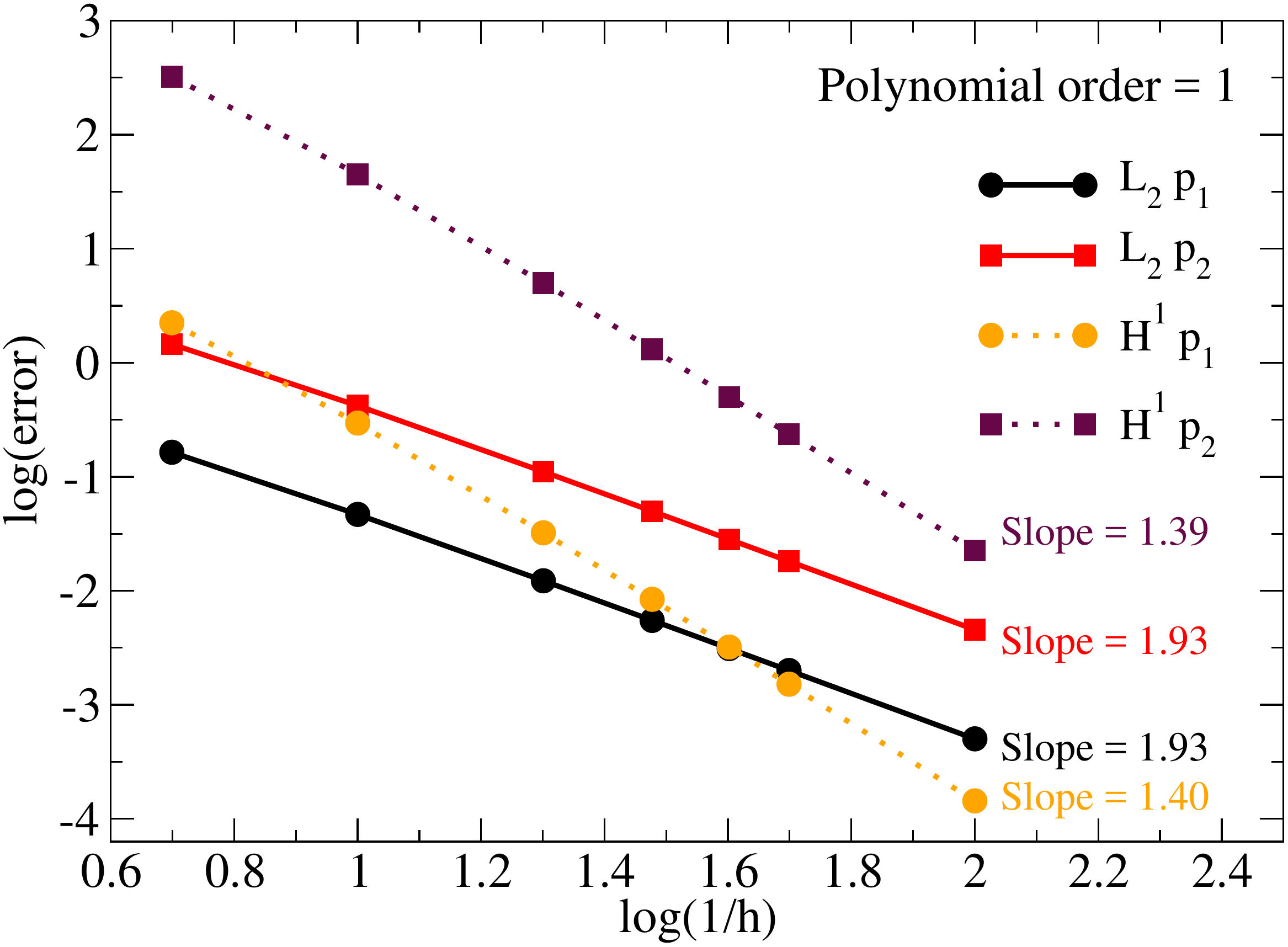}}
    \subfigure{
    \includegraphics[clip,scale=0.38]{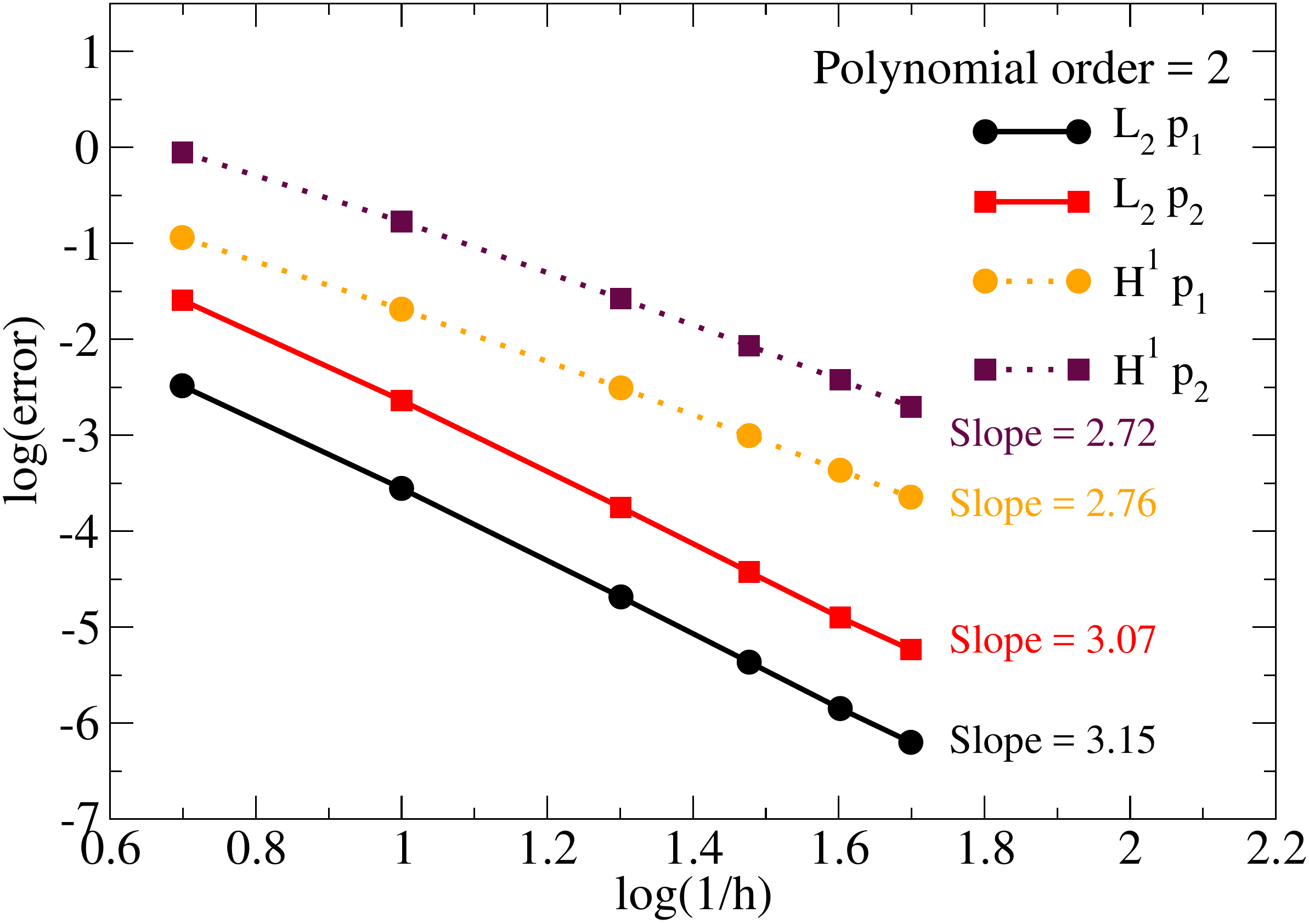}}
  \subfigure{
    \includegraphics[clip,scale=0.38]{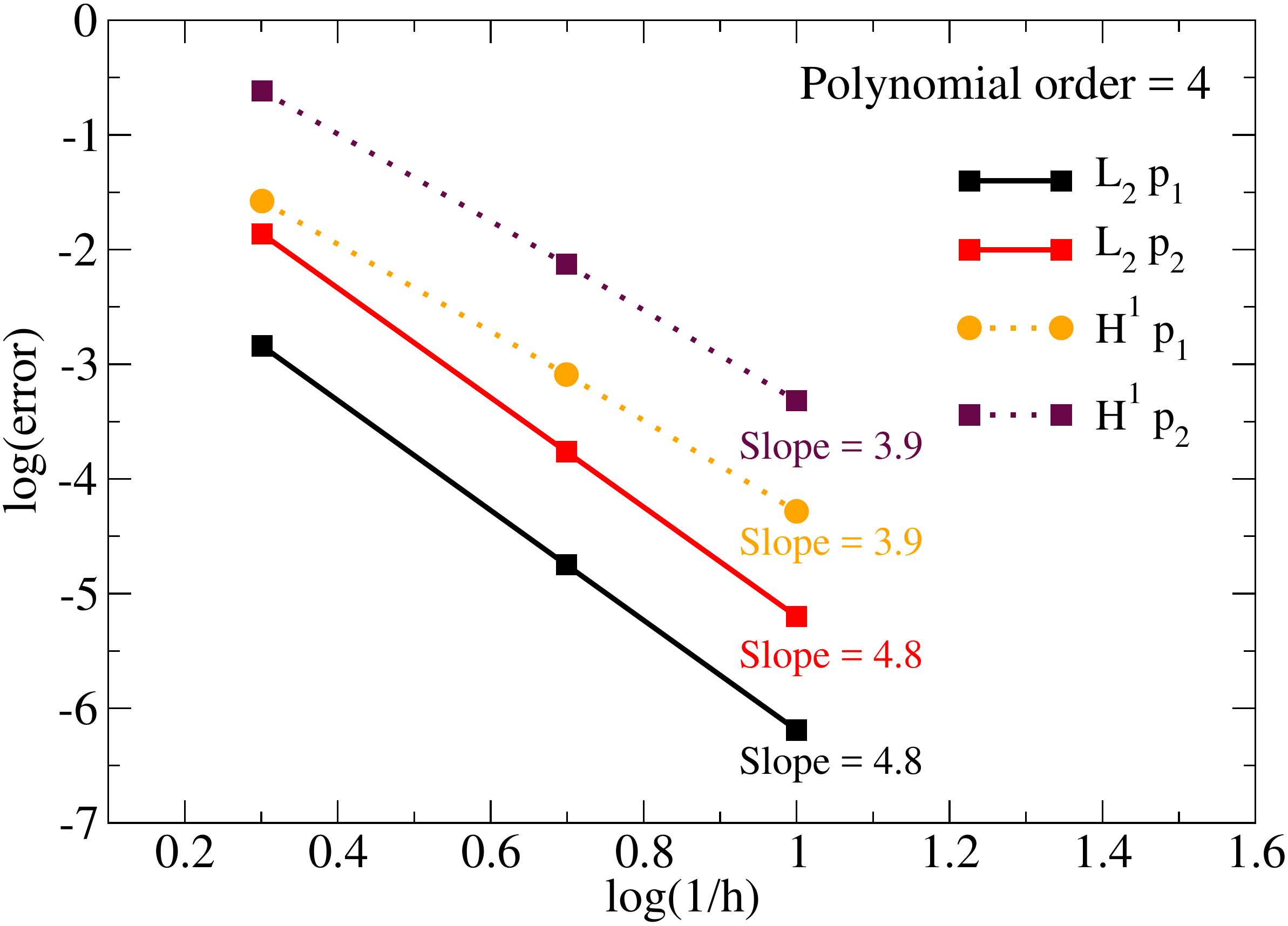}}
  \caption{\textsf{2D numerical convergence analysis:}~This figure shows the numerical convergence under $h$-refinement for various polynomial orders. The rate of convergence is
    polynomial, which is in accordance with the theory.}
    \label{Fig:Dual_Problem_2D_h_refinement}
    \end{figure}
%
\begin{figure}
  \includegraphics[clip,scale=0.37]{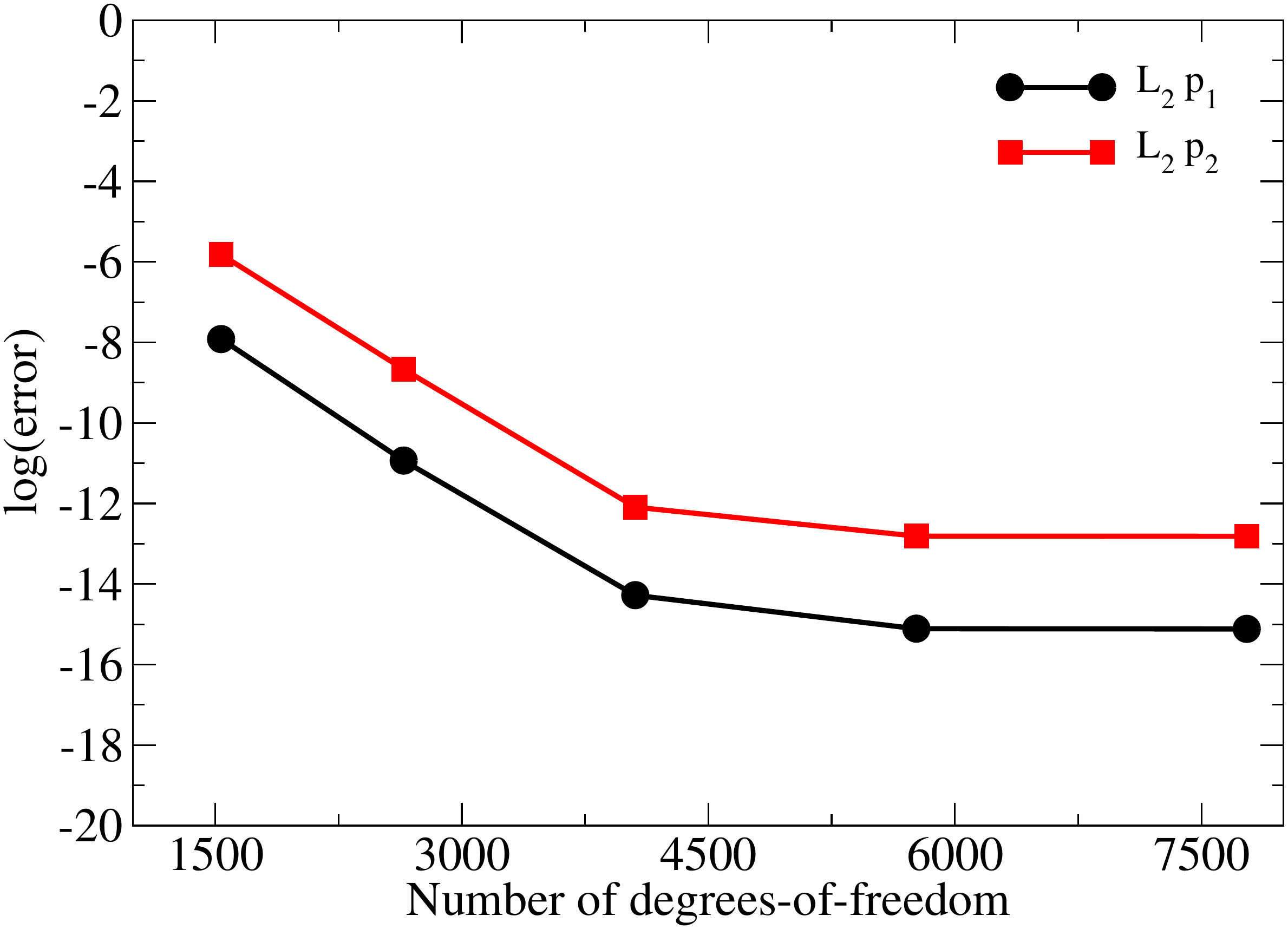}
    \caption{\textsf{2D numerical convergence analysis:}~This figure shows the numerical convergence
    under $p$-refinement for a fixed mesh size ($h = 0.2$).
    The number of degrees-of-freedom corresponds to $p
    = 3$ to $7$. The rate of convergence is exponential,
    which is in accordance with the theory. Note that
    the error flattened out around $10^{-16}$ for larger
    number of degrees-of-freedom. This is expected as
    the machine precision on a 64-bit machine is around
    $10^{-16}$.}
    \label{Fig:Dual_Problem_2D_p_refinement}
\end{figure}

\begin{figure}

\vspace{10mm}
\includegraphics[clip,scale=0.6]{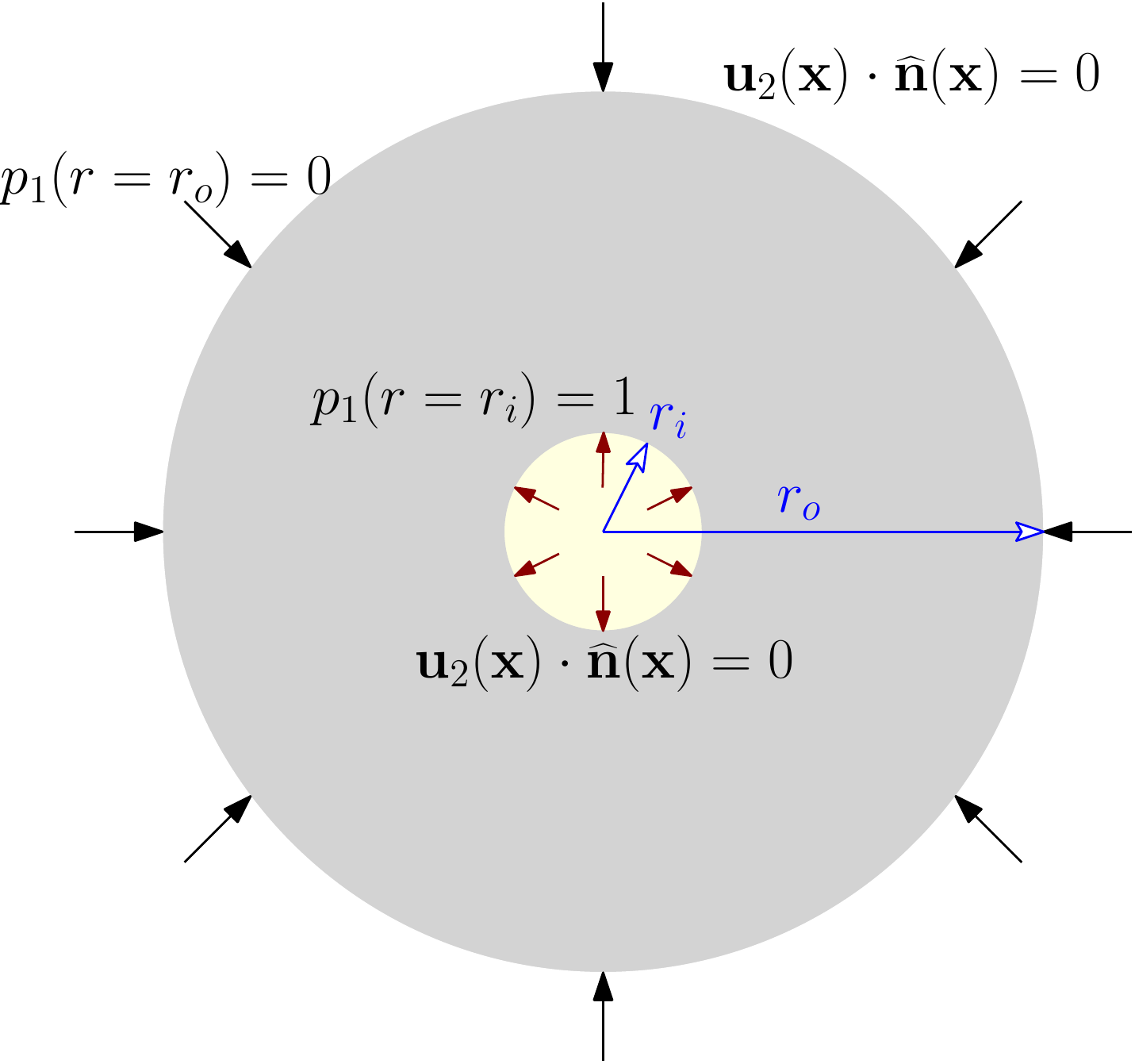}
\caption{\textsf{Two-dimensional candle filter problem:}~This figure provides a pictorial description of the candle filter problem which is used to study weak enforcement of velocity boundary conditions. There is no discharge on the inner and outer surfaces of the micro-pore network. For the macro-pore network, the inner surface is subjected to a pressure of unity, and the outer surface is subjected to zero pressure.}
\label{Fig:2D_problem_candle_filter_discription}
\end{figure}
%
\begin{figure}
  \subfigure[Pressure fields \label{Fig:2D_problem_candle_filter_p}]{
    \includegraphics[scale=0.6,angle=270]{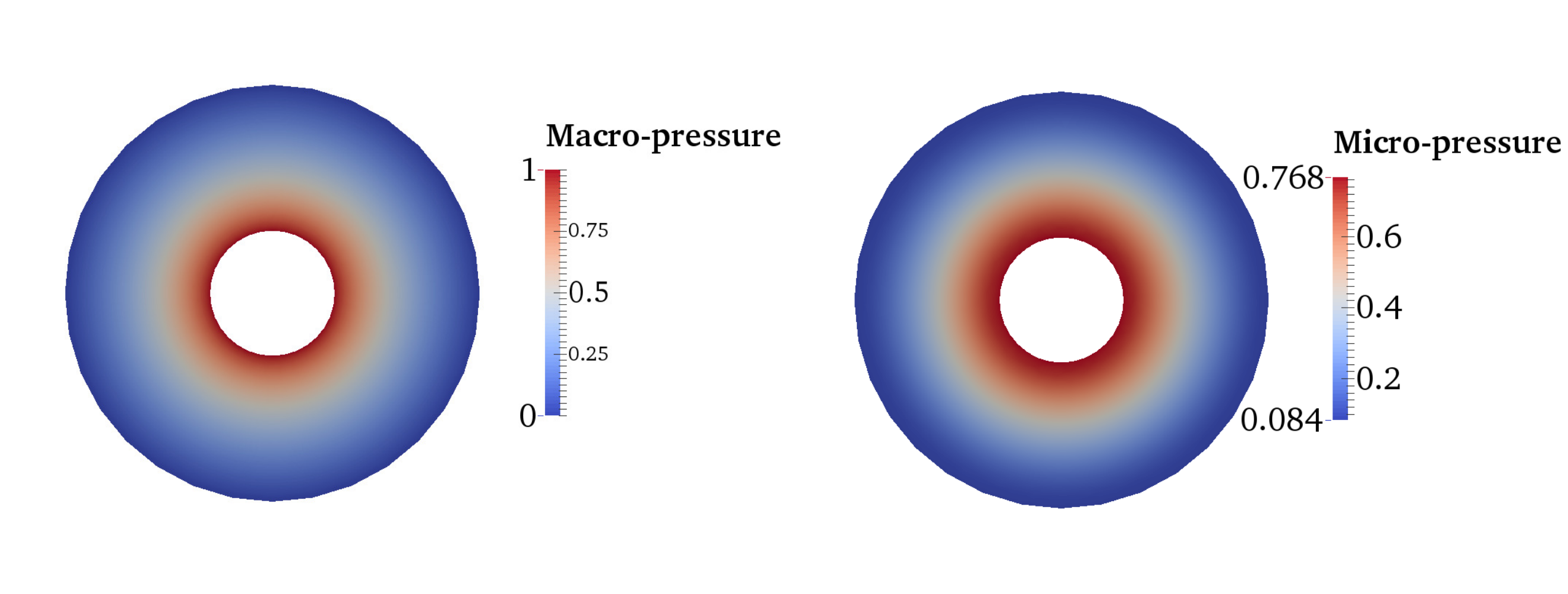}}
  \subfigure[Velocity vector fields \label{Fig:2D_problem_candle_filter_v}]{
    \includegraphics[scale=0.6,angle=270]{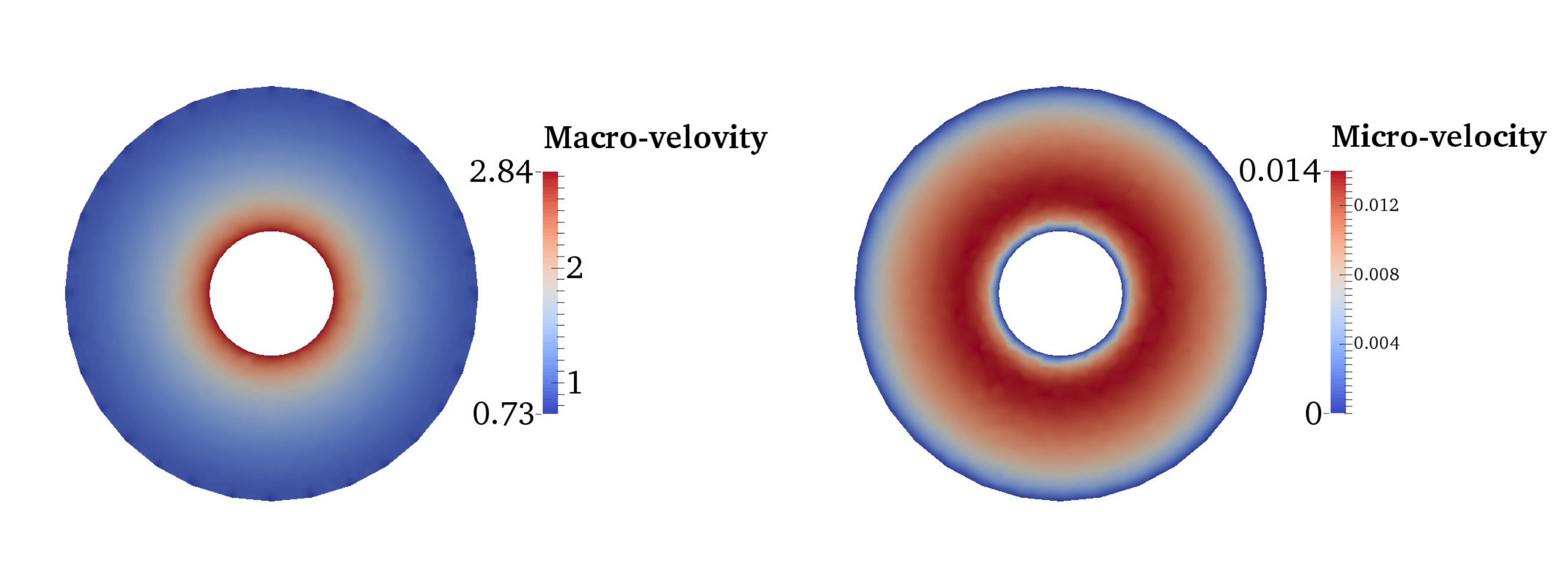}}
\caption{\textsf{Two-dimensional candle filter problem:}~This figure shows the contours of pressures and velocities in macro- and micro-pore networks under the extended framework for weak enforcement of velocity boundary conditions. Although there is no discharge from the micro-pore network on the boundary, there is discharge in the micro-pore network within the domain.}
    \label{Fig:2D_problem_candle_filter}
\end{figure}
%
%
\begin{figure}
  \subfigure[Pressure fields \label{Fig:3D_problem_candle_filter_p}]{
    \includegraphics[scale=0.54,angle=270]{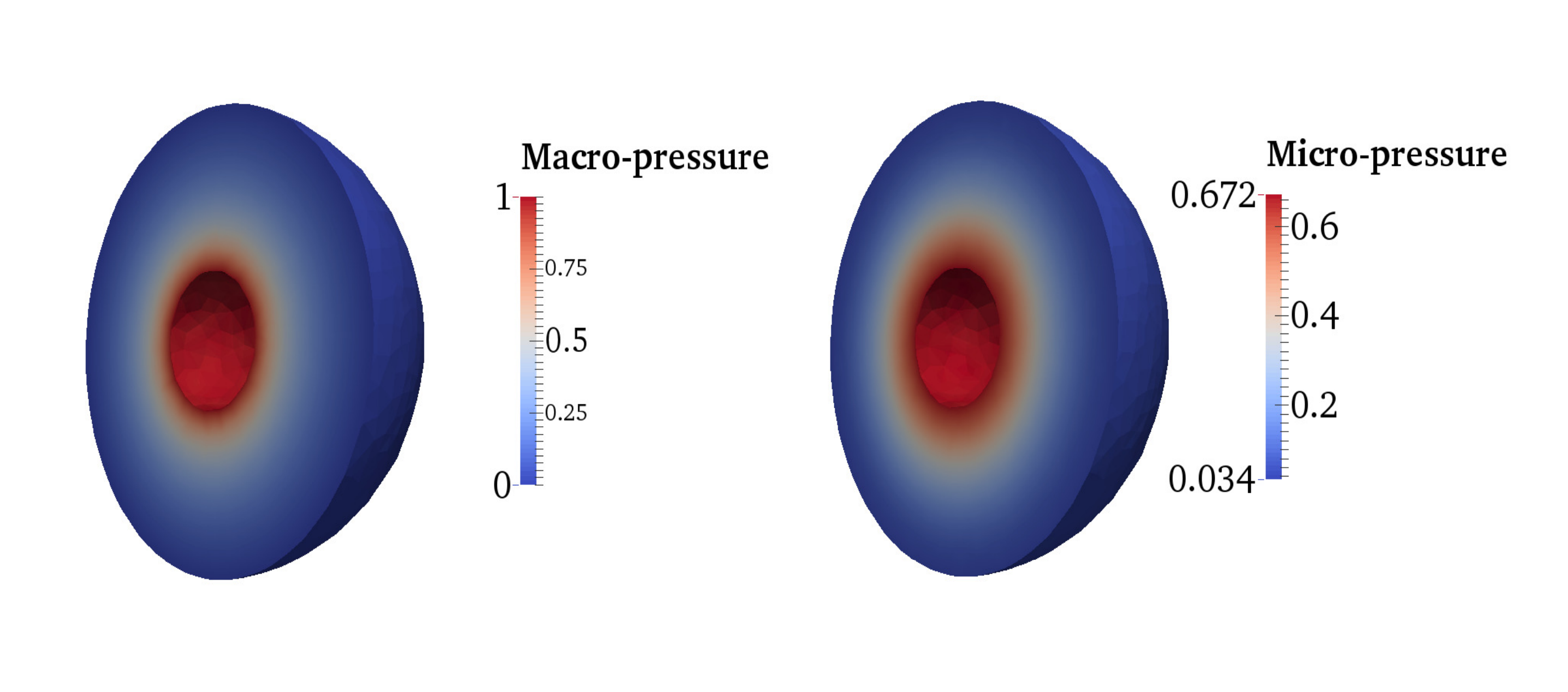}}
  \subfigure[Velocity vector fields \label{Fig:3D_problem_candle_filter_v}]{
    \includegraphics[scale=0.54,angle=270]{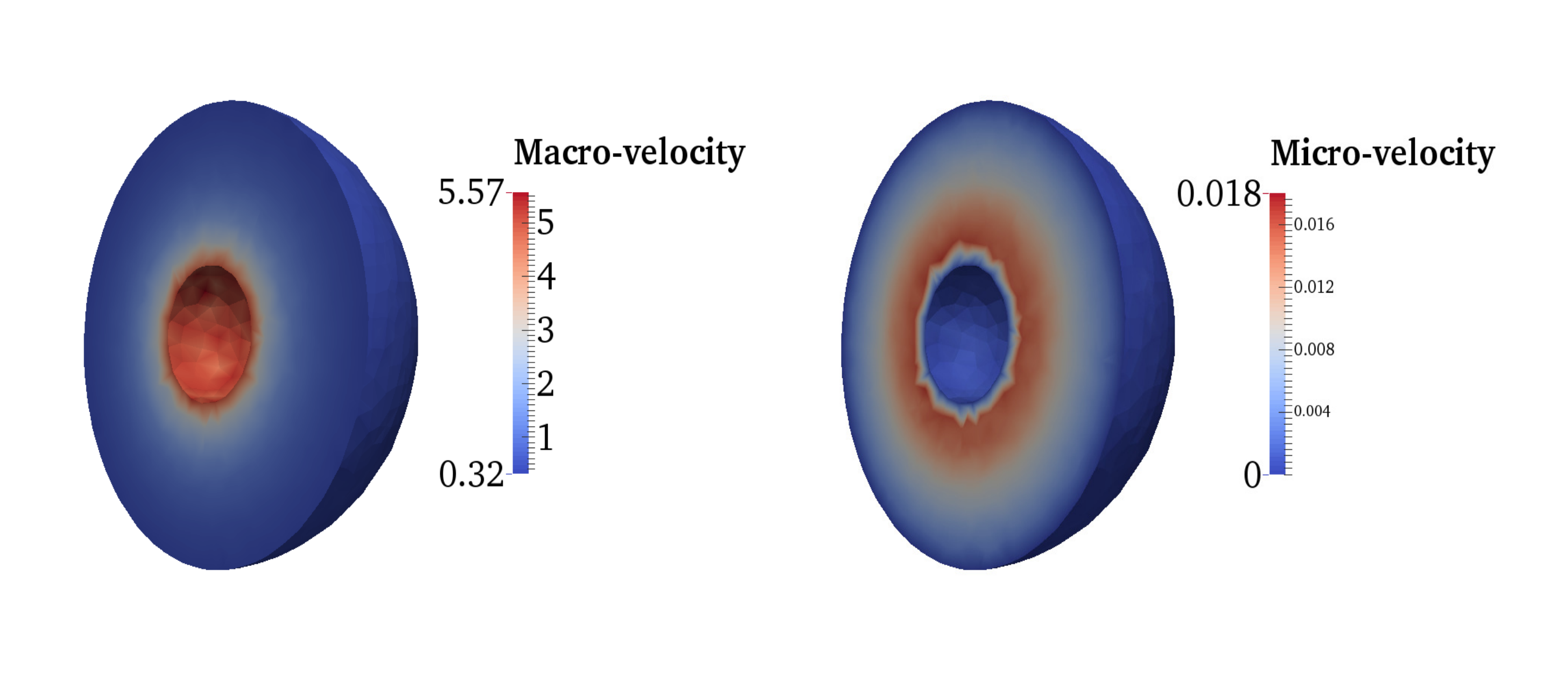}}
\caption{\textsf{Three-dimensional hollow sphere problem:}~This figure shows the contours of pressures and velocities in the two pore-networks under the extended framework for weak enforcement of velocity boundary conditions. On the inner and outer surfaces, pressure is prescribed for the macro-pore network while there is no discharge for the micro-pore network. Although there is no discharge from the micro-pore network on the boundary, there is discharge in the micro-pore network within the domain.}
    \label{Fig:3D_problem_candle_filter}
\end{figure}


\begin{figure}
\includegraphics[clip,scale=0.73]{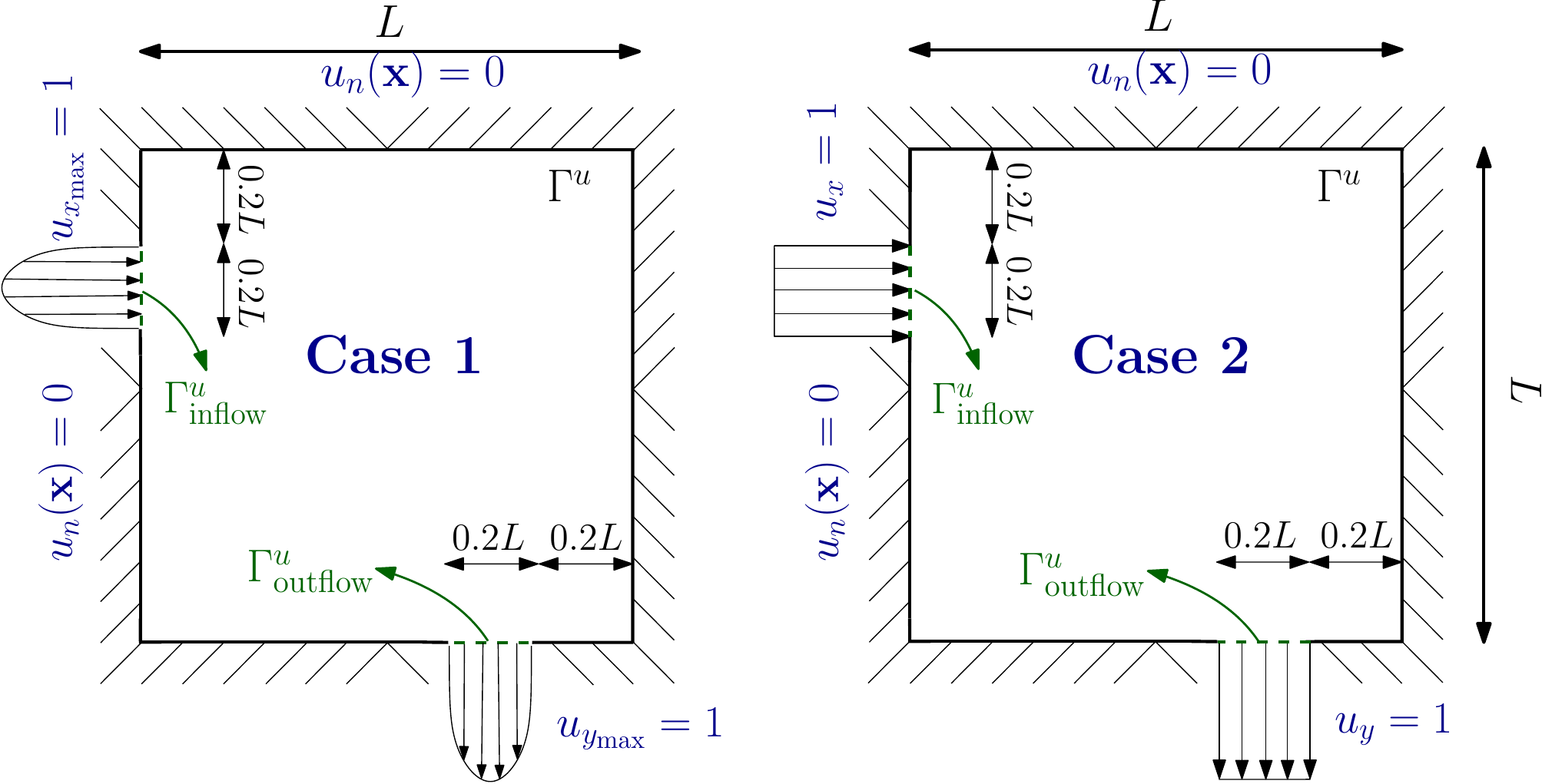}
\caption{\textsf{Pipe bend problem:}~In case 1, for the macro-pore network, an inflow parabolic velocity is enforced on $\Gamma^u_{\mathrm{inflow}}$ while an outflow parabolic velocity is applied on $\Gamma^u_{\mathrm{outflow}}$. In case 2, an inflow constant velocity is enforced on $\Gamma^u_{\mathrm{inflow}}$ while an outflow constant velocity is applied on $\Gamma^u_{\mathrm{outflow}}$ for the macro-pore network. On the other parts of the boundary, normal component of velocity is assumed to be zero.}
\label{Fig:pipe_bend_problem_domain}
\end{figure}
%

\begin{figure}

\vspace{10mm}
\subfigure[Dissipation\label{Fig:pipe_bend_problem_Dissipation}]{
\includegraphics[clip,scale=0.28]{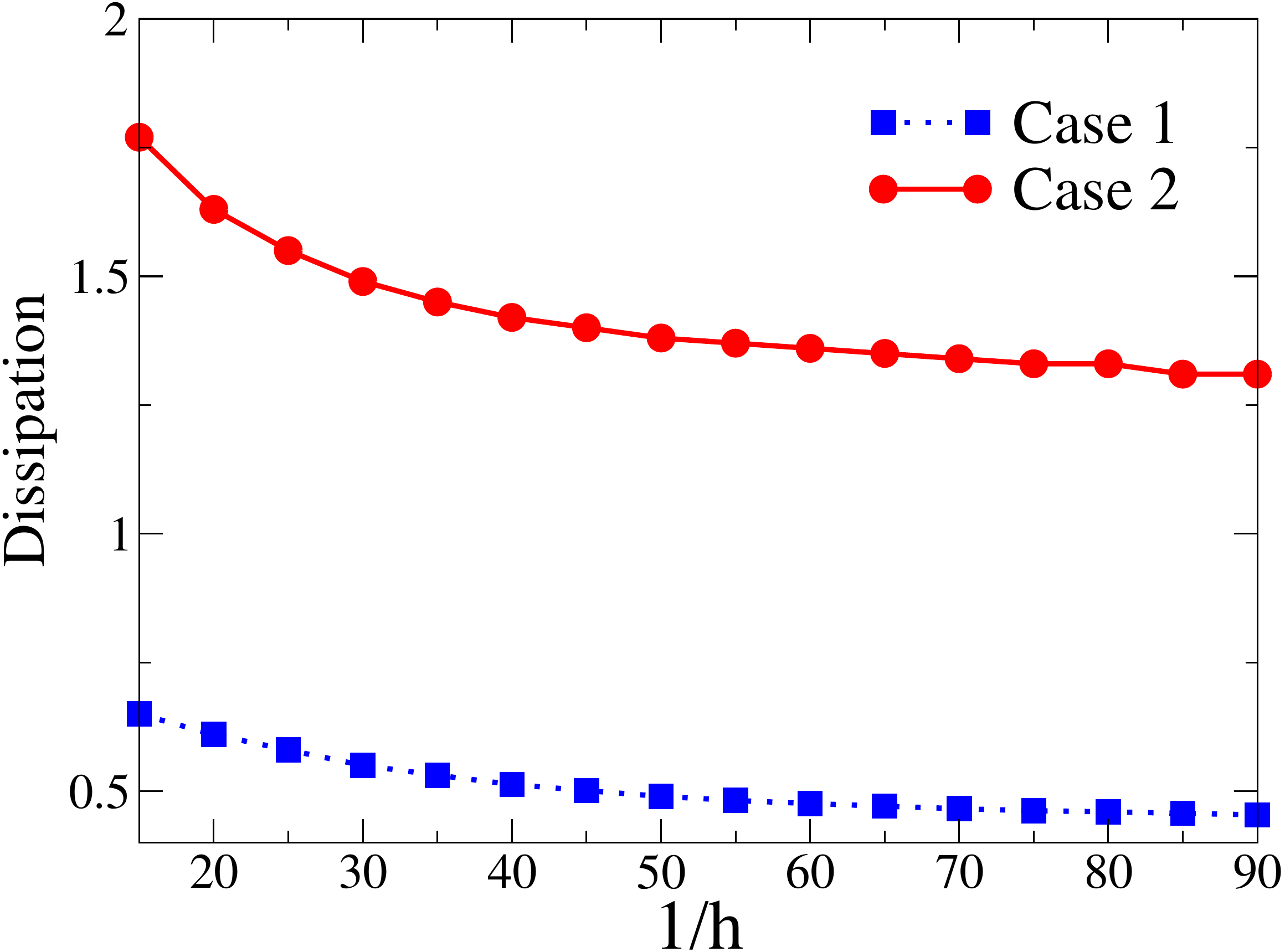}}
\hspace*{5mm}
\subfigure[Error in reciprocal relation\label{Fig:Pipebend_problem_reciprocal_error}]{
\includegraphics[clip,scale=0.28]{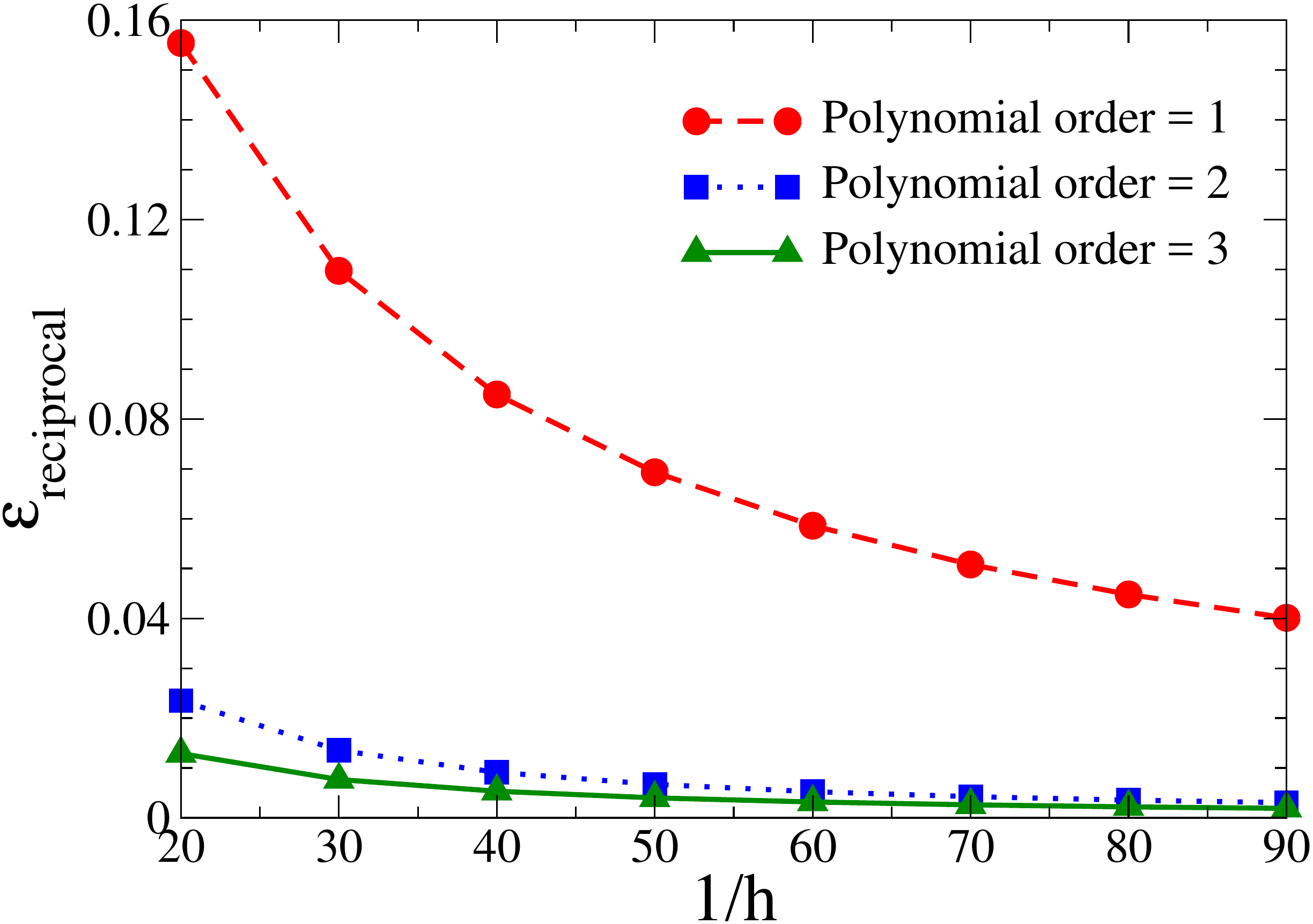}}
\caption{\textsf{Pipe bend problem}:~The left figure shows the variation of dissipation with mesh refinement for both cases shown in figure \ref{Fig:pipe_bend_problem_domain}. As can be seen, the dissipation decreases monotonically with mesh refinement which is in accordance with the theory for this problem. The right figure, shows the variation of $\varepsilon_{\mathrm{reciprocal}}$ with mesh refinement using the two cases for different orders of interpolation. The numerical error in the reciprocal relation decreases monotonically with mesh refinement for this test problem which shows the monotonic convergence of numerical solutions.}
\end{figure}


\begin{figure}

\vspace{10mm}
\includegraphics[clip,scale=0.9]{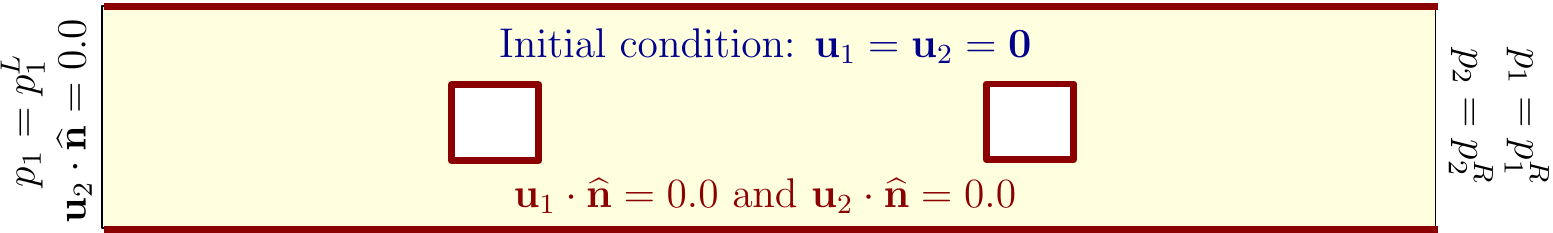}
\caption{\textsf{Transient 2D flow problem}:~This figure shows the computational domain, initial and boundary conditions for the transient
problem.}
\label{Fig:Flow_Transport}
\end{figure}

\begin{figure}
\vspace{1cm}
\includegraphics[clip,scale=0.8]{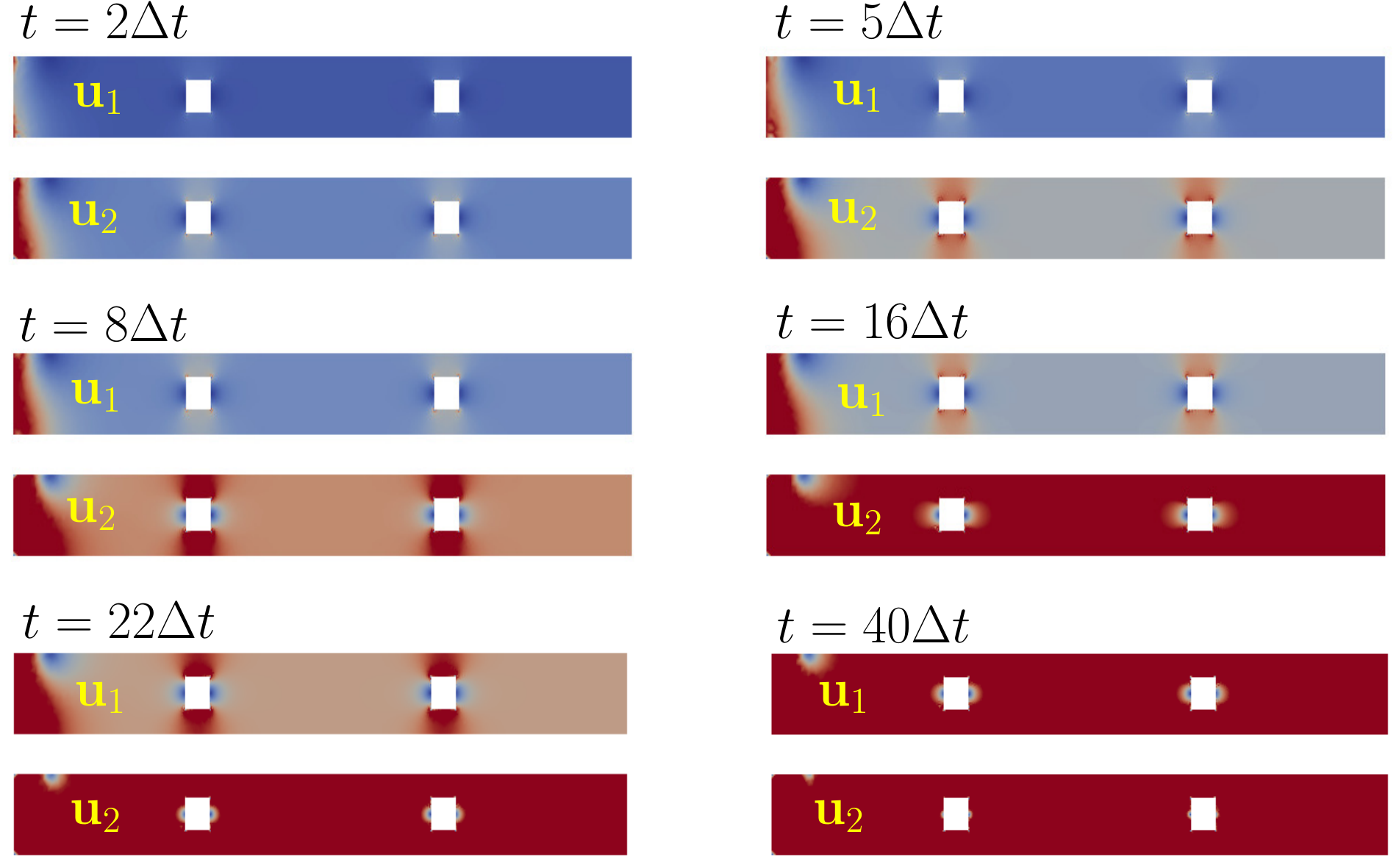}
\caption{\textsf{Transient 2D flow problem}:  This figure shows a comparison between macro- and micro-velocities at different time steps. As can be seen, the rate of decay of the solution in the macro-pore network is slower than that of the micro-pore network which is due to the higher permeability of the macro-pore network. Hence, the micro-velocity reaches the steady state faster than the macro-velocity.}
\label{Fig:V1_vs_V2_Transport}
\end{figure}

\begin{figure}

\vspace{10mm}
\includegraphics[clip,scale=0.6]{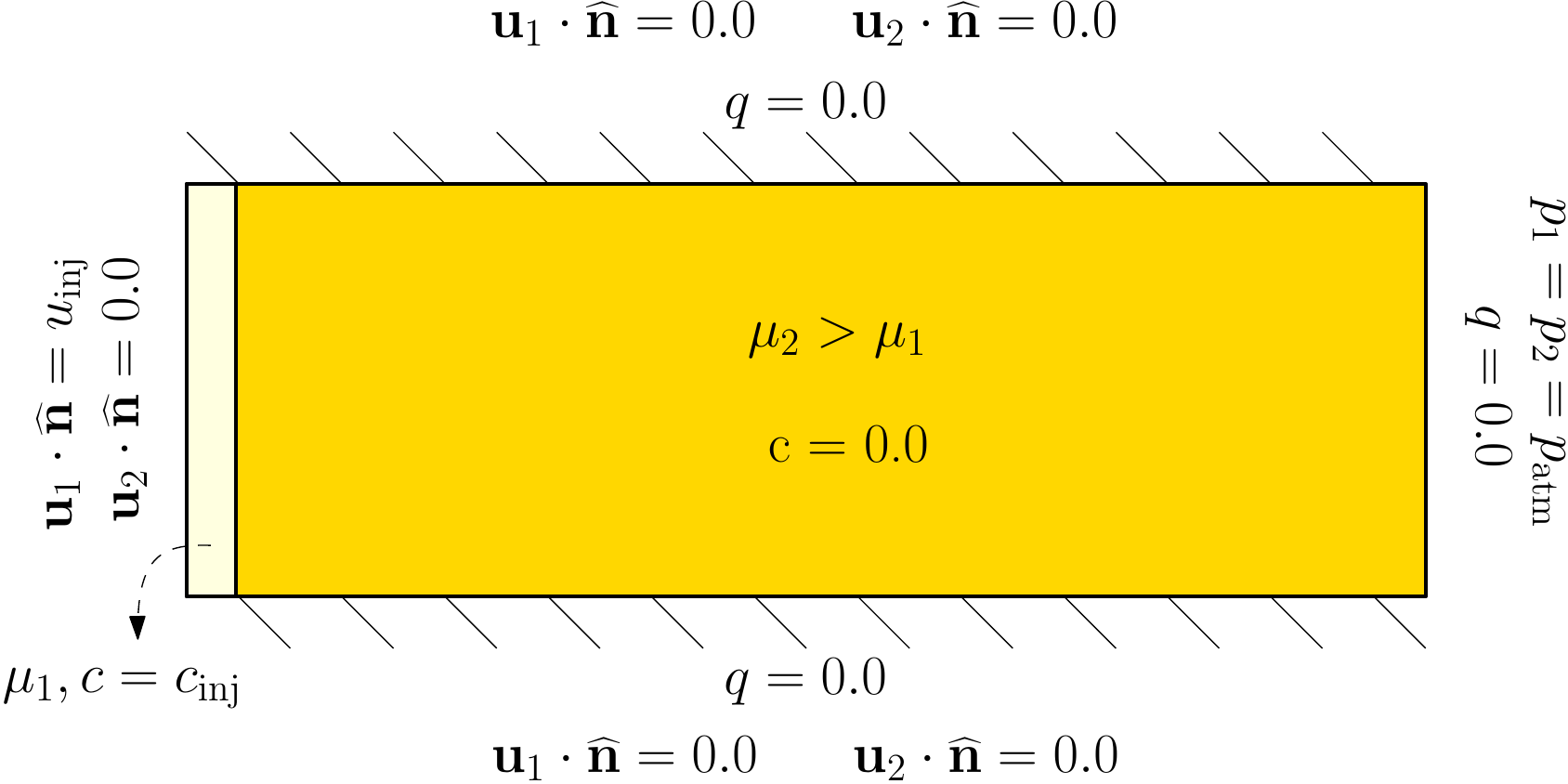}
\caption{\textsf{Hele-Shaw cell:}~This figure shows the pictorial description of the coupled flow-transport problem including initial and boundary conditions.}
\label{Fig:Hele_shaw_Schm}
\end{figure}

\begin{figure}
\vspace{1cm}
\includegraphics[clip,scale=0.8]{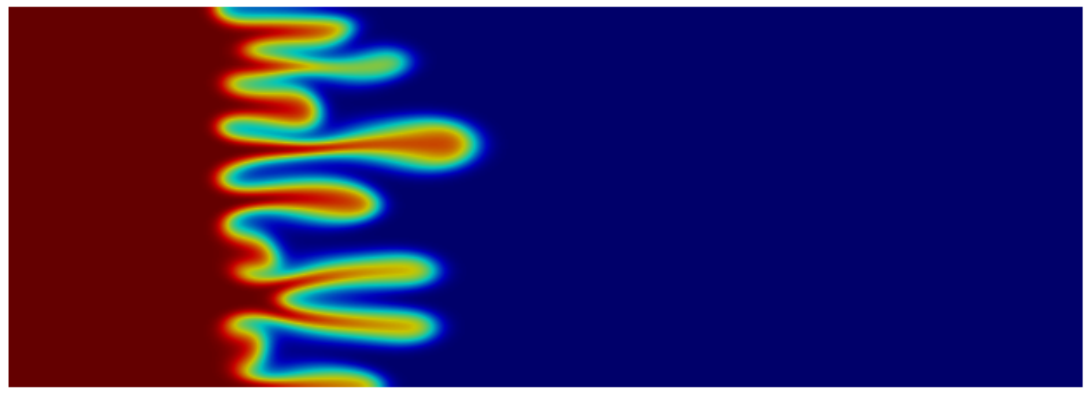}
\caption{\textsf{Coupled flow-transport problem}:~This figure shows that Saffman-Taylor-type physical instabilities can also occur in a porous domain exhibiting double porosity/permeability. As can be seen, the proposed stabilized formulation is capable of eliminating the spurious numerical instabilities without suppressing the underlying physical instability.}
\label{Fig:Hele_Shaw_Concentration_profiles}
\end{figure}